\documentclass[11pt,a4paper]{article}

\usepackage[utf8]{inputenc} 
\usepackage[T1]{fontenc}    %
\usepackage[english]{babel} 
\usepackage[final]{microtype} 

\usepackage{amssymb,amsmath,mathtools} 
\usepackage{amsthm} 

\usepackage{xcolor}
\usepackage{bbm}
\usepackage{braket}

\usepackage[hmargin=0.12\paperwidth,vmargin=0.16\paperwidth,bindingoffset=0cm]%
{geometry} 

\numberwithin{equation}{section} 

\theoremstyle{plain}
  
  \newtheorem{prop}{Proposition}
  
  \newtheorem{cor}{Corollary}
\theoremstyle{definition}
  
  \newtheorem{remark}{Remark}
  
  \numberwithin{prop}{section}
   \numberwithin{cor}{section}
   \numberwithin{remark}{section}

\title{\Large\bfseries Dualities  in random matrix theory}%
   
\author{Peter J. Forrester}
\date{}

\begin{document}

\maketitle

School of Mathematics and Statistics,  The University of Melbourne,
Victoria 3010, Australia. \: \: Email: {\tt pjforr@unimelb.edu.au}; \\

\bigskip

\begin{abstract}
\noindent
 Duality identities in random matrix theory for products and powers of characteristic polynomials, and for moments,
 are reviewed. The structure of a typical duality identity for the average of a  positive integer power $k$ of the characteristic polynomial
 for particular ensemble of $N \times N$ matrices is that it is expressed as the average of the power $N$ of the  characteristic polynomial
 of some other ensemble of random matrices, now of size $k \times k$. With only a few exceptions, such dualities involve (the $\beta$
 generalised) classical   Gaussian, Laguerre and Jacobi ensembles Hermitian ensembles, the circular Jacobi ensemble, or the various
 non-Hermitian ensembles relating to Ginibre random matrices. In the case of unitary symmetry in the Hermitian case, they can be studied
 using the determinantal structure. The $\beta$ generalised case requires the use of Jack polynomial theory, and in particular Jack polynomial
based  hypergeometric functions. Applications to the computation of the scaling limit of various $\beta$ ensemble correlation and distribution functions are also reviewed.
 The non-Hermitian case relies on the particular cases of Jack polynomials corresponding to zonal polynomials,
 and their integration properties when their arguments are eigenvalues of certain matrices. The main tool to study dualities for moments of
 the spectral density, and generalisations, is the loop equations.
\end{abstract}

\tableofcontents

\vspace{3em}

\section{Introduction}

As a prelude to reviewing duality relations in random matrix theory, a few remarks about dualities more generally are in order.
In the essay ``Duality in mathematics and physics'', Atiyah \cite{At08} posits ``Fundamentally, duality gives two different points of view
of looking at the same object''. In internet forums, an obvious key word search reveals the descriptions ``In laymen terms, it's when you consider an
opposite concept in such a way that some properties have `flipped' analogous properties''; ``The common idea is that there are two things
which basically are just two sides of the same coin''; and ``When applied to specific examples, there is usually a precise definition for just that
context'', among others.

Indeed dualities in random matrix theory involve `flipping'. The flipping may be of two parameters (e.g.~the order of a moment, and the
size of the matrix), or of two ensembles (e.g.~orthogonal symmetry, and symplectic symmetry), these being  the primary classes of
dualities to be discussed below. It is furthermore the case that these contexts offer a precise mathematical definition. Introductory examples
can be given.

In relation to a duality relation involving flipping the order of a moment and the
size of the matrix, one can consider the first moment (i.e.~average value) of the characteristic polynomial
$\det ( x \mathbb I_N - H)$ for the matrix $H$ an Hermitian Wigner matrix.
The construction of the latter is according to
 $H = {1 \over 2} (X + X^\dagger)$, where the $N \times N$ random matrix $X$ has
all elements independently distributed with zero mean and variance $\sigma^2$ (for complex elements with zero mean,
$\sigma^2 := \langle |x_{ij} |^2 \rangle$). With the ensemble of such matrices denoted by $\mathcal H_N$, one has
\cite[equivalent to Prop.~11]{FG06}
\begin{equation}\label{1.0}
\langle \det ( x \mathbb I_N - H) \rangle_{H \in \mathcal H_N} = i^{-N}
\langle  ( i x  - h)^N \rangle_{h \in \mathcal N(0,\sigma^2/2)},
 \end{equation}
 where $ \mathcal N(0,\sigma^2/2)$ denotes the zero mean normal distribution with variance $\sigma^2/2$. 
 Comparing both sides, one sees the interchange (flipping) of the pair $(N,n)$, where
on the left hand side $N$ is the matrix size, and $n=1$ is the power of the characteristic polynomial, while
on the right hand side the meaning of these parameters is reversed. 

We again turn to moments for our introductory example illustrating a random matrix duality flipping two ensembles.
The moments of interest here are those associated with the spectral density for finite $N$, which itself after normalisation
can be considered as a probability density function. As an ensemble average, this is equivalent to considering the
expected value of the power sum $\sum_{j=1}^N \lambda_j^k$ ($k=1,2,\dots$), where $\{\lambda_j\}$ are the eigenvalues.
For a given $k$, we will do this for the two distinct ensembles of the Gaussian orthogonal ensemble (GOE) of real symmetric
matrices, and the Gaussian symplectic ensemble (GSE) of quaternion self dual Hermitian matrices. Only the even moments are
non-zero, with the $2k$-th moment in fact a polynomial in $N$ of degree $k+1$, which furthermore vanishes at $N=0$.
 Before listing some of these polynomials, we
recall that matrices from the GOE and GSE are specific classes of Hermitian Wigner matrices, where the elements of the  matrix
$X$ introduced above are zero mean Gaussians. For the GOE, these elements are all real. For the GSE, the elements are 
themselves $2 \times 2$ complex matrices
representing quaternions (see e.g.~\cite[\S 1.3.2]{Fo10}), and the eigenvalues here are  two fold degenerate. The probability
measure on the space of real symmetric matrices for
the GOE, being proportional to $e^{- {\rm Tr} \, H^2/2}$,
 is invariant under the similarity transformation, real symmetric matrices automorphism,
 $H \mapsto O^THO$ for any real orthogonal matrix $O$ of the same
size as $H$. This is the reason for the `orthogonal' in the naming of the ensemble. Similar considerations
apply to the naming of the GSE, where the relevant similarity transformation, quaternion self dual Hermitian matrices
automorphism, is $H \mapsto S^\dagger H S$, where $S$ is unitary symplectic matrix.

A classical result in random matrix theory is that the joint eigenvalue probability density function (PDF) of the
GOE and GSE is proportional to (see e.g.~\cite[Prop.~1.3.4]{Fo10})
\begin{equation}\label{1.1}
\prod_{l=1}^N e^{- \beta \lambda_l^2/2} \prod_{1 \le j < k \le N} | \lambda_k - \lambda_j|^\beta,
 \end{equation}
for $\beta = 1$ (GOE) and $\beta = 4$ (GSE). The parameter $\beta$ is referred to as the Dyson index.
The most efficient computational scheme for the moments $\{m_{2k} \}$ is via the
use of certain fourth order linear difference equations found by Ledoux \cite{Le09}, which furthermore only require knowledge
of the initial conditions $m_0 = N$ and $m_2^{\rm GOE} =  N^2 + N$, $m_2^{\rm GSE}= N^2 - N/2$ for their unique solution. For the particular low order cases
$k=3, 4$, one then
finds
{\small \begin{align}
& m_6^{\rm GOE} = 5 N^4 + 22 N^3 + 52 N^2  + 41 N, \: \:    m_8^{\rm GOE} =  14 N^5 + 93  N^4 + 374 N^3 + 690 N^2 + 509 N   \\
& 8m_6^{\rm GSE} = 40 N^4 - 88 N^3 + 104 N^2 -  41 N, \: \:
16 m_8^{\rm GSE} =  224 N^5 - 744 N^4 + 1496 N^3 - 1380 N^2 + 509N.
\end{align}}
Inspection reveals the functional (duality)  relation 
\begin{equation}\label{1.2}
m_{2k}^{\rm GOE} = (-1)^{k+1} 2^{k+1} m_{2k}^{\rm GSE} \Big |_{N \mapsto - N/2}, \quad (k=3,4)
 \end{equation}
 which is in fact valid for each $k \in \mathbb Z_{\ge 0}$ (\cite[Eq.~(5.3)]{MW03}, special case of \cite[Th.~2.8]{DE05}).
 Hence there is a flipping of ensembles on different sides of this equation, as well as a particular algebraic mapping of
 $N$ (we say algebraic, as there is no matrix construction once $N$ is so mapped). We remark that the GOE and GSE can be
 linked in other ways. For example, integrating over all the odd labelled eigenvalues in the GOE of size
 $2N+1$ gives the GSE of size $N$; for a recent review on results in random matrix theory of this sought,
 see \cite{Fo24}.
 
 The class of dualities involving characteristic polynomials will be considered first. As this class is large, the
 subject matter under this heading is to be further subdivided. This we do by having separate sections for when the dualities
 involve classical random matrix ensembles with $\beta=2$ (Section \ref{S2}), then the general $\beta$ case (Section \ref{S3}).
 Consequences for the large $N$ limit are considered in Section \ref{S4}. Dualities relating to
 characteristic polynomials for non-Hermitian ensembles are the topic of Section \ref{S5}.
Section \ref{S7} gives an account of dualities  relating to (mixed) moments for a large class of $\beta$ ensembles.

Not considered in this review (due to the present author's lack of expertise) are random matrix group integral dualities
coming from the field theory method of the colour-flavour transformation \cite{Zi96,FK07}, or from the related
mathematical theory of Howe pairs \cite{CFZ05,Zi21}.
 
 \section{Characteristic polynomial dualities for classical random matrix ensembles with $\beta=2$}\label{S2}
 \subsection{Invariant ensemble definitions}
  Consider the generalisation of (\ref{1.1}), defined as the eigenvalue PDF on the real line proportional to
 \begin{equation}\label{2.1}
\prod_{l=1}^N w(\lambda_l) \prod_{1 \le j < k \le N} | \lambda_k - \lambda_j|^\beta.
 \end{equation}
Viewed as relating to an ensemble of Hermitian random matrices, (\ref{2.1}) is denoted ME${}_{\beta,N}[w]$.
The function $w(x)$ is referred to as the weight.
As is consistent with the situation for (\ref{1.1}), which corresponds to choosing $w(\lambda) = e^{-\beta \lambda^2 / 2}$,
the cases $\beta = 1,4$ relate to orthogonal (symplectic) invariant 
ensembles of real (quaternion) valued Hermitian random matrices. Similarly complex valued
 Hermitian  random matrices permit the automorphism $H \mapsto U^\dagger H U$ for $U$ a member of the unitary group. If the
joint element PDF is invariant under this mapping, an eigenvalue PDF of the form (\ref{2.1}) results with $\beta = 2$
(specifically, the element PDF reduces to the functional form $\prod_{l=1}^N w(\lambda_l)$, while 
$\prod_{1 \le j < k \le N} | \lambda_k - \lambda_j|^\beta$ is identified as the eigenvalue factor of the Jacobian
in the change of variables from the matrix elements to the eigenvalues and eigenvectors; see \cite[Prop.~1.3.4]{Fo10}).

It turns out that the dualities of certain classes of Hermitian random matrices to be considered will also involve unitary random
matrices. The analogue of (\ref{2.1}) in this setting is
 \begin{equation}\label{2.1X}
\prod_{l=1}^N w(\theta_l) \prod_{1 \le j < k \le N} | e^{2 \pi i\theta_k} - e^{2 \pi i \theta_j}|^\beta, \quad - 1/2 < \theta_l \le 1/2,
 \end{equation}
 to be denoted CE${}_{\beta,N}[w]$ (the ``C'' here stands for circular). The naming circular
 $\beta$ ensemble is used for the case $w=1$, while the naming circular orthogonal,
 circular unitary, and circular symplectic ensemble (notation COE, CUE and CSE respectively)
 is used for the cases $\beta=1,2$ and 4 of the circular $\beta$ ensemble.

  A variation of (\ref{2.1})  is the eigenvalue PDF proportional to
  \begin{equation}\label{3.1}
\prod_{l=1}^N w(\lambda_l) \prod_{1 \le j < k \le N} | \lambda_k^2 - \lambda_j^2|^\beta.
 \end{equation}
Typically this results from matrices with the property that the eigenvalues come
in $\pm$ pairs, which is termed a chiral symmetry. Thus the naming chME${}_{\beta,N}(w)$.
An example of matrices with a chiral symmetry is the structure
  \begin{equation}\label{3.1a}
Y = \begin{bmatrix} 0_{n \times n} & Z \\
Z^\dagger & 0_{N \times N}   \end{bmatrix}
 \end{equation}  
 for $Z$ an $n \times N$ ($n \ge N$) matrix --- in general there are $N$ eigenvalues in $\pm$ pairs, and $n-N$ zero
 eigenvalues. In the case that $Z$ has independent standard complex entries, and thus being a so called (complex)
 rectangular Ginibre matrix \cite{BF24}, 
 the positive eigenvalues of $Y$ have 
 PDF    chME${}_{2,N}(w)$, $w(x) = x^{2(n - N) + 1}e^{-x^2}$; see
 e.g.~\cite[Prop.~3.1.3]{Fo10}.

 \subsection{Products of characteristic polynomials and the GUE}
In light of (\ref{1.0}), an immediate question to ask is if there is an analogous duality for
 \begin{equation}\label{2.2}
\Big  \langle  \prod_{l=1}^p \det ( x_l \mathbb I_N - H) \Big  \rangle_{ \mathcal H_N} 
 \end{equation}
 for $p > 1$. Here we may consider the cases of $ \mathcal H_N$ a class of Wigner matrices, or an invariant ensemble
 characterised by the eigenvalue PDF (\ref{2.1}) (note that choosing the matrix $X$ --- recall the text above (\ref{1.0}) ---
 to have independent standard Gaussian elements, which are real, complex or quaternion for $\beta = 1,2$ and 4
 respectively gives rise to random matrix ensembles, the GOE,  GUE, GSE, which are both of the Wigner and invariant
 class). The simplest case to investigate is  for $ \mathcal H_N$ an invariant ensemble with complex elements, and
 thus corresponding to ME${}_{2,N}(w)$. Introduce the family of monic orthogonal polynomials
 $\{ p_j(x) \}_{j=0,1,\dots}$ by the requirement that $\int_I w(x) p_j(x) p_k(x) \, dx = h_j \delta_{j,k}$ (here $I$ is the support of $w(x)$ and $h_j > 0$
 is the normalisation). Then it is a standard result  \cite{BH00}, \cite{FS03e}, \cite{BDS03}, \cite{Fo10} that
  \begin{equation}\label{2.3}
\Big  \langle   \prod_{l=1}^p \det ( x_l \mathbb I_N - H)    \Big  \rangle_{ {\rm ME}_{2,N}[w]} = 
{1 \over \prod_{1 \le j < k \le p} (x_k - x_j)} \det [ p_{N+j-1}(x_k) ]_{j,k=1,\dots,p}.
 \end{equation}
 
 In the case $p=1$ (\ref{2.3}) gives
  \begin{equation}\label{2.3q}
\Big  \langle    \det ( x \mathbb I_N - H)    \Big  \rangle_{ {\rm ME}_{2,N}[w]} = p_N(x),
 \end{equation}
 which in fact is a classical formula in orthogonal polynomial theory due to Heine \cite{Sz75}.
 In the case $w(x) = e^{-x^2}$ corresponding to the GUE, one has
$p_N(x) = 2^{-N} H_N(x)$, where $H_N(x)$ denotes the Hermite
polynomial of degree $N$.  Use of the integral form of the Hermite polynomials
 \begin{equation}\label{2.3c}
   H_{N}(x) =  {(2i)^{N} \over \sqrt{\pi}} e^{x^2} \int_{-\infty}^\infty e^{-t^2} t^{N} e^{-2 i x t} \, dt
 \end{equation}
 gives agreement with    (\ref{1.0}).

 We would like to identify the  RHS of (\ref{2.3}) with a random matrix average over
 an ensemble of $p \times p$ matrices for general $p \ge 1$. 
 It turns out that this is possible for both the GUE and LUE, although in the latter case the
 ensemble that appears is hard to anticipate (see \S \ref{S3.3}).
 For the GUE weight 
 $w(x) = e^{-x^2}$, the RHS consists of another GUE average.

 \begin{prop}\label{P2.1} (\cite{BH01})
Let $X = {\rm diag} \, (x_1,\dots,x_p)$. We have
    \begin{equation}\label{2.X1}
    \Big  \langle   \prod_{l=1}^p \det ( x_l \mathbb I_N - H)     \Big  \rangle_{ {\rm GUE}_{N}} =
   i^{-pN} \Big  \langle   \det ( i X - H)^N \Big  \rangle_{ {\rm GUE}_{p}}.
   \end{equation} 
   In particular 
    \begin{equation}\label{2.X1a}
    \Big  \langle    \det ( x \mathbb I_N - H)^p     \Big  \rangle_{ {\rm GUE}_{N}} =
   i^{-pN} \Big  \langle   \det ( i  x \mathbb I_p - H)^N \Big  \rangle_{ {\rm GUE}_{p}}.
   \end{equation}     
    \end{prop}
 
 \begin{proof}
 We know that for the Gaussian weight,  $p_k(x)=2^{-k} H_k(x)$.
  The task is to show that the RHS of (\ref{2.X1}) evaluates to the determinant formula in (\ref{2.3}) with this
 substitution. We begin by noting that, up to proportionality, the average on the RHS can be rewritten as
   \begin{equation}\label{2.X2}
 \int_{H \in \mathcal H_p} (\det H)^N e^{- {\rm Tr} \, (H + i X)^2} \, (dH),
  \end{equation}  
 where $ \mathcal H_p$ is the space of $p \times p$ Hermitian matrices and $(dH)$ denotes the product of all the
 independent differentials (real and imaginary parts) of $H$. Changing variables according to the
 diagonalisation formula $H = U \Lambda U^\dagger$, where $\Lambda = {\rm diag} (\lambda_1,\dots,\lambda_p)$,
 shows that (\ref{2.X2}) is proportional to 
    \begin{equation}\label{2.X3}
    e^{\sum_{j=1}^p x_j^2} \bigg \langle \prod_{l=1}^p \lambda_l^N \int e^{-2 i {\rm Tr} \, U^\dagger \Lambda U X} \,
    d^{\rm H}U \bigg \rangle_{\Lambda \in {\rm GUE}_p},
  \end{equation}      
 where $d^{\rm H}U$ denotes the (normalised) Haar measure on the unitary group $U(p)$. This group integral is
 evaluated according to the well known HCIZ formula --- see e.g.~\cite[Prop.~11.6.1]{Fo10} and
 \cite{No24} for recent developments --- showing that (\ref{2.X3})
 is proportional to
    \begin{equation}\label{2.X4}
  e^{\sum_{j=1}^p x_j^2}  \int_{-\infty}^\infty d \lambda_1 \cdots   \int_{-\infty}^\infty d \lambda_p \,
  \prod_{l=1}^p \lambda_l^N   e^{-\sum_{j=1}^p \lambda_j^2} \prod_{1 \le j < k \le p} (\lambda_k - \lambda_j)
  \det [ e^{-2i \lambda_j x_k} ]_{j,k=1,\dots,p}.
   \end{equation}      
   Application now of Andr\'eief's integration formula \cite{Fo19} reduces this to
      \begin{equation}\label{2.X5}
  e^{\sum_{j=1}^p x_j^2} \det \bigg [ \int_{ -\infty}^\infty  \lambda^{N+j-1}   e^{- \lambda^2} e^{-2i \lambda x_k} \, d \lambda \bigg ]_{j,k=1,\dots,p}.
   \end{equation}   
  Using now the integral form for the Hermite polynomials (\ref{2.3c})
  allows (\ref{2.X5}) to be recognised as being proportional to the determinant formula in (\ref{2.3}) in the GUE
   case. The proportionality is fixed by noting that both sides of (\ref{2.X1}) reduce to $ \prod_{l=1}^p x_l^N $ is the limit that
 each $x_l \to \infty$. 
 
 The duality (\ref{2.X1a}) follows from (\ref{2.X1}) by setting all the $x_l$ equal. However, for reference in settings where an
 analogue of (\ref{2.X1}) is not available, we outline a direct strategy beginning with (\ref{2.3}). This is to first 
 set $x_1 = x$ and then subtract column 1 from column two, divide by the existing prefactor $(x_2-x)$ and take the limit $x_2 \to x$. Next subtract
 columns 1 and $(x_3 - x)$ times column 2 from column 3, divide by the existing prefactor $(x_3-x)^2$ and take the limit $x_3 \to x$. By repeating this procedure in succession,
 it is possible to subtract from column $k$ the first $(k-1)$ terms in its Taylor series expansion about $x$, which after dividing by $(x_k-x)^{k-1}$ and taking the
 limit $x_k \to x$ gives $D_x^{k-1} p_{N+j-1}(x)/(k-1)!$ (here the $D_x$ denotes  derivative operation) as the entry in row $j$ and column $k$. 
 This shows
   \begin{equation}\label{2.3A}
\Big  \langle    \det ( x \mathbb I_N - H)^p  \Big  \rangle_{ {\rm ME}_{2,N}[w]} = 
{1 \over \prod_{l=1}^p l!} \det [ D_x^{k-1} p_{N+j-1}(x) ]_{j,k=1,\dots,p}.
 \end{equation}
Repeated use of the Hermite polynomial identity
 $H_n'(x) = 2x H_n(x) - H_{n+1}(x)$ then allows us to deduce that in the Gaussian case the RHS of
 (\ref{2.3A}) is proportional to
 the Hankel determinant 
    \begin{equation}\label{2.3b1}
\det [ H_{N+j+k-2}(x) ]_{j,k=1,\dots,p}.
 \end{equation}
 Next, a well known formula in random matrix theory \cite[Eq.~(5.75)]{Fo10}, due to Heine  in the context of the study of orthogonal
 polynomials \cite{Sz75}, gives that in general
 for $\alpha_j := \int_I a(t) t^j \, dt$, one has
     \begin{equation}\label{2.3b}
    \det [ \alpha_{j+k} ]_{j,k=0,\dots,n-1} = {1 \over n!} \int_I dx_1 \cdots  \int_I dx_n \,
    \prod_{l=1}^n a(x_l) \prod_{1 \le j < k \le n} (x_k - x_j)^2.
 \end{equation}     
 Using the integral form for the Hermite polynomials (\ref{2.3c})
   allows for (\ref{2.3b1})  to be rewritten in the form of the RHS of (\ref{2.3b}), which after minor manipulation can be recognised as the
 RHS of (\ref{2.X1a}).   
 \end{proof}
 
 If we consider now (\ref{2.2})  for $\mathcal H_N$ the class of Wigner matrices beyond the GUE, and restrict to the case $p=2$,
 a duality formula analogous to  (\ref{2.X1}) is known from \cite[Eq.~(2.16) after scaling]{Sh11}. 
 
 \begin{prop}
 Consider a complex Hermitian Wigner matrix $\tilde{H} = [\tilde{H}_{jk} ]_{j,k=1,\dots,N}$. Require that the real and imaginary parts
 for $j < k $ are independently drawn from the same distribution, which has first and third moments zero, second moment
 ${1 \over 2}$, fourth moment $b$ and corresponding cumulant $\kappa_4 = b - {1 \over 4}$. Let $\tilde{H}_{jj}$ also be independently
 drawn from this distribution, together with the scaling of multiplying by $\sqrt{2}$. Further scale all the entries by considering
 ${H} = {1 \over \sqrt{2}} \tilde{H}$ to define the ensemble $\mathcal H_N$. With $X = {\rm diag} (x_1, x_2)$ one has
   \begin{multline}\label{2.3d}
 \Big  \langle   \det ( x_1 \mathbb I_N - H)   \det ( x_2 \mathbb I_N - H)  \Big  \rangle_{ \mathcal H_N}  \\
  = i^{-2N}
 \bigg  \langle   \Big ( \det ( i X - Q ) + t \epsilon(\kappa_4) \Big )^N \bigg  \rangle_{ Q \in {\rm GUE}_{2}, \, t \in \mathcal N(0,1/(2|\kappa_4|))},
 \end{multline}  
 where $\epsilon(x) = x$, $(x>0)$ and $\epsilon(x) = -i x$, $(x<0)$.
 \end{prop}
 
 \begin{proof} (Comment only.) The proof is based on Grassmann integration involving formal anti-symmetric variables
 $\{ \psi_j \}$,  $\{ \bar{\psi}_j \}$. Each determinant on the LHS can be rewritten according to the key formula
   \begin{equation}\label{2.3e}
   \int \exp \Big ( \sum_{j,k=1}^N A_{j,k} \bar{\psi}_j \psi_k \Big ) \prod_{j=1}^N d  \bar{\psi}_j  d  \psi_j = \det A;
 \end{equation}  
 see e.g.~\cite{Ha00}.   An instructive preliminary exercise is to make use of this identity and the calculus of
 Grassmann integration to give a derivation of (\ref{1.0}).
 \end{proof}
 
 \subsection{Averages of reciprocals of characteristic polynomials for the GUE}
 A companion to the determinant identity (\ref{2.3}) applying to the average of a product of characteristic polynomials in a
 general $\beta = 2$ matrix ensemble ME${}_{2,N}[w]$ relates to a product of reciprocals of characteristic polynomials
 \cite{FS03e}, \cite{BDS03}, \cite[Prop.~5.3.1]{Fo10},
   \begin{equation}\label{2.9} 
   \bigg \langle {1 \over \prod_{l=1}^q \det  ( v_j - X)} \bigg \rangle_{{\rm ME}_{2,N}[w]} =
   {(-1)^{q(q-1)/2} \over \prod_{j=N-q}^{N-1} h_j} \det \bigg [ \int_{-\infty}^\infty {p_{N-q + k - 1}(x) \over v_j - x} w(x)\, dx \bigg ]_{j,k=1,\dots,q}.
  \end{equation}
 Set $q=1$ and consider the case of the Gaussian weight, for which the $p_n(x)$ are proportional to the Hermite polynomials. 
 Since the support of the Gaussian weight is the whole real line, it is required that the $v_j$ have a nonzero imaginary part.
 Using the contour integral form
 (equivalent to Rodrigues formula)
 $$
e^{-x^2} H_{N-1}(x) = {(N-1)! \over 2 \pi i } \int_{I_+ + I_-} {e^{-z^2} \over ( x - z)^N} \, dz,
$$
cf.~(\ref{2.3c}),
where for $x$ real $I_+$ ($I_-$) can be taken as the real line shifted up (down) by $i \epsilon$ ($\epsilon > 0$) and traversed
in the negative (positive)  sense. Substituting in the integral on the RHS of (\ref{2.9}), and taking the $z$ integral outside of the $x$ integral
allows the former to be computed, with the result
 \begin{equation}\label{2.9A} 
   \bigg \langle {1 \over  \det  ( v - X)} \bigg \rangle_{{\rm ME}_{2,N}[e^{-x^2}]} = {1 \over \sqrt{\pi}} \int_{-\infty}^\infty {e^{-z^2} \over (v - z)^N} \, dz = \Big \langle {1  \over (v - z)^N} 
   \Big \rangle_{z \in \mathcal N(0,1/\sqrt{2})},
      \end{equation}
      which itself is a companion duality to the GUE case of (\ref{1.0}).
      
     Extending (\ref{2.9A}), Proposition \ref{P2.1}  has a counterpart for averaged reciprocal characteristic polynomials with respect to the GUE.
      \begin{prop}\label{P2.1+} (\cite{De09})
Let $X = {\rm diag} \, (x_1,\dots,x_p)$, where each $x_i$ has a nonzero imaginary part. We have
    \begin{equation}\label{2.X1+}
    \bigg  \langle   \prod_{l=1}^p {1 \over \det ( x_l \mathbb I_N - H)   }  \bigg  \rangle_{ {\rm GUE}_{N}} =
    \bigg  \langle  {1 \over  \det ( X - H)^N } \bigg  \rangle_{ {\rm GUE}_{p}}.
   \end{equation} 
   In particular
    \begin{equation}\label{2.X1a+}
    \bigg  \langle    {1 \over \det ( x \mathbb I_N - H)^p   }  \bigg  \rangle_{ {\rm GUE}_{N}} =
    \bigg  \langle  {1 \over \det (   x \mathbb I_p - H)^N } \bigg  \rangle_{ {\rm GUE}_{p}}.
   \end{equation}     
    \end{prop}
      
 We will see in Section \ref{S3.3} that these results are special cases of dualities applying    to the GUE
 generalised to include a source matrix.
 
 \subsection{Powers of characteristic polynomials and the LUE and JUE}
 
 There are no known (literal) analogues of (\ref{2.X1}) beyond the Gaussian case. However progress can be made if one
 seeks analogues of (\ref{2.X1a}) only. This is possible for the case of the unitary invariant ensemble with weight
 $w(x) = x^a e^{-x} \mathbbm 1_{x>0}$, referred to as the Laguerre unitary ensemble and denoted LUE (or LUE${}_{N,a}$
 when the size of the matrix $N$ and the Laguerre parameter $a$ are to be emphasised), or the weight
 $w(x) = x^{a_1}(1 - x)^{a_2}   \mathbbm 1_{x>0}$, referred to as the Jacobi unitary ensemble JUE (or 
 JUE${}_{N,(a_1,a_2)}$). Whereas the GUE case of the determinant in (\ref{2.3A}) was shown to give rise to a Hankel
 determinant, it turns out that both the LUE and JUE cases give rise to Toeplitz determinants.
 
 \begin{prop}\label{P2.3} (\cite[for the LUE case with $\alpha = 2$]{FH94}) 
We have
  \begin{align}\label{2.4}
  \det [D_x^{k-1} L_{N+j-1}^{(a)}(x) ]_{j,k=1,\dots,p} & = (-1)^{p(p-1)/2} \det [  L_{N+j-k}^{(a+k-j)}(x) ]_{j,k=1,\dots,p}, \nonumber \\
   \det [D_x^{k-1} P_{N+j-1}^{(a_1,a_2)}(x) ]_{j,k=1,\dots,p} & \propto \det [  P_{N+j-k}^{(a_1+k-j,a_2+k-j)}(x) ]_{j,k=1,\dots,p}. 
   \end{align}  
    \end{prop}
 
 \begin{proof}
  Consider for definiteness the LUE case. The relevant orthogonal polynomials (up to the normalisation condition of 
  being monic) are the Laguerre polynomials
  $\{ L_n^{(a)}(x) \}_{n=0,1,\dots}$. The determinant  in  (\ref{2.3A}) is to be manipulated by the elementary row operation of
  replacing the row $p-j$ by row $p-j$ minus row $p-j+1$ for $j=0,\dots,p-2$ in order, then simplifying using the
  Laguerre polynomial identity $L_n^{(a)}(x) = L_n^{(a+1)}(x)  - L_{n-1}^{(a+1)}(x)$. We repeat this same
  procedure, finishing one further row down from the top at each iteration, a total of $p-2$ further times to conclude
  $$
  \det [ D_x^{k-1} L_{N+j-1}^{(a)}(x) ]_{j,k=1,\dots,p} =
  \det [D_x^{k-1}  L_{N+j-1}^{(a+1 - j)}(x) ]_{j,k=1,\dots,p}.
  $$
  Use of the Laguerre polynomial formula $D_x L_p^{(a)}(x) = -  L_{p-1}^{(a+1)}(x)$ a total of $j$ times in row
  $j$ allows for  a substitution of the matrix element on the RHS of this to give the RHS of (\ref{2.4}).

  \end{proof}
  
  The identities in Proposition \ref{P2.3} provide duality formulas for the average in (\ref{2.3A}) in the cases of the LUE and JUE.
  However, unlike the case of the GUE, these dualities no longer relate back to the same ensembles, but rather to a class of
  $\beta = 2$ circular ensembles (\ref{2.1X}). The reason for this is that with $\beta_j = \int_{-1/2}^{1/2} b(e^{2 \pi i \theta})
  e^{2 \pi i \theta (j-k)} \, d \theta$ one has the general identity \cite{Sz75}
  \begin{equation}\label{2.3B}
    \det [ \beta_{j-k} ]_{j,k=1,\dots,n} = {1 \over n!} \int_{-1/2}^{1/2}  d\theta_1 \cdots  \int_{-1/2}^{1/2} d\theta_n \,
    \prod_{l=1}^n b(e^{2 \pi i \theta_l}) \prod_{1 \le j < k \le n} | e^{2 \pi i \theta_k}  - e^{2 \pi i \theta_j} |^2;
 \end{equation}  
  cf.~(\ref{2.3b}).
  
   \begin{cor}\label{C2.1}
  We have \cite{Fo93c}
   \begin{equation}\label{2.4C} 
   \Big  \langle    \det ( x \mathbb I_N - H)^p     \Big  \rangle_{ {\rm LUE}_{N,a}}  \propto
     \Big  \langle \prod_{l=1}^p e^{- x e^{2 \pi i \theta_l}} \Big \rangle_{{\rm CE}_{2,p} [e^{\pi i (a - N)\theta}
     |1 + e^{2 \pi i \theta} |^{a + N}]}  
   \end{equation}    
    and
     \begin{multline}\label{2.4D}   
   \Big  \langle    \det ( x \mathbb I_N - H)^p     \Big  \rangle_{ {\rm JUE}_{N,(a_1,a_2)}}  \\ \propto
(1 - x)^{pN}  \Big  \langle \prod_{l=1}^p \Big ( 1 - {x \over 1 - x}  e^{ 2 \pi i  \theta_l} \Big )^{N+a_2}
  \Big \rangle_{{\rm CE}_{2,p} [e^{\pi i (a_1 - N)\theta}
     |1 + e^{2 \pi i \theta} |^{a_1 + N}]}.   
   \end{multline}
  \end{cor}
  
 \begin{proof}
 To apply (\ref{2.3B}) it is necessary to be able to identify the corresponding generating functions $b(z)$ that are consistent
 with the Toeplitz determinants in (\ref{2.4}). For this we first note the hypergeometric polynomial forms of the Laguerre and
 Jacobi polynomials
 \begin{equation}\label{2.4A} 
 L_n^{(a)}(x) \propto {}_1 F_1(-n,a+1;x), \quad P_n^{(a_1,a_2)}(1 - 2x) \propto {}_2 F_1 (-n,1+a_1 + a_2 + n; a_1 + 1;x)
  \end{equation}  
  (the relevant orthogonal polynomials for the Jacobi weight $w(x) = x^{\alpha_1} (1 - x)^{a_2} \mathbbm 1_{0 < x < 1}$ are the Jacobi
  polynomials $P_n^{(a_1,a_2)}(1 - 2x)$). For these hypergeometric functions ${}_1F_1$ and ${}_2 F_1$ more generally
  we have the circular integral forms (see e.g.~\cite[Exercises 13.1~q.4(i) with $N=1$ and Eq.~(13.11) with $N=1$]{Fo10})
  \begin{align}\label{2.4B} 
 {}_1F_1(-b;a+1;t) & \propto \int_{-1/2}^{1/2} e^{\pi i \theta (a - b)} | 1 + e^{2 \pi i \theta} |^{a+b} e^{- t e^{2 \pi i \theta}} \, d \theta \nonumber \\
  {}_2F_1(r,-b;a+1;t) & \propto \int_{-1/2}^{1/2} e^{\pi i \theta (a - b)} | 1 + e^{2 \pi i \theta} |^{a+b}  (1 + t e^{2 \pi i \theta})^{-r} \, d \theta,
  \end{align}
  valid for Re$\, (a+b) > -1$.
  
  The LUE duality  (\ref{2.4C}) follows immediately by substituting for the Laguerre polynomial on the RHS of the first
  identity in (\ref{2.4}) as is consistent with (\ref{2.4A}) and (\ref{2.4B}). In the Jacobi case the particular ${}_2 F_1$ in
  (\ref{2.4}) (after the replacement $x \mapsto 2x-1$) does not allow the circular integral in (\ref{2.4B}) to be identified
  in a form as required by (\ref{2.3B}). However, upon use of the Euler transformation
  \begin{equation}\label{2.4E+}  
   {}_2 F_1 (-n,1+a_1 + a_2 + n; \alpha_1 + 1;x) = (1-x)^{n}  {}_2 F_1 \Big (-a_2 - n,-n; a_1 + 1; - {x \over 1 - x}
   \Big )
   \end{equation}  
   and factoring $ (1-x)^{n} $ out of the determinant, the form as  required by (\ref{2.3B}) is obtained, and (\ref{2.4D})
   results.
 \end{proof}     
 
 \begin{remark} ${}$ \\
 1.~The RHS of (\ref{2.4C}) and (\ref{2.4D}) equals unity for $x=0$. Requiring the same on the LHS provides a way
 to fix the proportionality constants. \\
 2.~Later (see (\ref{2.4Di}) below with $\beta = 2$), as an alternative to (\ref{2.4D}) we will obtain a result which implies that the RHS can be rewritten in the manifestly polynomial form
 \cite[Exercises 13.2 q.6(ii) with $\beta = 2$]{Fo10}
   \begin{equation}\label{2.4D+}   
(1 - x)^{pN}  \Big  \langle \prod_{l=1}^p \Big ( 1 - {x \over 1 - x}  e^{ 2 \pi i  \theta_l} \Big )^{N}
  \Big \rangle_{{\rm CE}_{2,p} [e^{\pi i (a_1-a_2 - N)\theta}
     |1 + e^{2 \pi i \theta} |^{a_1 + a_2 + N}]} .  
   \end{equation} 
\end{remark} 

\subsection{Characteristic polynomial dualities for unitary random matrices}
The Wronskian type determinant formula (\ref{2.3A}) remains valid in the setting of
circular ensembles (\ref{2.1X}) with $\beta = 2$. In saying this, it is assumed that
the orthogonality relation now has the form $\int_{-1/2}^{1/2} w(\theta)
p_j(e^{2 \pi i \theta}) p_k(e^{-2 \pi i \theta}) \, d \theta = h_j \delta_{j,k}$, and thus
involves the complex conjugate operation. For the choice of weight
 \begin{equation}\label{2.6a}
 w(\theta) = e^{\pi i \theta (a_1 - a_2)} |1 + e^{2 \pi i \theta} |^{a_1+a_2}
 \end{equation} 
 one has the hypergeometric polynomial expressions
 \begin{equation}\label{2.6b}
 p_n(z) \propto {}_2 F_1(-n,a_1+1;-a_2-n+1;-z) \propto  {}_2 F_1(-n,a_1+1;a_1+a_2+1;1+z),
 \end{equation}  
 where the second form follows upon use of a transformation formula satisfied
 by the $ {}_2 F_1$ function (for the first see \cite[Exercises 5.5 q.3]{Fo10}). Upon use now of the classical integral representation
 of Euler type we can deduce the duality
 \begin{equation}\label{2.6c} 
 \Big \langle \prod_{l=1}^N ( z - e^{2 \pi i \theta_l}) \Big \rangle_{{\rm ME}_{2,N}[w]} \propto
 \langle (1 - (1 + z) x)^N \rangle_{x \in {\rm B}[a_1+1,a_2]},
 \end{equation}  
 where B$[b,c]$ denotes the beta distribution, supported on
 $(0,1)$ with density proportional to $x^{b-1}(1-x)^{c-1}$. 
 
 For circular ensembles, of more interest than averages of characteristic polynomials
 is averages of their absolute value squared. Noting that for $|z| = 1$
 \begin{equation}\label{2.6d} 
 \prod_{l=1}^N | z - e^{2 \pi i \theta_l} |^2 = (-1)^N z^N \prod_{l=1}^N e^{- 2 \pi i \theta_l} (1 - \bar{z}  e^{ 2 \pi i \theta_l})^2,
 \end{equation}
 the equivalent problem in relation to averages of even powers of this quantity is to consider
 \begin{equation}\label{2.6e}  
 \Big \langle \prod_{l=1}^N (1 - \bar{z}  e^{ 2 \pi i \theta_l})^{2q} \Big \rangle_{{\rm ME}_{2,N}[e^{-2 \pi i  q\theta} w]} 
  \end{equation}
  for $q$ a positive integer. With $N$ replaced by $p$, $2q$ by $N$, and $a_1+a_2$ replaced by $a_1+a_2 + 2q$ we
  recognise this average from (\ref{2.4D+}), which we know satisfies a duality with the JUE average
  in (\ref{2.4D}). Working along this line, together with its $\beta$ extension, will be given in \S \ref{S3.6}.

\section{Powers and products of characteristic polynomial dualities for $\beta$ ensembles}\label{S3}

\subsection{Jack polynomial theory}

  Common to all $\beta$ generalisations is the
   essential use of Jack
   polynomial theory \cite[Ch.~VI.6]{Ma95}, \cite[Ch.~12]{Fo10}, \cite[Ch.~7]{KK09}.
Jack polynomials, to be denoted $P_\kappa^{(\alpha)}(\mathbf x)$,
 depend on a  set of variables $\mathbf x = \{x_j\}$  and a parameter $\alpha > 0$. When  
expanded in terms of the monomial symmetric polynomials $\{ m_\mu(\mathbf x) \}$ they have the triangular structure
\begin{equation}\label{4.0}
P_\kappa^{(\alpha)}(\mathbf x) = m_\kappa(\mathbf x) +
\sum_{\mu < \kappa} c_{\kappa, \mu}^{(\alpha)} 
m_\mu (\mathbf x).
\end{equation}
Here the 
notation $\mu < \kappa$ denotes the partial order on partitions defined by the requirement that $\sum_{i=1}^s \mu_i \le \sum_{i=1}^s \kappa_i$, for each $s=1,\dots,\ell(\kappa)$ ($\ell(\kappa)$
 denotes the number of nonzero parts of $\kappa$). Once this structure is stipulated, the Jack polynomials can be uniquely determined by an orthogonality property.
 Within the theory, several orthogonalities hold. Here we make note the one which relates to a random matrix average
\begin{equation}\label{3.41}
\Big \langle   P_\kappa^{(2/\beta)}(\mathbf z ) P_\mu^{(2/\beta)}(\bar{\mathbf z} ) \Big \rangle_{{\rm CE}_{\beta,N}[1]} \propto \delta_{\kappa,\mu},
\end{equation}
where $\mathbf z = \{ e^{2 \pi i \theta_l} \}_{l=1}^N$;
see e.g.~\cite[Prop.~12.6.3]{Fo19} for the normalisation. 

The Jack polynomials with $\alpha = 1$ are the familiar Schur polynomials \cite{Ma95}, denoted $s_\kappa(\mathbf x)$,
while for $\alpha = 2$ they up to choice of
normalisation  the zonal polynomials
of mathematical statistics \cite{Mu82}. Both these cases permit an interpretation in terms of spherical functions, as does the case
$\alpha =1/2$ \cite{Ma95}. These three cases furthermore relate to certain random matrix integrals.
In this regard,  let us define the particular hypergeometric function in two sets of
variables $\mathbf x = \{x_i\}_{i=1,\dots,N}$ and $\mathbf y = \{y_i\}_{i=1,\dots,N}$ 
 as a series in Jack polynomials according to
\begin{equation}\label{6.3g}
{\vphantom{\mathcal F}}_0^{\mathstrut}\mathcal F_0^{(\alpha)}(\mathbf x; \mathbf y) = \sum_{\kappa} {\alpha^{| \kappa |} \over h_\kappa'} 
{P_\kappa^{(\alpha)}(\mathbf x) P_\kappa^{(\alpha)}(\mathbf y) \over
P_\kappa^{(\alpha)}((1)^N)
}.
\end{equation}
Here $(1)^N$ denotes the point $x_i=1$ ($i=1,\dots,N$) and
 \begin{equation}\label{2.2h}
h_\kappa' = \prod_{s \in \kappa} ( \alpha (a(s) + 1) + l(s))
\end{equation}
(in some Jack polynomial literature,  for example in \cite[Ch.~12 and 13]{Fo10}, this quantity is alteratively denoted $d_\kappa'$),
where the quantities $a(s), l(s)$ are the arm and leg lengths at position $s$ in the diagram associated with $\kappa$ \cite{Ma95}.

While the quantity (\ref{6.3g}) is not restricted to the three special $\alpha$, its matrix integral interpretation is. For the latter,
let $U$ be a complex unitary ($\beta = 2$), real orthogonal ($\beta = 1$), or symplectic ($\beta = 4$ unitary matrix, and let $d^{\rm H} U$
denote the corresponding Haar measure. Then in terms of (\ref{6.3g}) we have (see e.g.~\cite[Eq.~(13.146)]{Fo10})
\begin{equation}\label{6.2g2}
\int e^{ {\rm Tr} ( UG U^\dagger  X_0)} \,
d^{\rm H} U =:  {\vphantom{\mathcal F}}_0^{\mathstrut} \mathcal F_0^{(2/\beta)}(\mathbf x; \mathbf x^{(0)}),
\end{equation}
where $G, X_0$ are Hermitian matrices with complex $(\beta = 2)$, real ($\beta = 1$) and quaternion ($\beta = 4$) elements with
eigenvalues $\mathbf x$ and $\mathbf   x^{(0)}$ respectively.  Note the simplification in the case that $X_0$ is proportional to
the identity (i.e.~all eigenvalues take on the constant value $c$)
\begin{equation}\label{6.2g1}
{\vphantom{\mathcal F}}_0^{\mathstrut}\mathcal F_0^{(\alpha)}(\mathbf x; \mathbf y) \Big |_{\mathbf y = (c)^N} = e^{cx_1 + \cdots + cx_N},
\end{equation}
which in fact holds for general $\alpha \ge 0$ (see e.g.~\cite[Eq.~(13.3)]{Fo10}). General aspects of the group integral
in (\ref{6.2g2}), along the lines initiated by Harish-Chandra \cite{HC57}, can be found in \cite{Mc21}.

\subsection{Gaussian $\beta$ ensemble}
There is a very clean generalisation of (\ref{2.X1a}) and (\ref{2.X1a+}) to the case of the 
Gaussian $\beta$ ensemble.

\begin{prop}\label{P3.1}
We have
\cite{BF97a}, \cite{De09}, \cite[Eq.~(13.162)]{Fo10}
\begin{equation}\label{2.X1c}
    \Big  \langle    \prod_{j=1}^N ( x  -  \sqrt{\alpha} \lambda_j )^p     \Big  \rangle_{ {\rm ME}_{2/\alpha,N}[e^{-\lambda^2}]} =
   i^{-pN} \Big  \langle   \prod_{j=1}^p  ( i  x  -  \lambda_j)^N \Big  \rangle_{{\rm ME}_{2 \alpha,p}[e^{-\lambda^2}]}
   \end{equation}  
   and
   \begin{equation}\label{2.X1d}
    \bigg  \langle    {1 \over \det ( x \mathbb I_N - H)^{p\beta/2}   }  \bigg  \rangle_{ {\rm ME}_{N,\beta}[e^{-\lambda^2}]} =
    \bigg  \langle  {1 \over \det (   x \mathbb I_p - H)^{N \beta/2} }\bigg  \rangle_{ {\rm ME}_{p,\beta}[e^{-\lambda^2}]},
   \end{equation} 
   (note that in this latter duality, in distinction to (\ref{2.X1c}), the value of $\beta$ in the ensemble is the same
   on both sides, and $x$ must have a nonzero imaginary part).
   \end{prop}
   
 In Section \ref{S3.3} the Gaussian $\beta$ ensemble will be generalised to involve a source. Specialising dualities that can be derived in that
setting will be shown also to imply Proposition \ref{P3.1}. Other derivations of (\ref{2.X1c}) are also known.
The first \cite{BF97a} begins by generalising the LHS from a function of one variable, to a function of $p$ variables
by the replacement
 \begin{equation}\label{3.g+}
 \Big  \langle    \det ( x \mathbb I_N - H)^p     \Big  \rangle  
  \mapsto  \Big  \langle   \prod_{l=1}^p \det ( x_l \mathbb I_N - H)     \Big  \rangle.
  \end{equation}     
  It is then verified that this generalised quantity satisfies a set of partial differential equations associated with the
  Calogero-Sutherland quantum many body problem with harmonic confinement. This leads to the conclusion that
  the generalised average is in fact the polynomial part of a particular class of eigenfunctions for that model system,
  known as the generalised Hermite polynomials $\{P^{(\rm H)}_\kappa(\mathbf x;\alpha)\}$ (specifically, with
  $\kappa = (p)^N$). On the other hand, by developing the theory associated with the latter, an alternative integral
  representation can be deduced, which leads to the RHS of (\ref{3.g+}). For this, introduce the measure
  \begin{equation}\label{3.h+} 
  d \mu^{\rm G}(\mathbf y) :=  {\rm ME}_{p,2/\alpha}[e^{-y^2}] d \mathbf y.
    \end{equation}   
  Then we have from \cite[Cor.~3.2]{BF97a}
   \begin{equation}\label{3.h+A} 
  e^{-(x_1^2+  \cdots + x_p^2)} P^{(\rm H)}_\kappa(\mathbf x;\alpha) \propto \int_{\mathbb R^p }
   {\vphantom{\mathcal F}}_0^{\mathstrut}\mathcal F_0^{(\alpha)}(2 \mathbf y; -i\mathbf x) P_\kappa^{(\alpha)}(i \mathbf y) \,   d \mu^{\rm G}(\mathbf y). 
   \end{equation}  
   Now note that for $\kappa = (p)^N$, $ P_\kappa^{(\alpha)}(i \mathbf y) = \prod_{l=1}^p y_l^N $, then set all the entries of $\mathbf x$ to equal
   $x$ so that the $ {\vphantom{\mathcal F}}_0^{\mathstrut}\mathcal F_0$ function can be simplified according to  (\ref{6.2g1}).
   Then the RHS of (\ref{3.g+}) can be identified as the RHS of (\ref{2.X1c}).
   
   Another distinct derivation of (\ref{2.X1c}) proceeds by first establishing a class of
   duality formulas of independent interest.
   
\begin{prop}\label{C1} (\cite[Th.~8.5.3]{Du03},  \cite[Lemma 2.6]{DE05},  \cite[Prop.~4]{De09})
Let $\kappa'$ denote the conjugate partition, defined by interchanging the roles of the rows
and columns in the diagram of $\kappa$.
For $\ell(\kappa) \le N$ and $\ell(\kappa') \le p$ we have
\begin{equation}\label{GF}
(-2/\beta)^{|\kappa|/2} \Bigg \langle {P_\kappa^{(2/\beta)}(\mathbf x) \over 
P_\kappa^{(2/\beta)}((1)^N)}
\Bigg \rangle_{{\rm ME}_{\beta,N}[e^{-\lambda^2}]}
=
\Bigg \langle {P_{\kappa'}^{(\beta/2)}(\mathbf x) \over 
P_{\kappa'}^{(\beta/2)}((1)^p) }
\Bigg \rangle_{{\rm ME}_{4/\beta,p}[e^{-\lambda^2}]}. 
\end{equation}  
(Note in particular that under the stated conditions there is no dependence on $N$ or $p$ in this identity.)
\end{prop}   

To derive (\ref{2.X1c}) from knowledge of this, we require too knowledge of 
the Jack polynomial homogeniety property $ P_\kappa^{(\alpha)}(c \mathbf x) = c^{|\kappa|}
P_\kappa^{(\alpha)}( \mathbf x) $ for any scalar $c$, and also the dual Cauchy identity (see e.g.~\cite[Eq.~(12.187)]{Fo10})
\begin{equation}\label{SM2}
\prod_{k=1}^N \prod_{l=1}^p (1 - x_k y_l) = \sum_{\kappa \subseteq (p)^N} (-1)^{|\kappa|} P_\kappa^{(\alpha)}(\mathbf x)  P_{\kappa'}^{(1/\alpha)}(\mathbf y), 
\end{equation}
where $(N)^p$ denotes the partition with $p$ parts all equal to $N$. In (\ref{GF}) we cross multiply the denominators, and multiply both sides by
$(-1)^{|\kappa|/2} c^{|\kappa|}$. Summing both sides over $\kappa \subseteq (N)^p$ (which is equivalent to $\kappa' \subseteq (p)^N$), we can apply
the dual Cauchy identity (\ref{SM2}) to obtain, after replacing $c$ by $1/x$ and multiplying both sides by $x^{pN}$,  (\ref{2.X1c}).

\subsection{Laguerre and Jacobi $\beta$ ensembles --- positive powers of characteristic polynomials}
The dualites of Corollary \ref{C2.1} also allow for clean $\beta$ generalisations \cite{Fo93c}, \cite[Exercises 13.2 q.6]{Fo10}.
   \begin{prop}\label{P2.6}
   We have
   \begin{equation}\label{2.4Ci} 
   \Big  \langle    \det ( x \mathbb I_N - H)^p     \Big  \rangle_{ {\rm ME}_{\beta,N}[x^a e^{-\beta x/2}]}  \propto
     \Big  \langle \prod_{l=1}^p e^{- x e^{2 \pi i \theta_l}} \Big \rangle_{{\rm CE}_{4/\beta,p} [e^{2\pi i (a + 1)\theta/\beta - \pi i  \theta ( N + 1)}
     |1 + e^{2 \pi i \theta} |^{2(a +1)/\beta + N-1}]}  
   \end{equation}    
    and
     \small
     \begin{multline}\label{2.4Di}   
   \Big  \langle    \det ( x \mathbb I_N - H)^p     \Big  \rangle_{ {\rm ME}_{\beta,N}[x^{a_1} (1 - x)^{a_2}]}  \\ \propto
(1 - x)^{pN}  \Big  \langle \prod_{l=1}^p \Big ( 1 - {x \over 1 - x}  e^{ 2 \pi i  \theta_l} \Big )^{N}
  \Big \rangle_{{\rm CE}_{4/\beta,p} [e^{2 \pi i (a_1 - a_2  )\theta/\beta - \pi i \theta N}
     |1 + e^{2 \pi i \theta} |^{2(a_1 + a_2 + 2)/\beta + N -2 }]} \\
   \propto
(1 - x)^{pN}  \Big  \langle \prod_{l=1}^p \Big ( 1 - {x \over 1 - x}  e^{ 2 \pi i  \theta_l} \Big )^{N-1+2(a_2+1)/\beta}
  \Big \rangle_{{\rm CE}_{4/\beta,p} [e^{2 \pi i (a_1 +1  )\theta/\beta - \pi i \theta (N+1)}
     |1 + e^{2 \pi i \theta} |^{2(a_1 + 1)/\beta + N -1}]} .     
   \end{multline}
   \normalsize
  \end{prop}
  
  In providing a proof outline, use will be made of hypergeometric functions based on Jack polynomials which
  relate to the series implied by the LHS of (\ref{6.2g1}); an extended account is given in \cite[\S 13.1]{Fo10}. To define these, introduce
  the generalised Pochhammer symbol
  \begin{equation}\label{4.3m}
\quad[u]_\kappa^{(\alpha)} := \prod_{l=1}^N {\Gamma(u - (j-1)/\alpha + \kappa_l) \over \Gamma(u - (j - 1)/\alpha)};
\end{equation}
see e.g.~\cite[Eq.~(12.46)]{Fo10}. The family of hypergeometric functions of interest are then specified by
\begin{equation}\label{3.40}
  {\vphantom{F}}_p^{\mathstrut} F_q^{(\alpha)}(a_1,\dots,a_p;b_1,\dots,b_q; \mathbf x):=\sum_\kappa  {\alpha^{| \kappa |} \over h_\kappa'} 
 \frac{[a_1]^{(\alpha)}_\kappa\dots [a_p]^{(\alpha)}_\kappa }{[b_1]^{(\alpha)}_\kappa
\dots [b_q]^{(\alpha)}_\kappa} 
P_\kappa^{(\alpha)}(\mathbf x). 
\end{equation}
With $\mathbf x = \{x_i\}_{i=1}^m$, the sum is over all partitions $\kappa_1 \ge \kappa_2 \ge \cdots \ge \kappa_m \ge 0$, and conventionally the sum is
performed in order of increasing 
$|\kappa|$.

\smallskip
\noindent
{\it Proof outline of Prop.~\ref{P2.6}.} In both the Laguerre and Jacobi cases, a pathway to the dualities is to generalise the
average on the LHS by introducing variables $\mathbf x = \{x_l\}_{l=1}^p$ according to (\ref{3.g+}).
This function of $p$ variables can be uniquely characterised by a set of $p$ partial differential equations which moreover
admit the Jack polynomial based hypergeometric function series solutions \cite{Ka93}
 \begin{equation}\label{2.4Ea}  
 {\vphantom{F}}_1^{\mathstrut} F_1^{(\beta/2)}(-N;2(a+p)/\beta; \mathbf x), \quad
  {\vphantom{F}}_2^{\mathstrut} F_1^{(\beta/2)}(-N,N - 1 + 2(a_1 + a_2 + p + 1)/\beta;2(a_1+p)/\beta;\mathbf x).
\end{equation}
One notes that the first of these can be obtained from the second by the scaling $\mathbf x \mapsto {\beta \over 2 a_2} \mathbf x$ and then
taking the limit $\alpha_2 \to \infty$, as is consistent with the relation between the Jacobi and Laguerre weights.
Specialising then to the Jacobi case, it follows
 \begin{equation}\label{2.4Ea+}  
 \Big  \langle    \det ( x \mathbb I_N - H)^p     \Big \rangle_{ {\rm ME}_{\beta,N}[x^{a_1} (1 - x)^{a_2}]}   \propto 
   {\vphantom{F}}_2^{\mathstrut} F_1^{(\beta/2)}(-N,N - 1 + 2(a_1 + a_2 + p + 1)/\beta;2(a_1+p)/\beta;(x)^p).
 \end{equation}  
 
With $\mathbf x = \{ x_l \}_{l=1}^p$, the method of partial differential equations can be used to establish the explicit functional form (a generalised binomial formula) \cite{Ka93}
 \begin{equation}\label{2.4E}  
 {\vphantom{F}}_1^{\mathstrut} F_0^{(\beta/2)}(r; \underline{\: \:} ; \mathbf x) = \prod_{l=1}^p ( 1 - x_l)^{-r}.
 \end{equation} 
 This together with the orthogonality (\ref{3.41}) can be used to deduce the integration formula \cite[Eq.~(12.142)]{Fo10}
  \begin{equation}\label{2.4Dj}   
   \Big  \langle    P_\kappa^{(\alpha)}(-\mathbf z)     \Big  \rangle_{ {\rm CE}_{2/\alpha,p}[e^{ \pi i (a -b ) \theta}
     |1 + e^{2 \pi i \theta} |^{a + b }]} =    P_\kappa^{(\alpha)}((1)^p) {[-b]_\kappa^{(\alpha)} \over [1+a+(p-1)/\alpha]_\kappa^{(\alpha)}}.
   \end{equation}
   Note that for the LHS to be well defined, it is required that Re$\,(a+b) > - 1$.
 Multiplying through by $ {(t \alpha)^{| \kappa |} [r]_\kappa^{(\alpha)}/  h_\kappa'}$ and summing over $\kappa$, we see that on the LHS
 we encounter $ {\vphantom{F}}_1^{\mathstrut} F_0^{(\alpha)}(r; \underline{\: \:} ; -\mathbf z)$, which we know can be summed according
 to (\ref{2.4E}), while  the RHS is an example of   the generalised hypergeometric function $ {\vphantom{F}}_2^{\mathstrut} F_1^{(\alpha)}$
 of $p$ variables all equal to $t$. Explicitly, it follows from (\ref{2.4D}) that  \cite[Eq.~(13.11)]{Fo10}
 \begin{equation}\label{2.4F}   
 \Big \langle   \prod_{l=1}^p(1 + t e^{2 \pi i \theta_l})^{-r} \Big \rangle_{{\rm CE}_{2/\alpha,p}[e^{ \pi i (a -b ) \theta}
     |1 + e^{2 \pi i \theta} |^{a + b }]}
  = {\vphantom{F}}_2^{\mathstrut} F_1^{(\alpha)}\Big (r,
-b; {1 \over \alpha} (p-1) +a + 1;
(t)^p \Big ). 
 \end{equation}  
As with (\ref{2.4Dj}), for this to be well defined, it is required that Re$\,(a+b) > - 1$. However, comparison with 
$  {\vphantom{F}}_2^{\mathstrut} F_1^{(\beta/2)}$ in (\ref{2.4Ea+}) shows that for this case of interest, the latter
requirement is violated.

To overcome this circumstance, we can use the fact that $ {\vphantom{F}}_2^{\mathstrut} F_1^{(\alpha)}$ satisfies
an analogue of the classical Euler transformation (\ref{2.4E+}) \cite{Ya92}
 \begin{equation}\label{2.4H}   
{\vphantom{F}}_2^{\mathstrut} F_1^{(\alpha)}(a,b;c;t_1,\dots,t_p)  = 
\prod_{j=1}^p (1 - t_j)^{-a} \, {}_2^{} F_1^{(\alpha)}
\Big (a, c - b;
c; - {t_1 \over 1 - t_1}, \dots, - {t_p \over 1 - t_p} \Big );
 \end{equation}  
cf.~(\ref{2.4E+}). Applying (\ref{2.4H}) in  (\ref{2.4Ea+}) gives rise to a ${\vphantom{F}}_2^{\mathstrut} F_1^{(\alpha)}$ for
which (\ref{2.4F}) holds, and  (\ref{2.4Di}) results (the two stated dualities result by utilising the fact
that ${\vphantom{F}}_2^{\mathstrut} F_1^{(\alpha)}$ is symmetric in its first two arguments).  \hfill $\square$

\subsection{Laguerre and Jacobi $\beta$ ensembles --- negative powers of characteristic polynomials}
As for the case of the positive powers of these ensembles, use will be made of the hypergeometric functions based on
Jack polynomials. As they are analytic at the origin, the first step is to write (considering the Jacobi case for definiteness)
 \begin{equation}\label{2.9a} 
    \Big  \langle    \det ( x \mathbb I_N - H)^{-r}     \Big  \rangle_{ {\rm ME}_{\beta,N}[x^{a_1} (1 - x)^{a_2}]} =
    y^{rN}  \Big  \langle    \det (  \mathbb I_N - y H)^{-r}       \Big  \rangle_{ {\rm ME}_{\beta,N}[x^{a_1} (1 - x)^{a_2}]}, \quad y = {1 \over x},
  \end{equation}   
  so that the average on the RHS is analytic at the origin in the variable $y$. Further progress is possible for $r = \beta p/2$,
  when the later average permits the ${\vphantom{F}}_2^{\mathstrut} F_1^{(\alpha)}$ evaluation \cite{Ka93}, \cite[Eq.~(13.10)]{Fo10}
 \begin{equation}\label{2.10a}   
  \Big  \langle    \det (  \mathbb I_N - y H)^{-\beta p/2}       \Big  \rangle_{ {\rm ME}_{\beta,N}[x^{a_1} (1 - x)^{a_2}]} =
{\vphantom{F}}_2^{\mathstrut} F_1^{(\alpha)} \Big ( {\beta  N \over 2}, {\beta \over 2} (N - 1) + a_1 + 1; \beta(N-1) + a_1 + a_2 + 2; (y)^p \Big ). 
\end{equation}
With this as the starting point, duality formulas for negative powers of characteristic polynomials in the Laguerre and Jacobi $\beta$ ensemble
cases can now be deduced.

\begin{prop}
Require that $y < 0$. We have
 \begin{equation}\label{2.11a}  
  \Big  \langle    \det (  \mathbb I_N - y H)^{-\beta p/2}       \Big  \rangle_{ {\rm ME}_{\beta,N}[x^{a_1} (1 - x)^{a_2}]} =  
 \Big  \langle    \det (  \mathbb I_p - y H)^{-\beta N/2}       \Big  \rangle_{ {\rm ME}_{\beta,p}[x^{a_1} (1 - x)^{a_2}]} 
\end{equation}
and
  \begin{equation}\label{2.11b}  
  \Big  \langle    \det (  \mathbb I_N - y H)^{-\beta p/2}       \Big  \rangle_{ {\rm ME}_{\beta,N}[x^{a} e^{-x}]} =  
 \Big  \langle    \det (  \mathbb I_p - y H)^{-\beta N/2}       \Big  \rangle_{ {\rm ME}_{\beta,p}[x^{a} e^{-x}]}. 
\end{equation}
\end{prop}

\begin{proof}
The duality (\ref{2.11b}) follows from (\ref{2.11a}) by a suitable scaling and limit process. In relation to (\ref{2.11a}), we
read off from \cite[Eq.~(13.12)]{Fo10} a matrix average formula for the ${\vphantom{F}}_2^{\mathstrut} F_1^{(\alpha)}$ function
in (\ref{2.10a}) distinct from that given therein, which implies the stated result.
\end{proof}

\subsection{Gaussian and Laguerre $\beta$ ensembles with a source}\label{S3.3}
A generalisation of the classical Gaussian ensembles, which can be traced back to Dyson \cite{Dy62b}, is to introduce a parameter
$\tau \ge 0$ by defining $G = |1 - e^{-2 \tau} |^{1/2} H + e^{-\tau} H^{(0)}$. Here $H$ is a member of a classical Gaussian ensemble
and $H^{(0)}$ is a fixed matrix of the same class (e.g.~real symmetric for $H \in {\rm GOE}$). The joint element PDF is then
proportional to
 \begin{equation}\label{2.5}
e^{-(\beta/2) {\rm Tr} \, (G - e^{-\tau} H^{(0)})^2/(1 - e^{-2\tau})},
\end{equation}
showing that the model of independent Gaussian elements is maintained, but now the means are no longer zero, but rather determined
by $ H^{(0)}$ (referred to as an external source \cite{BH16}). Changing variables to the eigenvalues and eigenvectors, there is no
longer a separation of the two classes of variables, since (\ref{2.5}) depends on both. Rather integration over the eigenvectors gives
for the eigenvalue PDF, up to proportionality
 \begin{equation}\label{2.5a}
 \prod_{1 \le j < k \le N} | \lambda_k - \lambda_j|^\beta e^{-\tilde{\beta} \sum_{j=1}^N \lambda_j^2 - \tilde{\beta}  t^2 \sum_{j=1}^N \mu_j^2}
 \int e^{2 \tilde{\beta} t {\rm Tr}(U \Lambda U^\dagger L)} d^{\rm H} U.
\end{equation} 
Here $\{ \lambda_j \}$ are the eigenvalues of $G$, $\{ \mu_j \}$ are the eigenvalues of $H^{(0)}$, $\Lambda = {\rm diag}(\lambda_1,\dots,\lambda_N)$,
$L = {\rm diag}(\mu_1,\dots,\mu_N)$, $\tilde{\beta} = \beta/(2(1 - e^{-2 \tau}))$ and $t = e^{-\tau}$. Recalling (\ref{6.2g2}) gives the rewrite of
(\ref{2.5a}) 
\begin{equation}\label{2.5a1}
 \prod_{1 \le j < k \le N} | \lambda_k - \lambda_j|^\beta e^{-\tilde{\beta} \sum_{j=1}^N \lambda_j^2 - \tilde{\beta}  t^2 \sum_{j=1}^N \mu_j^2}
  {\vphantom{\mathcal F}}_0^{\mathstrut} \mathcal F_0^{(2/\beta)}(2 \tilde{\beta} t  \boldsymbol \lambda; \boldsymbol \mu).
  \end{equation} 
  With $\tilde{\beta} = t = 1$ (these parameters are effectively scales associated with the eigenvalues), we will denote the
  PDF corresponding to (\ref{2.5a1}) by ME${}_{\beta,N}[e^{-\lambda^2};\boldsymbol \mu]$. The case $\boldsymbol \mu = \mathbf 0$
  reduces back to ME${}_{\beta,N}[e^{-\lambda^2}]$. In the case $\beta = 2$, a generalisation of (\ref{2.X1}) to involve this
  matrix ensemble with a source was given by Br\'ezin and Hikami \cite{BH01}, which in turn was generalised to all $\beta > 0$
  by Desrosiers \cite{De09}.
  
  \begin{prop}\label{P3.5}
  We have
  \begin{equation}\label{2.5b}
 \Big ({1 \over i} {\sqrt{2 \over \beta}}  \Big  )^{pN/2} \Big \langle  \prod_{l=1}^p \det \Big ( i \sqrt{\beta \over 2} \nu_l\mathbb I_N - H \Big ) \Big \rangle_{{\rm ME_{\beta,N}[e^{-\lambda^2};\boldsymbol \mu}]} =
\Big  \langle  \prod_{j=1}^N \det \Big ( i \sqrt{2 \over \beta} \mu_j  \mathbb I_p - H \Big ) \Big \rangle_{{\rm ME_{{4 \over \beta},p}[e^{-\lambda^2};\boldsymbol \nu}]},   
    \end{equation} 
    with the case $\boldsymbol \nu = \mathbf 0$, $\beta = 2$ being equivalent to (\ref{2.X1}). Also
 \begin{equation}\label{2.5bD}
 \Big \langle  \prod_{l=1}^p \det \Big (  \nu_l\mathbb I_N - H \Big )^{-\beta/2} \Big \rangle_{{\rm ME_{\beta,N}[e^{-\lambda^2};\boldsymbol \mu}]} =
\Big  \langle  \prod_{j=1}^N \det \Big (  \mu_j  \mathbb I_p - H \Big )^{-\beta/2}  \Big \rangle_{{\rm ME}_{ \beta,p}[e^{-\lambda^2};\boldsymbol \nu]},   
    \end{equation}    
    which with $\{ \nu_l \}$ and $\{ \mu_j \}$ real and distinct requires $\beta < 2$ to be well defined, and for consistent branches of the fractional power
    function to be chosen.
    
  \end{prop}
  
  \begin{proof}
 We will consider only (\ref{2.5b}), as the working required to establish (\ref{2.5bD}) is similar. Introduce the operator
   \begin{equation}\label{2.5c} 
   \Delta_{N,\mathbf z}^{(\alpha)} := \sum_{j=1}^N {\partial^2 \over \partial z_j^2} + {2 \over \alpha} \sum_{1 \le j < k \le N}
   {1 \over z_j - z_k} \Big ( {\partial \over \partial z_j} -  {\partial \over \partial z_k}  \Big ).
 \end{equation}    
  A formula following from the development of the theory of generalised Hermite polynomials based on Jack polynomials \cite[Eq.~(3.22)]{BF97a}
  (see also \cite{Ro98})
  gives
    \begin{equation}\label{2.5d}  
  \langle f(i \boldsymbol \lambda) \rangle_{{\rm ME}_{2/\alpha,N}[e^{-\lambda^2}; \mathbf z]} = e^{- {1 \over 4} \Delta_{N,\mathbf z}^{(\alpha)} } f(\mathbf z) \Big |_{\mathbf z \mapsto i \mathbf z}.
 \end{equation}      
 
 Choose 
   \begin{equation}\label{2.5e}  
   f(\mathbf z) =  f(\mathbf z; \boldsymbol \nu) =  \prod_{l=1}^p \prod_{k=1}^N \Big (   \nu_l +   \sqrt{\alpha} z_k \Big ).
 \end{equation}  
 One observes $ f(i \mathbf z; \boldsymbol \nu)    = i^{pN}  f( \mathbf z; -i\boldsymbol \nu)$ and moreover, via a direct calculation, that 
    \begin{equation}\label{2.5f} 
   \Delta_{N,\boldsymbol \mu}^{(\alpha)} f(\boldsymbol \mu;   \boldsymbol \nu)  \Big |_{\boldsymbol \mu \mapsto i \boldsymbol \mu} 
   = 
   i^{pN}    \Delta_{p,\boldsymbol \nu}^{(1/\alpha)} f(\boldsymbol \mu;   -\boldsymbol \nu) \Big |_{\boldsymbol \nu \mapsto i \boldsymbol \nu}. 
 \end{equation}  
 Consequently
 \begin{equation}\label{2.5g}   
  e^{- {1 \over 4} \Delta_{N,\boldsymbol \mu}^{(\alpha)} } f(\boldsymbol \mu;   \boldsymbol \nu)    \Big |_{\boldsymbol \mu \mapsto i \boldsymbol \mu}    =
 i^{pN}   e^{- {1 \over 4} \Delta_{p,\boldsymbol \nu}^{(1/\alpha)} } f(\boldsymbol \mu;   -\boldsymbol \nu)     \Big |_{\boldsymbol \nu \mapsto i \boldsymbol \nu}.
  \end{equation} 
  Comparing (\ref{2.5d}) and (\ref{2.5g}) shows
 \begin{equation}\label{2.5h}   
   \langle f(i \boldsymbol \lambda, \boldsymbol \nu) \rangle_{{\rm ME}_{1/\alpha,N}[e^{-\lambda^2}; \boldsymbol \mu]} = i^{pN} 
     \langle f( \boldsymbol \lambda, -i \boldsymbol \mu) \rangle_{{\rm ME}_{\alpha,p}[e^{-\lambda^2}; \boldsymbol \nu]}, 
  \end{equation}   
  which is (\ref{2.5b}). 
  
  \end{proof} 
 
 \begin{remark}
 In the case $p=1$, $\beta = 2$, (\ref{2.5b}) can be rewritten to read
  \begin{equation}\label{2.5i} 
  \langle \det ( \mu \mathbb I_N - (H + H_0) ) \rangle_{{\rm GUE}_N} = i^{pN/2}
 \Big  \langle\prod_{j=1}^N (  \mu- \nu_j + i x ) \Big \rangle_{x \in \mathcal N(0,1/{2})},
  \end{equation}  
  where $\{\nu_j\}$ are the eigenvalues of $H_0$. In fact this holds in the more general case
  when $H$ is a Wigner matrix as specified in (\ref{1.0}) \cite{Fo13}, with the latter corresponding
  to the case that $H_0$ is the zero matrix.
  \end{remark}
  
  The chiral Gaussian random matrices based on the structure (\ref{3.1a}) admit a parametric extension involving
  a source matrix by replacing $Z$ therein by $Y = |1 - e^{-2 \tau} |^{1/2} Z + e^{-\tau} Z^{(0)}$. Of interest
  are the squared 
  non-zero singular
  values $Y$ (i.e.~the eigenvalues $\{y_j\}_{j=1}^N$ of $Y^\dagger Y$). 
  We know from \cite[Eq.~(11.105)]{Fo10}
  that the PDF for these singular values is given by,
 up to proportionality
 \begin{equation}\label{2.5a+}
\prod_{j=1}^N y_j^{\beta a } \prod_{1 \le j < k \le N} | y_k - y_j|^\beta e^{-\tilde{\beta} \sum_{j=1}^N y_j - \tilde{\beta}  t^2 \sum_{j=1}^N \mu_j}
  {\vphantom{\mathcal F}}_0^{\mathstrut} \mathcal F_1^{(2/\beta)}(\beta (a+N-1)/2+1;  \tilde{\beta} t  \mathbf y; \boldsymbol \mu),
\end{equation}  
where $a= n-N+1-2/\beta$.
 Here, with the components of $\mathbf z$ given by $\{z_j\}_{j=1}^N$, we have denoted by $\mathbf z^2$  the vector with components $\{z_j^2\}_{j=1}^N$,
 and $\{\mu_j\}_{j=1}^N$ are the eigenvalues of $(Z^{(0)})^\dagger Z^{(0)}$.
 Further, with $\Lambda$ ($L$) the $n \times N$ matrix with non-zero entries on the diagonal only, $\{y_1,\dots,y_N\}$ ($\{\mu_1,\dots,\mu_N\}$), and with $\beta = 1,2$ or 4 the function
 $ {\vphantom{\mathcal F}}_0^{\mathstrut} \mathcal F_1^{(2/\beta)}$ is defined as the matrix integral relating to the singular value decomposition of $Z$
 \begin{equation}\label{2.5b+Z} 
 \int d^{\rm H} U \int d^{\rm H} V \, e^{\tilde{\beta} t {\rm Tr} (V \Lambda^\dagger U^\dagger L + L^\dagger U \Lambda V^\dagger)}.
\end{equation}  
Here the unitary matrices $U,V$ are of size $N \times N$ and $n \times n$ respectively, from the orthogonal ($\beta = 1$), unitary ($\beta = 2$) and
symplectic ($\beta = 4$) matrix groups respectively.
Important is that with  $\tilde{\beta} = t = 1$ (this is just for convenience as done in relation to (\ref{2.5})), and $\alpha = 2/\beta$, the latter admits the Jack polynomial
expansion (see e.g.~\cite[\S 4]{Fo13})
\begin{equation}\label{6.3g1}
{\vphantom{\mathcal F}}_0^{\mathstrut}\mathcal F_1^{(\alpha)}(u; \mathbf y;\boldsymbol \mu) = \sum_{\kappa} {\alpha^{| \kappa |}\over  [u]_\kappa^{(\alpha)}  h_\kappa'} 
{P_\kappa^{(\alpha)}(\mathbf y) P_\kappa^{(\alpha)}(\boldsymbol \mu) \over
P_\kappa^{(\alpha)}((1)^N)
},
\end{equation}
with $u = (a+N-1)/\alpha + 1 = \beta n/2$, thus allowing for a $\beta$ generalisation. This function of two sets of variables relates
to the $p=0$, $q=1$ Jack polynomial based hypergeometric function defined in (\ref{3.40}) according to
\begin{equation}\label{6.3z}
{\vphantom{\mathcal F}}_0^{\mathstrut}\mathcal F_1^{(\alpha)}(u; \mathbf y;(1)^N) = {\vphantom{ F}}_0^{\mathstrut} F_1^{(\alpha)}(u; \mathbf y).
\end{equation}
In the case $\boldsymbol \mu = \mathbf 0$, 
${\vphantom{\mathcal F}}_0^{\mathstrut}\mathcal F_1^{(\alpha)}$ equals unity and we recognise (\ref{2.5a+}) as
ME${}_{\beta,N}[y^a e^{-y}]$; the generalised PDF with source variables $\boldsymbol \mu$ will be denoted
ME${}_{\beta,N}[y^a e^{-y};\boldsymbol \mu]$. In both cases we have assumed $\tilde{\beta} = t = 1$ and it is regarded as implicit
that the eigenvalues can only take on positive values.
Desrosiers \cite{De09} has deduced an analogue of the duality (\ref{2.5b}) for this source generalised ensemble.

 \begin{prop}
  We have
  \begin{equation}\label{2.5b+V}
\Big \langle  \prod_{l=1}^p \det \Big ( \nu_l\mathbb I_N + {2 \over \beta} H \Big ) \Big \rangle_{{\rm ME}_{\beta,N}[\lambda^a e^{-\lambda};\boldsymbol \mu]} =
  \Big ({2 \over \beta}  \Big  )^{pN} 
\Big  \langle  \prod_{j=1}^N \det \Big ( \mu_j  \mathbb I_p + {\beta \over 2} H \Big ) \Big \rangle_{{\rm ME}_{{4 \over \beta},p}[\lambda^{a'} e^{-\lambda};\boldsymbol \nu]},  
    \end{equation} 
where $a' = {2 \over \beta} (a+1) - 1$. Also
\begin{equation}\label{2.5b+D}
\Big \langle  \prod_{l=1}^p \det \Big ( \nu_l\mathbb I_N +  H \Big )^{-\beta/2} \Big \rangle_{{\rm ME}_{\beta,N}[\lambda^a e^{-\lambda};\boldsymbol \mu]} =
\Big  \langle  \prod_{j=1}^N \det \Big ( \mu_j  \mathbb I_p +  H \Big )^{-\beta/2}  \Big \rangle_{{\rm ME}_{\beta,p}[\lambda^{a'} e^{-\lambda};\boldsymbol \nu]}. 
    \end{equation} 
\end{prop}

\begin{proof} (Sketch.)
A strategy analogous to that used in the proof of Proposition \ref{P3.5} applies.
Taking the place of (\ref{2.5c}) is the operator
   \begin{equation}\label{2.5cD} 
   D_{N,a,\mathbf z}^{(\alpha)} := \sum_{j=1}^N z_j {\partial^2 \over \partial z_j^2} + {2 \over \alpha} \sum_{1 \le j < k \le N}
   {1 \over z_j - z_k} \Big ( z_j {\partial \over \partial z_j} - z_k {\partial \over \partial z_k}  \Big ) + (a+1) \sum_{j=1}^N z_j {\partial \over \partial z_j},
 \end{equation}    
  while taking the place of (\ref{2.5d}) is the operator based integration formula \cite[Eq.~(3.22)]{BF97a}
    \begin{equation}\label{2.5dD}  
  \langle f(- \boldsymbol \lambda) \rangle_{{\rm ME}_{2/\alpha,N}[\lambda^a e^{-\lambda}; \mathbf z]} = e^{- D_{N,a,\mathbf z}^{(\alpha)} } f(\mathbf z).
 \end{equation}      
 In relation to (\ref{2.5b+V}) the crucial observation from here is that with the choice
 \begin{equation}\label{2.5eD}  
   f(\mathbf z) =  f(\mathbf z; \boldsymbol \nu) =  \prod_{l=1}^p \prod_{k=1}^N \Big (   \nu_l -   \alpha z_k \Big ),
 \end{equation}  
one has
    \begin{equation}\label{2.5f+} 
   D_{N,a,\boldsymbol \mu}^{(\alpha)} f(\boldsymbol \mu;   \boldsymbol \nu) 
   = 
       D_{p,a',\boldsymbol \nu}^{(\alpha)} f(\boldsymbol \mu;   \boldsymbol \nu).
 \end{equation}  
 
\end{proof}

\begin{remark} $ $ \\
1.~
Let $M$ be an $n \times N$ ($n \ge N$) standard Gaussian complex matrix. Let $M_0$ be of the same size as $M$
and have singular values $\{\mu_l\}_{l=1,\dots,N}$. It follows from (\ref{2.5b+V}) with $p=1$, $\beta = 2$, that
  \begin{equation}\label{2.5i+} 
  \langle \det ( \mu \mathbb I_N - (M + M_0)^\dagger (M + M_0)  ) \rangle = 
 \Big  \langle\prod_{j=1}^N (  \mu- |\mathbf x + \boldsymbol \mu |^2 ) \Big \rangle_{x_j \in \mathcal N(0,1/2)},
  \end{equation}  
  where $\mathbf x = (x_1,\dots,x_N)$.
 In \cite[Cor.~3]{Fo13}, a result equivalent to this was shown to hold for all $n \times N$ random matrices $M$ with
 independent entries of mean zero and averaged modulus squared equal to unity. \\
 2.~Setting $\nu_1 = \cdots = \nu_p = -\beta x/2$ and $\boldsymbol \mu = 0$ in 
 (\ref{2.5b+V}) gives on the LHS, up to a simple factor, the same average as on the LHS of (\ref{2.4Ci}).
 However the matrix averages on the RHS are different. \\
 3.~In \cite{FGS18} the duality (\ref{2.5b+D}) with $\beta = 2$ is considered from a different perspective. In the case
 that all the components of $\boldsymbol \nu$ are equal, it is used therein the study the large $N$ asymptotics.
  \end{remark}
  
  \subsection{Circular and circular Jacobi $\beta$ ensembles}\label{S3.6}
  The original circular $\beta$ ensemble of Dyson \cite{Dy62} is the case $w=1$ of (\ref{2.1X}). The name given
  to (\ref{2.1X}) with the more general weight (\ref{2.6a}) is the (generalised) circular Jacobi $\beta$ ensemble
  (the restriction to $a_1 = a_2$ was the original meaning to this ensemble  \cite[\S 3.9]{Fo10}). It turns out that there
  are special features of averaged (squared modulus) characteristic polynomials in the case $w=1$ relative
  to  (\ref{2.6a}), so warranting a separate discussion.
  
  First, in the case $w=1$, for $|z|=1$ we observe that $\langle \prod_{l=1}^N |z + e^{2 \pi i \theta_l}|^{2q} \rangle_{{\rm CE}_{\beta,N}[1]}$ is
  independent of $z$ (we write $z + e^{2 \pi i \theta_l}$ rather than $z - e^{2 \pi i \theta_l}$ to account for the range being $-1/2 < \theta_l < 1/2$ in
  (\ref{2.1X}) rather than $0 < \theta_l < 1$ as is conventional for CE${}_{\beta,N}[1]$.)
  This is a consequence of rotational invariance. A natural generalisation is to replace each factor
  $ |z + e^{2 \pi i \theta_l}|^{2q} $ by $ |z_1 + e^{2 \pi i \theta_l}|^{2 q}  |z_2 + e^{2 \pi i \theta_l}|^{2\mu} $. Rotational invariance then gives
    \begin{equation}\label{3.8a}
 \Big  \langle \prod_{l=1}^N |z_1 + e^{2 \pi i \theta_l}|^{2q}|z_2 + e^{2 \pi i \theta_l}|^{2\mu}  \Big   \rangle_{{\rm CE}_{\beta,N}[1]}  \propto
 \Big  \langle \prod_{l=1}^N |z + e^{2 \pi i \theta_l}|^{2q}  \Big  \rangle_{{\rm CE}_{\beta,N}[|1 + e^{2 \pi i \theta}|^{2\mu}]}, \quad z := {z_1 \over z_2 }.
  \end{equation}  
  The working which lead to (\ref{2.6e}) shows that  for $q$ a positive integer, upon simple manipulation, the absolute value signs can be removed, giving that (\ref{3.8a}) is
  proportional  to
   \begin{equation}\label{3.8b}
  z^{qN}    
  \Big  \langle \prod_{l=1}^N (1 + \bar{z} e^{2 \pi i \theta_l})^{2q}  \Big  \rangle_{{\rm CE}_{\beta,N}[e^{-2 \pi i \theta q}|1 + e^{2 \pi i \theta}|^{2\mu}]}.
  \end{equation}  
  The average in (\ref{3.8a}) in the case $2q=2\mu=\beta$ relates to the two-point correlation function of Dyson's circular $\beta$ ensemble.
  For $\beta$ even an evaluation in terms of the Jack polynomial based hypergeometric function 
 $ {\vphantom{F}}_2^{\mathstrut} F_1^{(\alpha)} $ was given in the 1992 work \cite{Fo92j}, and a duality identity obtained two years
 later in  \cite{Fo94j} --- these findings are summarised in \cite[\S 13.2.1]{Fo10}. The same considerations can be used to obtain
 a duality identity for the average in (\ref{3.8a}).
 
 \begin{prop}\label{P3.7}
 Let $q$ be a positive integer, let $\mu \ge q$ be real, and let $z = z_1/z_2$ with $|z_1| = | z_2| = 1$. We have
   \begin{multline}\label{3.8c} 
  \Big  \langle \prod_{l=1}^N |z_1 + e^{2 \pi i \theta_l}|^{2q}|z_2 + e^{2 \pi i \theta_l}|^{2\mu}  \Big   \rangle_{{\rm CE}_{\beta,N}[1]} \\ 
  \propto   z^{qN}  \Big \langle   \prod_{l=1}^{2q}(1 + (1 - \bar{z}) e^{2 \pi i \theta_l})^{N} \Big \rangle_{{\rm CE}_{4/\beta,2q}[e^{ \pi i (a -b ) \theta}
     |1 + e^{2 \pi i \theta} |^{a + b }] \Big |_{\substack{a+b = -1 +2(\mu - q + 1)/\beta \\ \: \: a-b= -1 + 2(3 \mu + q + 1)/\beta}} }\\
     \propto
     z^{qN}  \Big \langle   \prod_{l=1}^{2q}(1 -  (1 - \bar{z})  x_l)^{N} \Big \rangle_{{\rm ME}_{4/\beta,2q}[x^{a} (1 - x)^{a}] \Big |_{a = -1+2(\mu - q + 1)/\beta}}.
  \end{multline}      
 \end{prop}
 
 \begin{proof}
 As noted, it suffices to consider the average in (\ref{3.8b}). Reading off from \cite[Eq.~(13.6)]{Fo10} we have that this average
 is equal to
    \begin{equation}\label{3.8d} 
 {\vphantom{F}}_2^{\mathstrut} F_1^{(\beta/2)} \Big (-N,{2 \over \beta} (\mu+q);-(N-1)-{2 \over \beta}(\mu - q + 1); (\bar{z})^{2q} \Big ).   
   \end{equation}  
  Next we would like the apply (\ref{2.4F}) to express this as an average over $2q$ variables. However, this is not immediately
  possible as the requirement that Re$\, (a+b) > -1$ therein is violated. To overcome this, we make use of the transformation
  identity \cite[Prop.~13.1.7]{Fo10}
   \begin{equation}\label{3.8e}  
  {\vphantom{F}}_2^{\mathstrut} F_1^{(\beta/2)} (-N,b;c;t_1,\dots,t_m) \propto 
  {\vphantom{F}}_2^{\mathstrut} F_1^{(\beta/2)} \Big (-N,b;-N+b+1+2(m-1)/\beta - c;1-t_1,\dots,1-t_m \Big ) .
   \end{equation} 
   This shows (\ref{3.8d}) can be replaced by
  \begin{equation}\label{3.8f} 
   {\vphantom{F}}_2^{\mathstrut} F_1^{(\beta/2)} \Big (-N,{2 \over \beta} (\mu+q);{4 \over \beta}(\mu + q) ; (v)^{2q} \Big )\Big |_{v = 1 - \bar{z}}.   
   \end{equation} 
   In the notation of  (\ref{2.4F}) we now have $a+b = -1 + 2(\mu - q + 1)/\beta$, which is greater than $-1$ under the assumption
   $\mu \ge q$, allowing us to deduce 
    the first line on the RHS.
    
    In relation to the second line of the RHS, we require knowledge of a companion to (\ref{2.4F}), which expresses
the   $ {\vphantom{F}}_2^{\mathstrut} F_1^{(\alpha)} $   function of $p$ variables all equal as an average
over the Jacobi $\beta$ ensemble \cite[Eq.~(13.12)]{Fo10}
\begin{equation}\label{2.4F+}   
 \Big \langle   \prod_{l=1}^p(1 - tx_l)^{-r} \Big \rangle_{{\rm ME}_{2/\alpha,p}[x^{a_1} (1 - x)^{a_2}]}
  = {\vphantom{F}}_2^{\mathstrut} F_1^{(\alpha)}\Big (r, {1 \over \alpha} (p-1) +a_1+1;
    {2 \over \alpha} (p-1) +a_1+a_2+2 ;
(t)^p \Big ). 
 \end{equation}  
 Comparing with (\ref{3.8f}) gives the stated result.      
 \end{proof}
 
 We next consider
  \begin{equation}\label{3.8g} 
  I_{N,2q}(z) := \Big \langle \prod_{l=1}^N | z - e^{2 \pi i \theta_l} |^{2q} \Big \rangle_{{\rm CE}_{\beta,N}[1]}
  \end{equation}  
  in the case that $|z| < 1$. We know from \cite{FK04} that in this circumstance we have the rewrite
  \begin{equation}\label{3.8g+} 
  I_{N,2q}(z) = \Big \langle \prod_{l=1}^N (1+ |z|^2  e^{2 \pi i \theta_l} )^{q} \Big \rangle_{{\rm CE}_{\beta,N}[e^{-\pi i \theta q} |1 + e^{2 \pi i \theta} |^q]}
  \end{equation}  
  (cf.~(\ref{3.8b})), valid for $q > -1$. Using this as the starting point, with $q$ a positive integer, a minor modification of the working
  of the proof of Proposition \ref{P3.7} allows for the deduction of a duality formula.

 \begin{prop}
 Let $q$ be a positive integer, let $|z| < 1$, and define $  I_{N,2q}(z)$ as in (\ref{3.8g}). We have
   \begin{multline}\label{3.8c+} 
 I_{N,2q}(z) 
  \propto   \Big \langle   \prod_{l=1}^{q}(1 + (1 - |z|^2) e^{2 \pi i \theta_l})^{N} \Big \rangle_{{\rm CE}_{4/\beta,q}[e^{ \pi i (a -b ) \theta}
     |1 + e^{2 \pi i \theta} |^{a + b }] \Big |_{\substack{a+b = -1 +2 /\beta \\ \: \: a-b= -1 + 2(2 q + 1)/\beta}}} \\
     \propto
    \Big \langle   \prod_{l=1}^{q}(1 -  (1 - |z|^2)  x_l)^{N} \Big \rangle_{{\rm ME}_{4/\beta,q}[x^{a} (1 - x)^{a}] \Big |_{a = -1+2 /\beta}}.
  \end{multline}      
 \end{prop} 
 
 Neither (\ref{3.8a}) nor (\ref{3.8b}) give rise to dualities for the circular Jacobi ensemble unless $|z|=1$. Then
  \begin{equation}\label{3.8H} 
 \Big  \langle \prod_{l=1}^N |z + e^{2 \pi i \theta_l}|^{2q}  \Big  \rangle_{{\rm CE}_{\beta,N}[e^{ \pi i (a-b) }|1 + e^{2 \pi i \theta}|^{a+b}]} 
\propto   z^{qN}    
  \Big  \langle \prod_{l=1}^N (1 + \bar{z} e^{2 \pi i \theta_l})^{2q}  \Big  \rangle_{{\rm CE}_{\beta,N}[e^{ \pi i \theta(a-b-2  q}|1 + e^{2 \pi i \theta}|^{a+b}]},
  \end{equation}  
  which 
as follows from the working of the  proof of Proposition \ref{P3.7} permits a duality identity (for the special case $2q=\beta$ see \cite{FLT21}).

\begin{prop}
 Let $q$ be a positive integer, let $|z| = 1$, and require that $a \ge q-1$. We have
   \begin{multline}\label{3.8d+} 
  \Big  \langle \prod_{l=1}^N |z + e^{2 \pi i \theta_l}|^{2q}  \Big  \rangle_{{\rm CE}_{\beta,N}[e^{ \pi i (a-b) }|1 + e^{2 \pi i \theta}|^{a+b}]} \\
  \propto   z^{qN}  \Big \langle   \prod_{l=1}^{2q}(1 + (1 - \bar{z}) e^{2 \pi i \theta_l})^{N} \Big \rangle_{{\rm CE}_{4/\beta,q}[e^{ \pi i (a' -b' ) \theta}
     |1 + e^{2 \pi i \theta} |^{a' + b' }] \Big |_{\substack{a'+b' = -1 +2(a-q+1 /\beta \\ \: \: a'-b'= -1 + 2(2 b + a+ q+1)/\beta}}} \\
     \propto  z^{qN} 
    \Big \langle   \prod_{l=1}^{2q}(1 -  (1 - \bar{z} )  x_l)^{N} \Big \rangle_{{\rm ME}_{4/\beta,q}[x^{a_1} (1 - x)^{a_2}] \Big |_{\substack{a_1 = -1+2(b-q+1) /\beta  \\
a_2 = -1+2(a-q+1) /\beta    }}}.
  \end{multline}   
  (The final of these averages also requires   $b \ge q-1$.) 
 \end{prop}

 \section{Interpretations and scaled large $N$ asymptotics}\label{S4}
 \subsection{Wigner semi-circle law}\label{S4.1}
 With $\sigma^2 = 1/(2N)$, and making use of (\ref{2.3c}), we see that (\ref{1.0}) can be written
  \begin{equation}\label{6.1}
  \langle \det ( {x} \mathbb I_N - \tilde{H}) \rangle_{\tilde{H} \in \mathcal{H}_N} = 2^{-N} (\sqrt{2 N})^N H_N(\sqrt{2N} {x}).
   \end{equation}  
   The significance of this scaling is that the Wigner class $\mathcal{H}_N$ then has eigenvalue support on the compact
   interval $(-1,1)$. Thus in the setting that the normalised limiting eigenvalue density, $\rho_{\infty,(1)}(y)$ say, has support
   on a compact interval $I$, it is true that for a large class of $a(y)$ (slow enough decay at infinity), for $N$ large
     \begin{equation}\label{6.1a}
    \Big  \langle e^{\sum_{l=1}^N a(y_l)} \Big \rangle_{\mathcal{H}_N} = e^{N \int_I a(\lambda) \rho_{\infty,(1)}(y) \,  dy + {\rm O}(1)};
   \end{equation} 
   see e.g.~\cite[\S 3.1]{FW17}, \cite{FK07A}, which can be viewed as a law of large numbers. Now the LHS of (\ref{6.1}) is of the
   form of the LHS of (\ref{6.1a}) with $a(y) = a(y;x) = \log(x - y)$. Substituting  (\ref{6.1a}) in  (\ref{6.1}) and
   taking the logarithmic derivative with respect to $x$ then gives that to leading order
      \begin{equation}\label{6.1b}  
 W_1(x):=     \int_{-1}^1 {\rho_{\infty,(1)}(y) \over x - y} \, d y = {1 \over N} {1 \over u(x)} {d \over d x} u(x), \quad u(x) = H_N(\sqrt{2N} x), \: \: x \notin (-1,1).
  \end{equation}  
 Consideration of the second order linear differential equation satisfied by $u(x)$ leads to the conclusion that $W_1(x)$ satisfies the
 quadratic equation $W^2_1/4 - x W_1 + 1 = 0$. Choosing the root such that $W_1(x) \sim 1/x$ as $x \to \infty$, and applying the
 Sokhotski-Plemelj formula $\rho_{\infty,(1)}(y) = {1 \over \pi} {\rm Im}_{\epsilon \to 0^+} W_1(y+i \epsilon )$,  then gives
      \begin{equation}\label{6.1c}   
 \rho_{\infty,(1)}(y)  =  \rho^{\rm W}(y), \qquad         \rho^{\rm W}(y):=   {2 \over \pi} (1 - y^2)^{1/2} \mathbbm 1_{|y| < 1},    
   \end{equation}  
   which is the density function specifying the Wigner semi-circle law; see e.g.~\cite{PS11}.
   
   \subsection{Gaussian $\beta$ ensemble}
   \subsubsection{Global density}
  Generally for the functional form (\ref{2.1})  specifying ME${}_{\beta,N}[w]$, the definition of the eigenvalue density 
  $ \rho_{N,(1)}(\lambda) = \rho_{N,(1)}(\lambda;w)$ as an integral over all but one of the $\lambda_i$
  gives
      \begin{equation}\label{6.1d}   
 \rho_{N,\beta,(1)}(\lambda;w)  \propto w(\lambda)  \Big \langle \prod_{l=1}^{N-1} |\lambda - x_l|^{\beta} \Big \rangle_{{\rm ME}_{\beta,N}[w]}.
   \end{equation} 
   In the case of the   Gaussian $\beta$ ensemble it therefore follows from (\ref{2.X1c}) that  for $\beta$ even
   \begin{equation}\label{6.1e}   
  \rho_{N,(1)}(\lambda;e^{-\beta \lambda^2/2})    \propto     e^{-\beta \lambda^2/2} 
  \Big  \langle   \prod_{j=1}^{N-1}  ( i  x  -  \lambda_j)^{N-1} \Big  \rangle_{{\rm ME}_{4/ \beta,N-1}[e^{-\lambda^2}]}.
   \end{equation}  
  Starting with this duality identity, several scaled asymptotic limits have been analysed using saddle point analysis;
  for a detailed account of this method in the present context see \cite{DL14}.
  One is the direct evaluation of the global density limit \cite[\S 5.3]{BF97a},
   \begin{equation}\label{6.1f}    
\lim_{N \to \infty}        \sqrt{{2 \over N}}     \rho_{N,(1)}(\lambda;e^{-\beta N \lambda^2})    =    \rho^{\rm W}(\lambda).
  \end{equation}  
In \cite{DF06} this result was extended to an asymptotic expansion in which the first $\lfloor \beta/2 \rfloor$ oscillatory terms were
specified. Here we make note of the first of these.

\begin{prop}\label{P4.1}
With $P^{\rm W}(\lambda) := \int_{-1}^\lambda  \rho^{\rm W}(x) \, dx$, for
  large $N$ and $|\lambda| < 1$,
\small    \begin{equation}\label{6.2}
   \sqrt{{2 \over N}}     \rho_{N,(1)}(\lambda;e^{-\beta N \lambda^2})    =    \rho^{\rm W}(\lambda) - {2 \over \pi}
   {\Gamma(1 + 2/\beta) \over (\pi   \rho^{\rm W}(\lambda) )^{6/\beta - 1}} {1 \over N^{2/\beta}}
   \cos \Big (2 \pi N P^{\rm W}(\lambda) + (1 - 2/\beta) \arcsin \lambda \Big ) + \cdots, 
   \end{equation}    
   \normalsize   
where the next oscillatory term is at order $N^{-8/\beta}$, while the next non-oscillatory term is at order $N^{-1}$
($\beta \ne 2$) and order $N^{-2}$ for $\beta = 2$.  
\end{prop}

Analysis analogous to that of \S \ref{S4.1} can be applied to the Gaussian $\beta$ ensemble. This is because up to a scaling the formula
(\ref{6.1}) remains valid,
  \begin{equation}\
  \Big  \langle   \prod_{j=1}^N (x - \lambda_j ) \Big \rangle_{{\rm ME}_{\beta,N}[e^{-\beta \lambda^2/2}]} = 2^{-N} H_N(x),
  \end{equation}  
  as follows from (\ref{2.X1c}) with $p=1$ and the integral form of the Hermite polynomials (\ref{2.3c}). In particular the RHS is
  independent of $\beta$, and so the working of  \S \ref{S4.1} with respect to the global limit can be repeated to reclaim (\ref{6.1c}).

 \subsubsection{Edge  density}
With $w(\lambda) = e^{-\beta \lambda^2/2}$, to leading order the right spectrum edge is at $\lambda = \sqrt{2N}$.
It has been known since \cite{Fo93a} that a well defined limiting state results from the use of the 
so-called soft edge scaling variables $\{x_i\}$, specified by the requirement that $\lambda_i = \sqrt{2N} + x_i/(\sqrt{2} N^{1/6})$
(i.e.~the origin is shifted to the location of the right edge, and the distances are measured on the scale of $1/N^{1/6}$).
In \cite{DF06} a $\beta$ dimensional integral formula was obtained for the soft edge scaling limit of the eigenvalue
density for $\beta$ even, based on a coalescing saddle point analysis applied to (\ref{6.1e}). In \cite{FT19a} this
analysis was extended to give the next two large terms in the large $N$ asymptotic expansion. In fact some structure was
uncovered in this latter study, with the first order correction (occurring at order $N^{-1/3}$) being shown to be related to the leading order asymptotic form
by a derivative operation (see also \cite{FFG06}, \cite[Eq.~(1.35)]{BL24} for the case $\beta = 1$). 
Equivalently, by redefining the soft edge scaling to include a particular constant shift of the scaling variables $x_i$, this correction term can be
cancelled to obtain an optimal $N^{-2/3}$ correction term \cite[\S 5]{FT18}.

\begin{prop}
For large $N$ and $\beta$ even
\begin{equation}\label{4.9}
			\frac{1}{\sqrt{2}N^{1/6}}\rho_{N,(1)}\left( \sqrt{2N}+ 2^{-1/2} N^{-1/6}
			\left(x+  1/2-1/\beta \right); e^{-\beta \lambda^2/2}
			\right) =  \rho_{\infty,(1)}^{\rm soft}(x;\beta) + {\rm O}(N^{-2/3}),
		\end{equation}
		where with 
\begin{equation}\label{eq:K}
	K_{n,\beta}(x)=-\frac{1}{(2\pi i)^n}\int_{-i\infty}^{i\infty}dv_1 \cdots \int_{-i\infty}^{i\infty}dv_n\, \prod_{j=1}^n e^{v_j^3/3-xv_j}\prod_{1\leq k<l\leq n}|v_k-v_l|^{4/\beta},
	\end{equation}
	(a multidimensional extension of the integral form of the Airy function)
	the limiting soft edge density $ \rho_{\infty,(1)}^{\rm soft}(x)$ is specified by
\begin{equation}\label{eq:K1}	
 \rho_{\infty,(1)}^{\rm soft}(x;\beta) =  \frac{1}{2\pi} \left( \frac{4\pi}{\beta} \right)^{\beta/2} \Gamma(1+\beta/2)\prod_{j=1}^{\beta} \frac{\Gamma(1+2/\beta)}{\Gamma(1+2j/\beta)} K_{\beta,\beta}(x).
\end{equation} 
\end{prop}

\begin{remark}\label{R4.1} ${} $ \\
1.~Define $N' = N + {2 - \beta \over 2 \beta}$. An alternative to shifting $x$ as in (\ref{4.9}) is to consider
\begin{equation}\label{4.9a}
			\frac{1}{\sqrt{2}(N')^{1/6}}\rho_{N,\beta}\left( \sqrt{2N'}+ 2^{-1/2}(N')^{-1/6}x
			; e^{-\beta \lambda^2/2}
			\right),
			\end{equation}
			in which the leading order mean and standard deviation are shifted.
In the cases $\beta=1,2$ and 4, results of Bornemann \cite{Bo24a} imply that this quantity has an asymptotic expansion
in powers of $(N')^{-2/3}$. \\
2.~The multidimensional integral form (\ref{eq:K}) is well suited to determining the $|x| \to \infty$ asymptotics of $ \rho_{\infty,(1)}^{\rm soft}(x;\beta) $
\cite{DF06}, \cite[Eq.~(13.68) with factor $1/(2\pi)$ corrected to $1/\pi$ in the first line]{Fo10}.
\end{remark}		

A generalisation of $K_{n,\beta}(x)$ (\ref{eq:K}) is specified by \cite{De09,DL14}
\begin{equation}\label{eq:KA}
	{\rm Ai}_{n}^{(\alpha)}(\mathbf f)=-\frac{1}{(2\pi i)^n}\int_{-i\infty}^{i\infty}dv_1 \cdots \int_{-i\infty}^{i\infty}dv_n\, \prod_{j=1}^n e^{v_j^3/3}
	 {\vphantom{F}}_0^{\mathstrut} \mathcal F_0^{(\alpha)}(\mathbf v;\mathbf f) 
	\prod_{1\leq k<l\leq n}|v_k-v_j|^{2/\alpha},
	\end{equation}
(a further generalisation to ${\rm Ai}_{n,m}^{(\alpha)}(\mathbf f, \mathbf s)$, where $\mathbf s$ is a second set of variables ($m$ in number) has also
shown itself \cite{DL15}.) We have from \cite[Prop.~10]{De09} that this results from a soft edge scaling limit of (\ref{2.5b}) with $\boldsymbol \mu = \mathbf 0$.
Note that (\ref{eq:K}) corresponds to the case that all the entries of $\mathbf f$ are equal.
The soft edge scaling limit in  the case $\beta = 2$  (i.e.~(\ref{2.X1})) has relevance to the computation of a certain class of algebraic-geometric
quantities known as intersection numbers
\cite[\S 7.2]{BH16}.

As already remarked, (\ref{eq:K}) (and (\ref{eq:KA})) can be considered as generalisations of the Airy function.
We remark that the asymptotic study of the duality (\ref{2.5b}) in the case that the source $\boldsymbol \mu$ consists of of an even number
of variables, half of which take on the value $a$, and the other half take on the value $-a$, there is a scaling limit which leads
to multidimensional generalisations of the Pearcey integral \cite{BH95}, \cite{FL21}.

\subsubsection{Moments of the characteristic polynomial}		

Consider next
 \begin{equation}\label{6.1A}   
 \Big \langle \prod_{l=1}^{N} |\lambda - x_l|^{p} \Big \rangle_{{\rm ME}_{\beta,N}[e^{-\beta N x^2}]},
   \end{equation} 
   which corresponds to a global scaling of the power of the characteristic polynomial in the Gaussian $\beta$ ensemble.
 Motivation to study the large $N$ form of (\ref{6.1A}) in the case $\beta = 2$ corresponding to the GUE came from the
 application of random matrix ensembles  with a unitary symmetry to the modelling of the Riemann zeta function of the
 critical line due to Keating and Snaith \cite{KS00a}. The required analysis was undertaken by Br\'ezin and Hikami
 \cite{BH00} (see also \cite[Theorem 5.4 with $\lambda \mapsto 2 \lambda, \, 2 \rho(2\lambda) \mapsto \rho^{\rm W}(x)$]{BH16}, \cite[Eq.~(21) with $m=1$, after minor correction]{Ga05},
 \cite[Th.~1 with $m=1$]{Kr07},
 \cite[Eq.~(1.13) with $m=1$]{Ch19}, \cite[expanded in terms of $\lambda$]{JKM23}).
 
 \begin{prop}
 For large $N$ and $p$ a positive integer
  \begin{equation}\label{6.1B}   
\Big \langle \prod_{l=1}^{N} |\lambda - x_l|^{2p} \Big \rangle_{{\rm ME}_{2,N}[e^{-2 N x^2}]} =
( \pi N \rho^{\rm W}(\lambda))^{p^2} e^{2Np( \lambda^2 - 1/2-\log \, 2)} \prod_{l=0}^{p-1} {l! \over (p + l)!} \Big ( 1 + {\rm O} \Big ( {1 \over N}  \Big )\Big ).
   \end{equation} 
   \end{prop}

After noting the identity $ \prod_{l=0}^{p-1} {l! \over (p + l)!}  = (G(1+p))^2/G(1+2p)$, where $G(z)$ denotes the Barnes G-function, this
asymptotic expression was proved in \cite{Kr07} to be valid for general Re$(2p) >  -1$ (now with error term O$(\log N/N)$). In keeping
with (\ref{6.1a}) we expect  (\ref{6.1B}) to exhibit the leading large $N$ asymptotic form
  \begin{equation}\label{6.1C} 
  \log \Big \langle \prod_{l=1}^{N} |\lambda - x_l|^{2p} \Big \rangle_{{\rm ME}_{2,N}[e^{-2 N x^2}]} = 2p N \int_{-1}^1  \log | \lambda - x| \rho^{\rm W}(x) \, dx + \cdots,
    \end{equation} 
which from the integral evaluation $ \int_{-1}^1  \log | \lambda - x| \rho^{\rm W}(x) \, dx = \lambda^2 -{1 \over 2} - \log 2$ is indeed true.
One sees from (\ref{6.1B}) that the next term in (\ref{6.1C}) is $p^2 \log N$.  In the case of the circular ensemble for $\beta = 2$, when averages relate to
Toeplitz determinants according to (\ref{2.3B}), the analogous expansion is a modification of the Szeg\H{o} asymptotic formula known as
Fisher-Hartwig asymptotics \cite{DIK11}. Another point to note is that with $p=1$ (\ref{6.1B}) is
consistent with (\ref{6.1d}) with $\beta = 2$ and (\ref{6.1f}).

For a large class of GUE linear statistics $\sum_{l=1}^N a(\lambda_l)$, the asymptotic formula (\ref{6.1a}) has the generalisation that for large $N$
and  (continuous) $k$ small enough
  \begin{equation}\label{6.1D}
  \log \Big \langle e^{k \sum_{l=1}^N a(x_l)} \Big \rangle = k N \int_I a(x) \rho_{\infty,(1)}(x) \, d x + {k^2 \over 2} \sum_{n=1}^\infty n a_n^2 + \cdots, \: \:
  a_n = {1 \over \pi} \int_0^\pi a(\cos \theta) \cos n \theta \, d \theta.
    \end{equation} 
 The term proportional to $k^2$ corresponds to the variance of the    fluctuation of the linear statistic --- for more on this see \cite{Fo23}.
 Using the expansion (see e.g.~\cite[Exercises 1.4 q.4]{Fo10}) $\log(2|\cos \theta - \cos \phi|) = - \sum_{n=1}^\infty {2 \over n} \cos n \theta \cos n \phi$,
 one computes that for $a(x) = a(x;\lambda) = \log | x - \lambda|$, $a_n = - {2 \over n} \cos n \phi$ where $\lambda = \cos \phi$. Simple manipulation then gives
   \begin{equation}\label{6.1E}
  \sum_{n=1}^\infty n a_n^2 ={1 \over 2}  \Big ( \sum_{n=1}^\infty {}^* {1 \over n} - \sum_{n=1}^\infty {\cos 2 n \phi \over n} \Big ) ={1 \over 2} \sum_{n=1}^\infty {}^* {1 \over n} +
{1 \over 2}  \log |2 \sin \phi|;
  \end{equation}   
in relation to the second equality see e.g.~\cite[Eq.~(14.95)]{Fo10}. Here the asterisk indicates that a regularisation of the otherwise
divergent series is required.    One sees that choosing the latter to be $\log N$, (\ref{6.1E}) precisely matches the term raised to the power
of $p^2$ in (\ref{6.1B}).

By making use of the duality (\ref{2.X1c}), the recent work \cite{FS25+} considers 
 the analogue of (\ref{6.1B}) for the general Gaussian $\beta$ ensemble. 
  In fact  the compatibility of (\ref{6.1B}) with (\ref{6.1D}) makes its structure
easy to anticipate. Thus the $\beta$ generalisation of the latter leaves the first term on the RHS unchanged.
It gives rise too to a further term proportional to $k$ but which is independent of $N$ \cite{Jo98}
   \begin{equation}\label{6.1F}
   k \Big ( {1 \over \beta} - {1 \over 2} \Big ) \int_{-\infty}^\infty \Big (  \delta(x - 1) + \delta(x+1)  - {1 \over \pi} {1 \over (1 - x^2)^{1/2}} \mathbbm 1_{|x|<1}\Big ) a(x) \, dx,
\end{equation}  
which comes from the first subleading term of the large $N$ expansion of the smoothed density \cite[\S 3.2]{WF14},
while the only effect on the term proportional to $k$ is the requirement of 
the factor ${2 \over \beta}$; see e.g.~\cite{Fo23}. Indeed as shown in \cite{FS25+} a steepest descent strategy applied to the RHS of (\ref{2.X1c}) 
shows that this anticipated functional form holds true (for purpose of comparison one notes that for $a(x)=a(x;\lambda) = \log |x - \lambda|$
with $|\lambda|<1$
the integral in (\ref{6.1F}) evaluates to $\log|1 - \lambda^2| +  \log 2$; see \cite[\S 1.4.2]{Fo10}).

\begin{prop}\label{P4.4}
Let $p$ be a positive integer. For large $N$ we have
  \begin{multline}\label{6.1B1}   
\Big \langle \prod_{l=1}^{N} |\lambda - x_l|^{2p} \Big \rangle_{{\rm ME}_{\beta ,N}[e^{-\beta N x^2}]}  \\
= A_{\beta,p}^{\rm G} (\pi  \rho^{\rm W}(\lambda))^{p (2- \beta)/\beta}
( \pi N \rho^{\rm W}(\lambda))^{2 p^2/\beta} e^{2Np( \lambda^2 - 1/2-\log \, 2)}  \Big ( 1 + {\rm O} \Big ( {1 \over N}  \Big )\Big ),
   \end{multline}
   where, with $G(z)$ denoting the Barnes G-function,
   \begin{equation}\label{6.1B2}  
 A_{\beta,p}^{\rm G} = \binom{2p}{p} \prod_{j=1}^p {\Gamma(1 + 2j/\beta) \over \Gamma(1 + 2 (j+p)/\beta)} = \prod_{l=0}^{\beta/2} {(G(2(p-l)/\beta + 1))^2 \over G(2(2p-l)/\beta + 1) G(1-2l/\beta)},
   \end{equation}   
   where the equality requires $\beta$ even.
    \end{prop}
    
    \begin{proof} (Sketch only) The details in the case $p=2$ have been given in \cite{DF06}. One first has to modify the
    integration domain corresponding to the RHS of (\ref{2.X1c}) by certain contours
    so that the product of differences factor in the integrand
    is an analytic function. Next, the exponent of the $N$ dependent factors in the integrand is considered from the
    viewpoint of  its saddle points, with the conclusion being that these occur at $u_{\pm} = {1 \over 2}(i\lambda \pm
    \sqrt{1 - \lambda^2})$. After this, half of the integration contours are deformed to pass through $u_+$, and the other half through $u_-$.
    Expansion about these points gives rise to the Gaussian form of the Selberg integral, which has a known
    evaluation as a product of gamma functions   \cite[Eq.~(4.140)]{Fo10}, giving rise to (\ref{6.1B2}).

    \end{proof}
    
    \begin{remark} ${}$ \\
    1.~Substituting $2p = \beta$ gives that the average (\ref{6.1B1}) is proportional to $  \rho^{\rm W}(\lambda)$, which is
    consistent with (\ref{6.1f}) and (\ref{6.1d}). \\
    2.~The constant factor $A_{\beta,p}^{\rm G}$ in (\ref{6.1B1}), in the second of its stated forms, was first identified in the context of the chemical potential of
    an impurity charge of a log-gas on the circle (circular $\beta$ ensemble), and can be given an evaluation in terms
    of the Barnes G-function for all rational $\beta$ \cite{Fo92a}, \cite[Prop.~(14.5.3)]{Fo10}. \\
    3.~Generally
$$
\lim_{n \to 0} {1 \over n} \bigg ( \Big \langle \prod_{l=1}^N |x-x_l|^n \Big \rangle  - 1 \bigg )=  \Big \langle \sum_{l=1}^N \log |x-x_l| \Big \rangle,
$$ 
an identity which underlies the so-called replica trick \cite{EA75}. The problem of taking the limit
$p \to 0$ on the RHS of the duality (\ref{2.X1a}), under the assumption that $N$ is large, is discussed in \cite{KM99} and \cite{Ka02,Ka05}.
\end{remark}

    In the context of application to the Riemann zeta function on the critical line, specifically to an effect known as
singularity dominated strong fluctuations,    the large $N$ form of the reciprocal characteristic
    polynomial moments
   \begin{equation}\label{NC.1}    
   \bigg \langle {1 \over | \det ((x + i \epsilon) \mathbb I_N - H ) |^{2p}} \bigg \rangle_{{\rm ME}_{2,N}(e^{-2 N x^2}) }
 \end{equation}   
 is of interest \cite{BK02}.   Due to the absolute value sign in (\ref{NC.1}), the relevant duality identity is not (\ref{2.X1a+}) but
 rather (\ref{2.X1+}) with $p \mapsto 2p$, $x_1=\dots=x_p = x + i \epsilon$, $x_{p+1}=\dots=x_{2p} = x - i \epsilon$.
 It is found in \cite[Eq.~(4) with $\mu_i \mapsto 2 \mu_i$]{Fy02} (see too \cite{FK03})
 that after normalisation by $| \langle \det((x + i \epsilon) \mathbb I_N - H )^{-p} \rangle |^2$, the $N
 \to \infty$ limit of (\ref{NC.1}) exhibits, for $\epsilon \to 0$, the leading order behaviour
  \begin{equation}\label{NC.2}     
  \bigg ( {\pi \rho^{\rm W}(\lambda) \over 4 \epsilon } \bigg )^{p^2};
  \end{equation}  
 note in particular the exponent $p^2$ as in (\ref{6.1B}).
 
 \subsection{Laguerre and Jacobi $\beta$ ensemble}
 \subsubsection{Global scaling}
 For the density of the Laguerre $\beta$ ensemble to have compact limiting support there are two distinct possible scaled
 Laguerre weights, 
 \begin{equation}\label{W.1}  
 {\rm ME}_{\beta,N} [ x^a e^{-\beta N x/2}], \quad  {\rm ME}_{\beta,N} [ x^{\beta N \gamma /2}  e^{-\beta N x/2}], \quad (a>-1, \gamma \ge 0),
\end{equation} 
and so two distinct ensembles to consider. In both cases the normalised limiting eigenvalue density is independent of $\beta$,
while for the first of these it is independent of $a$ and given by the $\gamma = 0$ case of the second. The functional form is specified 
by the  Marchenko-Pastur density \cite{PS11}
 \begin{equation}\label{W.2} 
 \rho^{\rm MP}(x) = 
{1 \over 2\pi  x }
  \sqrt{( x - c_-) (c_+ - x) } \mathbbm 1_{x \in [c_-,c_+]}, \quad c_\pm := (\sqrt{\gamma + 1} \pm 1)^{1/2}.
    \end{equation}
    
    In the case of the Jacobi weight, the interval of support is compact from its definition, so there is no
    need to scale the eigenvalues. Still, there are distinct cases depending on the scaling of the exponents,
    with the analogue of (\ref{W.1}) being
 \begin{equation}\label{J.1}  
 {\rm ME}_{\beta,N} [ x^{a_1} (1 - x)^{a_2}], \quad  {\rm ME}_{\beta,N} [ x^{\beta N \gamma_1 /2}  (1 - x)^{\beta N \gamma_2/2}].
\end{equation}     
In further analogy to   (\ref{W.1}), in both cases   the normalised limiting eigenvalue density is independent of $\beta$,
while for the first of these it is independent of $a_1,a_2$ and is equal to  the $\gamma_1= \gamma_2 = 0$ case of the second. 
It is given by a functional form first identified by
 Wachter \cite{Wa78},
   \begin{equation}\label{Wa1}
 \rho^{\rm J}(x)  =  (\gamma_1 + \gamma_2+2) {\sqrt{(x - c^{\rm J}) ( d^{\rm J} - x)} \over 2 \pi x (1 - x)},
  \end{equation}   
  supported on $(c^{\rm J}, d^{\rm J})$ with these endpoints specified by
   \begin{equation}\label{Wa2}   
 {1 \over (\gamma_1 + \gamma_2 + 2)^2}  \Big ( \sqrt{(\gamma_1 +1)( \gamma_1 + \gamma_2 + 1)} \pm
 \sqrt{ ( \gamma_2+1 )}     \Big )^2.
  \end{equation}

  One expects that the analogue of Proposition \ref{P4.4} for the ensembles (\ref{W.1}) and (\ref{J.1}), which for $p = \beta/2$
  relates to the computation of the global density,  can be obtained
  by use of the dualities of Proposition \ref{P2.6}. In fact the required working is carried out in the recent paper \cite{FS25+}.
  We know from \S \ref{S4.1} that another pathway to the computation of the global density is via the computation of the
  averaged characteristic polynomial.
   

     \subsubsection{Soft edge scaling}
     The right hand spectrum edge of both versions of the Laguerre ensemble in (\ref{W.1}) are soft edges, as
     they do not border a region where the eigenvalue density is strictly zero. On the other hand, for the Jacabi
     ensemble, only the second version in 	(\ref{J.1}) exhibits a soft edge at its right boundary of leading order support.
     It is further true that the left edge of the second Laguerre and Jacobi versions exhibit a  soft edge at their
     left boundary of leading order support.
     
     At the present time only an analysis of the even $\beta$ density at the right hand soft edge in both the  Laguerre ensembles
     is available in the literature. The first such study was in \cite{DF06}, considering the case of a fixed Laguerre exponent, with
     the limiting density computed to reclaim $\rho_{\infty,(1)}^{\rm soft}(x;\beta)$ as appears in (\ref{4.9}) for the Gaussian $\beta$ ensemble.
     Some years later, in \cite{FT19a}, the analysis was extended to include the case that the Laguerre exponent is proportional to $N$,
     and moreover to the calculation of the first two corrections. The working in \cite{DF06} and \cite{FT19a} was based on the
     duality (\ref{2.4Ci}), with it being a common feature of the $\beta$ ensembles defined by (\ref{2.1}) that for even $\beta$ the density
     is given in terms of an even moment of the characteristic polynomial.
     For both the cases of fixed and growing with $N$ Laguerre exponent it was found that the leading order
     correction was related to $\rho_{\infty,(1)}^{\rm soft}$ by a derivative operation. Due to this, what was the first  correction term in
     the large $N$ expansion can be eliminated by tuning the scaling parameters (recall (\ref{4.9}) and (\ref{4.9a})). However, as presented
     in \cite{FT19a} the required tuning is different in both cases.
     
    In a recent development Bornemann \cite{Bo24a} has found a parametrisation of the mean and standard derivation associated
    with the soft edge scaling which allows both cases to be treated on an equal footing. This requires first considering the Laguerre
    $\beta$ ensemble with the Laguerre exponent parametrised as is natural from the viewpoint of the underlying Wishart matrix
    (see \cite[Eq.~(3.16)]{Fo10}). This is specified by ME${}_{\beta,N}[x^\alpha e^{-\beta x/2}]$ with $\alpha = \beta (n - N + 1)/2 - 1$,
    where $n - N + 1 > 0$. In terms of $N,n$ define
  \begin{equation}\label{abc}
  \mu_{N,n} = (\sqrt{N} + \sqrt{n})^2, \quad \sigma_{N,n} =    (\sqrt{N} + \sqrt{n}) \Big ( {1 \over \sqrt{N}} + {1 \over \sqrt{n}} \Big ), \quad
  h_{N,n} = {1 \over 4} \Big ( {1 \over \sqrt{N}} +  {1 \over \sqrt{n}} \Big )^{4/3}
 \end{equation} 
 (note that $\mu_{N,n} /N \to c_+$, with $c_+$ defined as in (\ref{W.2})).
 Covering both the circumstances of fixed and
 $N$ dependent Laguerre exponent, these variables can be used to write the respective results of \cite{FT19a} as a single
statement.

\begin{prop}
Define $N' = N + {2 - \beta \over 2 \beta}$, $n' = n + {2 - \beta \over 2 \beta}$. For $\beta$ even and $N$ large we have
\begin{equation}\label{4.9a+}
			 \sigma_{N',n'}\rho_{N,\beta}\left( \mu_{N',n'}  +  \sigma_{N',n'} x
			; \lambda^{\beta(n-N+1)/2 - 1} e^{-\beta \lambda / 2}
			\right) =  \rho_{\infty,(1)}^{\rm soft}(x;\beta) + {\rm O}( h_{N',n'}) ;
			\end{equation}
			cf.~LHS with (\ref{4.9a}).
Moreover, for $\beta=1,2$ and 4, results from \cite{Bo24a} imply that higher order terms on the RHS are a power series in
$ h_{N',n'}$.
\end{prop}

\subsubsection{Hard edge scaling}     
In distinction to the Gaussian ensemble, both the Laguerre and Jacobi ensembles exhibit a hard edge. This occurs
at the left edge in the neighbourhood of  the origin for the first ensemble of (\ref{W.1}), and both at the left and right edges, 
in the neighbourhood of $x=0$ and $x=1$ for the first ensemble of (\ref{J.1}). Considering the left edge for definiteness, in the
Laguerre case a well defined limiting statistical state (the hard edge state) results from the scalings $\lambda_i \mapsto \lambda_i/(4N)$
(the factor of 4 is just for convenience), while in the Jacobi case one requires $\lambda_i  \mapsto \lambda_i/(2N^2)$.
For even $\beta$ Corollary \ref{C2.1} provides access to study the limiting functional form (first carried out for the Laguerre case in
\cite{Fo93c}). In fact, as for the soft edge, the leading
correction was found to be related to the limiting density by a derivative operation (first observed in \cite{GFF05,FFG06} in relation
to the cases $\beta = 1,2$ and 4), and so can be eliminated by tuning the definition of the hard edge scaling.

At the hard edge, a further observable can similarly be analysed. This is the probability of no eigenvalue in a neighbourhood of the origin,
$E_{N,\beta}(0,(0,s);\lambda^a e^{- \beta \lambda/2})$ say. In the hard edge scaling limit for general $\beta > 0$ and with $a$ a non-negative
its limiting functional form was computed in \cite{Fo93c}, both in terms of the generalised hypergeometric function 
 ${\vphantom{F}}_0^{\mathstrut} F_1^{(\alpha)}$ with $a$ variables all equal, and as an $a$-dimensional integral. Using in an essential way
 the duality of Corollary \ref{C2.1}, the first two correction terms in powers of $1/N$ were computed in \cite{FT19} (Laguerre case)
 \cite[Appendix, arXiv version]{FL22} (Jacobi case, first correction only) with this same analysis in the Jacobi case also carried out directly from the
finite $N$ generalised hypergeometric expression in \cite{Wi23}. As for the density, the first correction was found to be related to the
limiting probability by a derivative operation (see also \cite{EGP16,Bo16,PS16,HHN16,MMM19} in relation to this effect specific to the case
$\beta = 2$). For brevity of presentation we make note here only of the results in relation to the probability $E_{N,\beta}(0,(0,s);\cdot)$.
The limit will be given in terms of  a particular ${\vphantom{F}}_0^{\mathstrut} F_1^{(\alpha)}$ in the case of equal variables, when
we have available the multidimensional integral form (\cite[modification of Eq.~(13.27)]{Fo10})
\begin{multline}\label{4.32}
{\vphantom{F}}_0^{\mathstrut} F_1^{(\beta/2)}(c + 2(p-1)/\beta;(u)^p) = \prod_{j=1}^p {\Gamma(1 + 2/\beta) \Gamma(c + 2 (j-1)/\beta) \over
\Gamma(1 + 2j /\beta)} \\
\times \int_{-1/2}^{1/2} d \theta_1 \cdots  \int_{-1/2}^{1/2} d \theta_p \, \prod_{j=1}^p e^{2 \pi i (c - 1)\theta_l} e^{u e^{2 \pi i \theta_l} + e^{-2 \pi i \theta_l}}
\prod_{1 \le j < k \le p} | e^{2 \pi i \theta_k} -  e^{2 \pi i \theta_j} |^{4/\beta}.
\end{multline}
With $z_l = e^{2 \pi i \theta_l}$, each integration herein can be thought of as over the unit circle in the complex $z_l$-plane. In fact
for (\ref{4.32}) to hold with this contour of integration it is required that $c$ be a positive integer. Otherwise the contour has to run along
the negative real axis (lower half plane side), before traversing the unit circle, and returning along the negative real axis
(upper half plane side); see the discussion in \cite[Proof of Prop.~2]{Fo13j}.

\begin{prop}\label{P4.7a}
Consider ME${}_{\beta,N}[ \lambda^a e^{- \beta \lambda/2}]$ and require that $a$ be a non-negative integer.
Define $N_{\rm L} = N+a/\beta$. For large $N$ we have
\begin{equation}\label{4.33}
E_{N,\beta}(0,(0,s/(4 N_{\rm L} );\lambda^a e^{-\beta \lambda/2}) = e^{-\beta s/8}
{\vphantom{F}}_0^{\mathstrut} F_1^{(\beta/2)}(2 a /\beta;(s/4)^a)  + {\rm O}( N_{\rm L}^{-2}).
\end{equation}

Consider ME${}_{\beta,N}[ \lambda^{a_1}(1 - \lambda)^{a_2} e^{- \beta \lambda/2}]$ and require that $a_1$ be a non-negative integer.
Define $N_{\rm J} = N-1/2+(1+a_1 + a_2)/\beta$. We have
\begin{equation}\label{4.34}
E_{N,\beta}(0,(0,s/(2 (N_{\rm J})^2 );\lambda^{a_1} (1 - \lambda)^{a_2}) = e^{-\beta s/8}
{\vphantom{F}}_0^{\mathstrut} F_1^{(\beta/2)}(2 a_1 /\beta;(s/4)^a)  + {\rm O}( N_{\rm J}^{-2}).
\end{equation}
\end{prop}

\begin{remark} ${}$ \\
1.~For finite $N$ the probability of no eigenvalues in $(0,s)$ for the Laguerre and Jacobi ensembles with $a$ and $a_1$ respectively non-negative
integers admit evaluation formulas in terms of Jack polynomial hypergeometric functions according to (\cite{Fo93c} Laguerre case, use of
(\ref{2.4Ea+}) Jacobi case)
\begin{align}\label{4.35}
E_{N,\beta}(0,(0,s);\lambda^a e^{-\beta \lambda/2}) & =  e^{-\beta N_{\rm L} s / 2} e^{a s /2} {\vphantom{F}}_1^{\mathstrut} F_1^{(\beta/2)}(-N_{\rm L};2 a /\beta;(-s)^a) \nonumber \\
E_{N,\beta}(0,(0,s);\lambda^{a_1} (1 - \lambda)^{a_2} ) &  = (1 - s)^{(\beta/2) (N_{\rm J}^2 - \gamma^2)}
 {\vphantom{F}}_2^{\mathstrut} F_1^{(\beta/2)}(-N_{\rm J} + \gamma, N_{\rm J} + \gamma;2 a /\beta;(-s/(1-s))^{a_1}) ,
\end{align}
where the notation $N_{\rm L}, N_{\rm J}$ is as in Proposition \ref{P4.7a} and
$\gamma :=  (a_1+a_2+1)/\beta - 1/2$. From the symmetry ${\vphantom{F}}_2^{\mathstrut} F_1^{(\beta/2)}(a,b;c;\mathbf x) =
{\vphantom{F}}_2^{\mathstrut} F_1^{(\beta/2)}(b,a;c;\mathbf x)$, 
we see immediately that the second of these is even in $N_{\rm J}$, telling us that the large $N$ expansion
with $s \mapsto s/(2N_{\rm J}^2)$ is
in powers of $1/N_{\rm J}^2$. Also, the generalised Kummer transformation for $ {\vphantom{F}}_1^{\mathstrut} F_1^{(\alpha)}$ (see e.g.~\cite[Eq.~(13.16)]{Fo10})
\begin{equation}
 {\vphantom{F}}_1^{\mathstrut} F_1^{(\alpha)}(a;c;t_1,\dots,t_m) = \prod_{j=1}^m e^{t_j}
 {\vphantom{F}}_1^{\mathstrut} F_1^{(\alpha)}(c-a;c;-t_1,\dots-,t_m) ,
 \end{equation}
applied to the first formula in (\ref{4.35})  tells us that $E_{N,\beta}(0,(0,s/(4N_{\rm L}));\lambda^a e^{-\beta \lambda/2})$ is even in $N_{\rm L}$ and so has an expansion for large $N$
in powers of  $1/N_{\rm L}^2$; cf.~Remark \ref{R4.1}.1.\\
2.~From the relation between (\ref{2.5b+Z}) and (\ref{6.3g1}) specialised to $n=N$ and $L= \mathbb I_N$ we have that for $\beta = 1,2$ and 4 (after changing
variables $V^\dagger U \mapsto U$)
\begin{equation}\label{4.37}
{\vphantom{F}}_0^{\mathstrut} F_1^{(2/\beta)}(\beta N/2;\mathbf y) = \int e^{{\rm Tr} (U^\dagger \Lambda + U \Lambda)} \, d^{\rm H}U.
 \end{equation}
 Here the matrix group integral is over the orthogonal group ($\beta = 1$), unitary group ($\beta = 2$) and symplectic unitary group $(\beta = 4)$.
 In the case that $ \Lambda = y \mathbb I_N$, the same ${\vphantom{F}}_0^{\mathstrut} F_1^{(2/\beta)}$ appear in (\ref{4.33}) and (\ref{4.34}),
 provided $2 a /\beta$ is a non-negative integer, thus giving an evaluation of those hard edge probabilities in terms of classical matrix group integrals
 \cite[\S 5.2]{FW04} (see also \cite[\S 3.1]{FM23} and \cite[Th.~2.1]{Bo24}). On the other hand these same group integrals appear as
 generating functions for the longest increasing subsequence length of classes of random permutations \cite{Ra98}, \cite{BR01a}. This fact has
 been put to use as a strategy for the analysis of the asymptotics of the  longest increasing subsequence length in these cases
 \cite{BF03}, \cite{Bo24b}, \cite{Bo24}.
\end{remark}

\subsection{Circular and circular Jacobi $\beta$ ensembles}
\subsubsection{Bulk scaling for the circular ensemble two-point correlation}  
Generally the bulk of the spectrum is any portion that is away from the edge or a singularity. In the circular ensemble the
statistical state is rotationally invariant, and there is no edge or singularity associated with the spectrum. Bulk scaling refers to
the use of scaled eigenvalues so that the mean spacing between eigenvalues is a nonzero constant (which we take to be
unity) independent of $N$.   This is achieved by setting $\theta_i = x_i/N$.

The two-point correlation function --- defined as the integral over all but two of the eigenvalues in the joint eigenvalue PDF,
multiplied by $N(N-1)$ as a type of normalisation --- is simply related to the average on the LHS of (\ref{3.8c}) with $N \mapsto
N-2$ and $2q=2\mu = \beta$. Thus,  following \cite{Fo94j}, a $\beta$-dimensional integral evaluation for the
limiting two-point correlation function, $\rho_{(2)}^{\rm bulk}(x_1,x_2;\beta)$ say, can be obtained by the duality averages on the
RHS of (\ref{3.8c}) (for definiteness, the second of these has been chosen). 
\begin{prop}\label{P4.5}
Define
\begin{equation}
S_N(\lambda_1,\lambda_2,\beta) =
 \prod_{j=0}^{N-1} {\Gamma (\lambda_1 + 1 + j\beta/2)
\Gamma (\lambda_2 + 1 + j\beta/2)\Gamma(1+(j+1)\beta/2) \over
\Gamma (\lambda_1 + \lambda_2 + 2 + (N + j-1)\beta/2) \Gamma (1 + \beta/2 )},
\label{3.2S}
\end{equation}
which is the normalisation associated with ME${}_{\beta,N}[x^{\lambda_1}(1 - x)^{\lambda_2})]$ as evaluated by
Selberg; see e.g.~\cite[\S 4.1]{Fo10}. For $\beta$ even we have
\begin{multline}\label{15.eev1}
\rho_{(2)}^{\rm bulk} (x_1,x_2;\beta)  = 
(\beta / 2)^\beta {((\beta /2)!)^3 \over \beta! (3 \beta /2)!} 
{e^{- \pi i \beta  (x_1 - x_2)} (2 \pi  (x_1 - x_2))^\beta \over
 S_\beta(-1+ 2/\beta,-1+ 2/\beta,4/\beta)} 
\int_{[0,1]^\beta} du_1 \cdots du_\beta \\
 \times \prod_{j=1}^\beta
e^{2 \pi i  (x_1 - x_2)u_j} u_j^{-1 + 2/\beta} (1 - u_j)^{-1 + 2/\beta}
\prod_{1 \le j < k \le \beta} |u_k - u_j|^{4/\beta}.
\end{multline}
 \end{prop}

 \subsubsection{Moments of moments of the characteristic polynomial}
 The moments of moments of the characteristic polynomial for the circular $\beta$ ensemble are defined as
 \cite{As22}
\begin{multline}\label{M.1}
{\rm MoM}^{(\beta)}(k;q)  = \bigg \langle \Big ( \int_0^1 \prod_{l=1}^N | e^{2 \pi i \phi} + e^{2 \pi i \theta_l} |^{2q} \, d \phi \Big )^k \bigg \rangle_{{\rm ME}_{\beta,N}[1]}
\\ = 
\int_0^1 d \phi_1 \cdots \int_0^1 d \phi_k \, \bigg \langle
\prod_{l=1}^N | e^{2 \pi i \phi_1} + e^{2 \pi i \theta_l} |^{2q} \cdots  | e^{2 \pi i \phi_k} + e^{2 \pi i \theta_l} |^{2q} \bigg \rangle_{{\rm ME}_{\beta,N}[1]},
 \end{multline}
 where the second line applies for $k$ a positive integer. In the case $\beta = 2$ this quantity is of interest for its application to the local maxima of
 the Riemann zeta function on the critical line, and more generally for its relevance to the study of Gaussian multiplicative chaos;
 see the review \cite{BK22}.
 
  In the case $k=2$ and $q$ a positive integer the duality (\ref{3.8c}) can be applied
 to give a rewrite of the second average in (\ref{M.1}). Taking into consideration that the implied
 proportionality constant therein is 
 \begin{equation}\label{M.1a}
 \Big \langle \prod_{l=1}^N | 1 + e^{2 \pi i \theta_l} |^{4q} \Big \rangle_{{\rm ME}_{\beta,N}[1]} = \prod_{l=1}^{2q}{ \Gamma(N+ 2(2q+l) /\beta)  \Gamma(2l /\beta)   \over  \Gamma(N+ 2l /\beta)   \Gamma(2(2q+l) /\beta) },
\end{equation}
where use has been made of the so-called Morris integral evaluation \cite[Eq.~(4.4)]{Fo10}, with a simplification particular to the assumption that $q$ is integer, we can reclaim
a result of Assiotis \cite[Eq.~(18) written in the form of the first displayed equation in \S 3.3 therein]{As22} obtained using different working.

\begin{prop}\label{P4.6}
For $ 4 q^2 > \beta$ we have
\begin{multline}\label{M.1b}
\lim_{N \to \infty} {1 \over N^{8 q^2/\beta - 1}} {\rm MoM}^{(\beta)}(k;q)  |_{k=2}  = {1 \over (2q)!} \prod_{j=0}^{2q-1} {\Gamma(2/\beta) \over (\Gamma(2(j+1)/\beta))^2}
\\ \times
 {1 \over 2 \pi} \int_{-\infty}^\infty e^{i q s}
\int_0^1 dx_1 \cdots \int_0^1 dx_{2q} \, \prod_{l=1}^{2q}e^{- i s x_l} (x_l(1-x_l))^{-1+2/\beta} \prod_{1 \le j < k \le 2q} | x_k - x_j|^{4/\beta}.
\end{multline}
\end{prop}

\begin{proof}
On the RHS of (\ref{M.1}) we use the periodicity in $\phi$ to shift the integration range to $[-1/2,1/2]$. As already remarked, for 
$k=2$ and $q$ a positive integer the duality (\ref{3.8c}) can be applied
 to give a rewrite of the second average in (\ref{M.1}). Changing variables $\phi = s/(2 \pi N)$ and use of the elementary limit $(1 + u/N)^N \to e^u$ shows that, up
 to the factor corresponding to the normalisation of the average and a factor of $1/N$, the second line of (\ref{M.1b}) results. However this procedure is only well defined
 if the resulting integral over $s$ is absolutely integrable. Using a method to analyse the multidimensional integral over $x_1,\dots,x_{2q}$ for large $s$ presented
 in \cite{Fo93y}, which involves considering the contribution of half of the integration variables clustered near 0, and the other half clustered near
 1 using a simple scaling of variables, we see that this requires $ 4 q^2 > \beta$, with the decay of the integrand then being of order $|s|^{-4 q^2/\beta}$. The factors of $N$ result from the large $N$ form of ratios of $N$ dependent
 gamma functions in (\ref{M.1a}), and the factor of a power of $1/N$ from the change of variables. The factors on the first line of the RHS of (\ref{M.1b}) result from (\ref{M.1a})
 and the product of gamma function function form of the normalisation factor for the average in the duality as read off from (\ref{3.2S}).
\end{proof}

\subsection{Circular Jacobi density for even $\beta$}
In (\ref{2.6a}) we set $a_1 = (\beta p + i q)/2$, $a_2 = \bar{a}_1$ so that the weight reads
 \begin{equation}\label{CW1}
 w(\theta) = e^{- \pi q \theta} |1 + e^{2 \pi i \theta} |^{\beta p}.
  \end{equation}
  The points $\theta = \pm 1/2$ are referred to as a spectrum singularity of Fisher-Hartwig type (the parameter
  $p$ can be viewed as determining the degeneracy of a conditioned eigenvalue at $\theta = \pm 1/2$, while $q$  determines
  the amplitude of the discontinuity in the factor $e^{- \pi q \theta}$ about the points $\theta = \pm \pi$). In the case
  $\beta$ even, the average on the LHS of (\ref{3.8d+}) is (upon appropriate identification of the parameters) simply related to the
  density in the Jacobi circular ensemble with weight (\ref{CW1}). Use of (either of) the dualities on the RHS of (\ref{3.8d+}) leads
  to a $\beta$-dimensional integral form of the limiting density in the neighbourhood of the singularity,
  which is obtained by the scalings $\theta \mapsto - 1/2  + x/(2 \pi N)$ $(x>0$) and
 $\theta \mapsto 1/2  + x/(2 \pi N)$ $(x<0$) 
   \cite{FLT21}.
  Alternatively, from the underlying theory as outlined in the proof
  of Proposition \ref{P3.7}, the evaluation formula can be given in terms of  a particular ${\vphantom{F}}_1^{\mathstrut} F_1^{(\alpha)}$ \cite{Li17} in the case of 
  $\beta$ equal variables, using the integral representation (see e.g.~\cite[Exercises 13.1 q.4(i)]{Fo10})
  \begin{multline}\label{CW1a} 
{\vphantom{F}}_1^{\mathstrut} F_1^{(\alpha)}(-b; a + 1 +  (N - 1) / \alpha;
(t)^n) \\ \propto
\int_{-1/2}^{1/2}dx_1 \cdots \int_{-1/2}^{1/2}dx_n \,
\prod_{l=1}^n 
e^{\pi i x_l(a-b)} |1 + e^{2 \pi i x_l}|^{a+b}
e^{-t e^{2 \pi i x_l}} \prod_{j < k}
|e^{2 \pi i x_k} - e^{2 \pi i x_j}|^{2/\alpha}.
  \end{multline}
(We remark that in the cases $\beta =2$ and 4 linear differential equations of degree 3 and 5 for the limiting
density have been given in \cite{FS23}.)

 \begin{prop}\label{P4.7}
Consider the Jacobi circular weight (\ref{CW1}) in (\ref{2.1X}), and for convenience shift each $\theta_l$ by $\theta_l \mapsto -1/2 + \theta_l$.
Denote by   $\rho_{\infty,(1)}^{\rm cJ}(x;\beta)$ the limiting scaled density about the origin, and
define $N_{\rm cJ} = N + p$. We have that for large $N$
 \begin{equation}\label{CW2} 
{1 \over N_{\rm cJ}} \rho_{N,(1)} \Big (x / N_{\rm cJ}; e^{- \pi q (\theta-1/2)}   |1 - e^{2 \pi i \theta} |^{\beta p} \Big ) =  \rho_{\infty,(1)}^{\rm cJ}(x) + {\rm O}(N_{\rm cJ} ^{-2}),
  \end{equation}
  where
  \begin{equation}\label{CW2a}  
  \rho_{\infty,(1)}^{\rm cJ}(x) \propto e^{q \pi {\rm sgn}(x)} e^{i \beta x/2} |x|^{p \beta} 
  {\vphantom{F}}_1^{\mathstrut} F_1^{(\beta/2)}(p+1-2iq/\beta;2p+2;(-ix)^\beta)). 
   \end{equation} 
  \end{prop}

\section{Characteristic polynomial dualities for non-Hermitian random matrices}\label{S5}
\subsection{Introductory remarks}
In view of the result (\ref{1.0}), a logical starting point to explore dualities in relation to characteristic
polynomials for non-Hermitian matrices is to consider
 \begin{equation}\label{6.0}
 \langle \det(z_1 \mathbb I_N - Z) \det({z}_2 \mathbb I_N - \bar{Z}) \rangle
 \end{equation}
 for $Z$ an $N \times N$ random matrix with all elements independently and identically
 distributed, chosen from a distribution with mean zero and variance $\sigma^2$ (interpreted as
 $\sigma^2 = \langle |z_{ij}|^2 \rangle$ is the elements are complex, which relates to the requirement
 of a complex conjugate in the second factor of (\ref{6.0})). One notes that if the second determinant were
 absent, (\ref{6.0}) would trivially evaluate to $(z_1)^N$. In the case of standard Gaussian real entries (a general variance
 can be inserted by scaling), use of Grassmann integration by Akemann, Phillips and Sommers \cite{APS09} gave the
 duality relation
 \begin{equation}\label{6.0a}
 \langle \det(z_1 \mathbb I_N - Z) \det({z}_2 \mathbb I_N - {Z}) \rangle = \langle (z_1 z_2 + |w|^2)^N \rangle_{w \in \mathcal{CN}(0,\sigma^2)},
  \end{equation}
  where $\mathcal{CN}(0,\sigma^2)$ denotes the complex normal distribution with mean zero and variance $\sigma^2$; cf.~(\ref{1.0}).
  The case $z_1 = z_2$ was considered earlier in \cite{Ed97}, where a duality per se was not identified but rather the average
  was evaluated as an incomplete gamma function, which is noted in \cite{APS09} is consistent with the computation of the RHS
  of (\ref{6.0a}) as a series in $(z_1 z_2)$.
  
  A generalisation of (\ref{6.0}) is to consider
 \begin{equation}\label{6.0v}
\Big  \langle \prod_{l=1}^k \det(z_l \mathbb I_N - Z)  \prod_{l=1}^k \det({w}_l \mathbb I_N - \bar{Z}) \Big \rangle.
 \end{equation}
 Note that if all the entries of $Z$ are real, the second product is not necessary. An evaluation of (\ref{6.0v}) in the case of
 the so-called normal matrix model (see e.g.~\cite[Ch.~5]{BF24}), defined as having an eigenvalue PDF in the complex plane 
 of the form proportional to
  \begin{equation}\label{6.0w}
  \prod_{l=1}^N e^{- V(z_l, \bar{z}_l)} \prod_{1 \le j < k \le N} | z_j - z_k |^2,
  \end{equation} 
 and moreoever  with the requirement of rotational invariance $V(z_l, \bar{z}_l ) = V(|z_l|)$, an evaluation formula analogous to
 (\ref{2.3}) has been obtained in \cite{AV03}. Prominent in the class of rotationally invariant normal matrix models is the Ginibre unitary ensemble (GinUE)
 consisting of $N \times N$ matrices with independent standard complex Gaussian entries. One then has $V(z_l, \bar{z}_l ) = |z|^2$; see  \cite[Eq.~(1.7)]{BF24}.
A duality identity for multiple products of characteristic polynomials for a non-Hermitian ensemble was first given by Nishigaki and Kamenev
\cite{NK02}, who deduced
 \begin{equation}\label{6.0x}
\Big  \langle | \det(z_1 \mathbb I_N - Z)|^{2k}
 | \det(z_2 \mathbb I_N - Z)|^{2k} \Big \rangle_{Z \in {\rm GinUE}_N}
= \bigg \langle \det \begin{bmatrix} D & -Y \\ Y^\dagger & D^\dagger \end{bmatrix}^N \bigg \rangle_{Y \in {\rm GinUE}_{2k}},
  \end{equation} 
 where $D$ is the $2k \times 2k$ diagonal matrix with the first $k$ diagonal entries equal to $z_1$ and the remaining $k$
 diagonal entries equal to $z_2$. For the generalisation of (\ref{6.0x}) applying to (\ref{6.0v}), see Proposition \ref{P5.2} below.
 
 The $\beta$ generalisation of (\ref{6.0w}) is to replace the exponent 2 on the product of differences with a general $\beta > 0$, which
 gives the two-dimensional Coulomb gas analogue of (\ref{2.1}). While in the Hermitian case of (\ref{2.1}) we have seen that  duality
 relations enjoy a rich $\beta$ generalisation, this is no longer true of two-dimensional Coulomb gas non-Hermitian  $\beta$ generalisations.
 Rather the variety of cases comes from the distinction of  real, complex and quaternion elements, distinguished choices of invariant measures
 related to the classical weights in the Hermitian case, and possible symmetries of the matrices.
 
 \subsection{Zonal polynomials}
 Relationships between the Jack polynomial based hypergeometric functions with parameter
 $\alpha = 2/\beta$ and matrix integrals over the orthogonal group ($\beta = 1$), unitary group ($\beta = 2$) and
 symplectic unitary group ($\beta = 4$) have been seen in (\ref{6.2g2}), (\ref{2.5b+Z}) and (\ref{6.3g1}), and (\ref{4.37}).
 With parameter  $\alpha = 2/\beta$ and $\beta = 1,2$ and 4 the Jack polynomials are known as zonal polynomials, 
 or sometimes as zonal spherical polynomials, due to their appearance in the broader theory of spherical functions \cite{Ma95}.
 Due to their relation to group integrals, they have received attention long before the more general Jack polynomials, particularly in the
 case $\alpha = 2$  ($\beta = 1$) in the context of multivariate statistics \cite[Ch.~7]{Mu82}. A comprehensive historical
 account of the early literature is given in the thesis \cite{Ro01}.
 
 Underpinning the evaluation of the group integrals (\ref{6.2g2}) and (\ref{4.37}) are group integral identities applying directly
 to the zonal polynomials. Conventionally, the normalisation
 used for the zonal polynomials is no longer that implied by (\ref{4.0}), but rather that given by introducing a scaling factor and defining,
  \begin{equation}\label{15.ckp}
C_\kappa^{(\alpha)}(\mathbf x) := {\alpha^{|\kappa|} |\kappa|! \over
h_\kappa'} P_\kappa^{(\alpha)}(\mathbf x);
\end{equation}
recall (\ref{2.2h}) in relation to $h_\kappa'$ and cf.~the definition of the Jack polynomial based hypergeometric function
(\ref{3.40}).

 \begin{prop}\label{P5.1} (\cite[case $\beta = 1$]{Ja64}, \cite{Ma95}, \cite{Ra95})
Let $U$ be from the matrix group of unitary matrices
as specified above depending on the value of $\beta$.
Let $A$ and $B$ be $N \times N$ matrices with real,
complex and real quaternion entries for $\beta=1,2,4$, respectively. We have 
\begin{equation}\label{16.jl}
\int 
C_\kappa^{(2/\beta)}(A U^\dagger B U) \,  d^{\rm H} U=
{C_\kappa^{(2/\beta)}(A)C_\kappa^{(2/\beta)}(B) \over C_\kappa^{(2/\beta)}((1)^{N})},
\end{equation}
where $C_\kappa^{(2/\beta)}(A U^\dagger B U)$ is
defined as $C_\kappa^{(2/\beta)}(\mathbf y)$ with $\mathbf y = (y_1,\dots,y_{N})$ denoting the
eigenvalues of $A  U^\dagger B  U$, and similarly the meaning of
$C_\kappa^{(2/\beta)}(A)$, $C_\kappa^{(2/\beta)}(B)$. In the quaternion case it is required
that all matrices have doubly degenerate eigenvalues, and only one copy is
included in the arguments of $C_\kappa^{(1/2)}$.

Furthermore, with $s_\kappa(\mathbf x) = P^{(\alpha)}_\kappa(\mathbf x)|_{\alpha = 1}$
denoting the Schur polynomial, one has
\begin{eqnarray}
\langle s_{\lambda}(A O) \rangle_{{O} \in O(N)} & = &
\left \{ \begin{array}{ll}
\displaystyle
{C_\kappa^{(2)}(A A^T) \over C_\kappa^{(2)}((1)^N)}, & \lambda = 2 \kappa, \\
0 & {\rm otherwise}, \end{array} \right.
\label{16.9.1} \\
\langle s_{\lambda}(A U) s_\kappa( U^\dagger  A^\dagger)
\rangle_{{ U} \in U(N)} & = & \delta_{\lambda, \kappa}
{C_\kappa^{(1)}( A  A^\dagger) \over C_\kappa^{(1)}((1)^N)},
\label{16.9.2} \\
\langle s_{\lambda}(A  S) \rangle_{{ S} \in {\rm Sp}(2N)} & = &
\left \{ \begin{array}{ll} \displaystyle
{C_\kappa^{(1/2)}( A A^\dagger) \over C_\kappa^{(1/2)}((1)^N)}, &
\lambda = \kappa^2,
\\ 0  & {\rm otherwise}, \end{array} \right.
\label{16.9.3}
\end{eqnarray}
where in (\ref{16.9.1}) the partition $2 \kappa$ is the partition obtained
by doubling each part of $\kappa$, while in (\ref{16.9.3}),
$\kappa^2$ is the partition obtained by repeating each part of
$\kappa$ twice. We note too that in (\ref{16.9.3}), the meaning of
$C_\kappa^{(1/2)}(A A^\dagger)$ is $C_\kappa^{(1/2)}(\mathbf x)$ where
$\mathbf x$ is the $N$ independent eigenvalues of $A A^\dagger$ (for $A$
quaternion, the eigenvalues of $A A^\dagger$ are doubly degenerate).

\end{prop}

Application to random non-Hermitian random matrices relies on consequences of these group integrals
in the case that the non-Hermitian random matrices being averaged over are bi-unitary invariant
\cite{Ta84}, \cite{Ra95}, \cite{FR09}, \cite[Exercises 13.4 q.1]{Fo10}.

\begin{cor}
For $\beta = 1,2$ and 4, let $X,A,B$ be $N \times N$ matrices with real, complex and
quaternion entries respectively. Let $X$ be random matrix chosen from a left and right
unitary invariant distribution, with the unitary matrices drawn from the orthogonal group,
unitary group and symplectic unitary group respectively. We have
\begin{equation}\label{16.jlX}
\langle
C_\kappa^{(2/\beta)}(A X^\dagger B X) \rangle_X =
{C_\kappa^{(2/\beta)}(A)C_\kappa^{(2/\beta)}(B) \over (C_\kappa^{(2/\beta)}((1)^{N}))^2} 
\langle
C_\kappa^{(2/\beta)}(X X^\dagger) \rangle_X.
\end{equation}
As in Proposition \ref{P5.1}, in the quaternion case it is required
that all matrices have double degenerate eigenvalues, and only one copy is
included in the arguments of $C_\kappa^{(1/2)}$.
Also
\begin{eqnarray}
\langle s_{\lambda}(A X) \rangle_{X} & = &
\left \{ \begin{array}{ll}
\displaystyle
{C_\kappa^{(2)}(A A^T) \over C_\kappa^{(2)}((1)^{N})} \langle C_\kappa^{(2)}( X X^T) \rangle_{X} , & \lambda = 2 \kappa, \\
0 & {\rm otherwise}, \end{array} \right.
\label{16.9.1X} \\
\langle s_{\lambda}(A X) s_\kappa( X^\dagger  A^\dagger)
\rangle_{X} & = & \delta_{\lambda, \kappa}
{C_\kappa^{(1)}( A  A^\dagger) \over C_\kappa^{(1)}((1)^{N})}   \langle C_\kappa^{(1)}( X X^\dagger) \rangle_{X} ,
\label{16.9.2X} \\
\langle s_{\lambda}(A  X) \rangle_{X} & = &
\left \{ \begin{array}{ll} \displaystyle
{C_\kappa^{(1/2)}( A A^\dagger) \over C_\kappa^{(1/2)}((1)^{N})}  \langle C_\kappa^{(1/2)}( X X^\dagger) \rangle_{X} , &
\lambda = \kappa^2,
\\ 0  & {\rm otherwise}, \end{array} \right.,
\label{16.9.3X}
\end{eqnarray}
where for the averages on the LHS the matrices are as required for $beta = 1,2$ and 4 respectively.
As in (\ref{16.9.3}), in (\ref{16.9.3X})
the meaning of
$C_\kappa^{(1/2)}(A A^\dagger)$ is $C_\kappa^{(1/2)}(\mathbf x)$ where
$\mathbf x$ is the $N$ independent eigenvalues of $A A^\dagger$.

\end{cor}

\begin{proof}
Let $d\mu(X)$ be a left and right unitary invariant measure, and let $f(X)$ be integrable with respect to this measure.
Left unitary invariance tells us that
\begin{equation}\label{Aw.1}
\langle f(AXBX^\dagger) \rangle_X = \langle \langle f(AUXBX^\dagger U^\dagger) \rangle_U \rangle_X.
\end{equation}
Choosing $f = C_\kappa^{(\alpha)}$ and using (\ref{16.jl}) to evaluate the average over $U$ on the RHS shows
\begin{equation}\label{Aw.2}
\langle C_\kappa^{(\alpha)}(AXBX^\dagger) \rangle_X = {C_\kappa^{(\alpha)}(A) \over C_\kappa^{(\alpha)}((1)^{N})}
 \langle C_\kappa^{(\alpha)}(XBX^\dagger) \rangle_X.
\end{equation}
Regarding the RHS, right unitary invariance and further use of  (\ref{16.jl}) gives
\begin{equation}\label{Aw.23}
 \langle C_\kappa^{(\alpha)}(XBX^\dagger) \rangle_X =  \langle \langle f(XUBU^\dagger X^\dagger U^\dagger) \rangle_U \rangle_X=
  {C_\kappa^{(\alpha)}(B) \over C_\kappa^{(\alpha)}((1)^{N})}
 \langle C_\kappa^{(\alpha)}(XX^\dagger) \rangle_X
 \end{equation}
 and (\ref{16.jlX}) results.
 
 The derivation of the remaining identities from their companion identities in Proposition \ref{P5.1} is similar.
 \end{proof}

\subsection{Dualities for products and powers of characteristic polynomials}
\subsubsection{Complex entry matrices}
A duality formula for the average of the product of random matrices in (\ref{6.0v}),
generalising (\ref{6.0x}), has been given in the thesis of Serebryakov \cite[second statement of Th.~4.1]{Se23}.

\begin{prop}\label{P5.2}
We have
\begin{equation}\label{6.0v+}
\Big  \langle \prod_{l=1}^k \det(z_l \mathbb I_N - Z)   \det({w}_l \mathbb I_N - \bar{Z}) \Big \rangle_{Z \in {\rm GinUE}_{N}} =
\bigg \langle \det \begin{bmatrix} D_1 & -Y \\ Y^\dagger & D_2 \end{bmatrix}^N \bigg \rangle_{Y \in {\rm GinUE}_{k}},
 \end{equation}
 where $D_1, D_2$ are diagonal matrices with diagonal entries $\mathbf z$ and $\mathbf w$ respectively.
\end{prop}

\begin{proof}
Simple manipulation and use of the dual Cauchy identity (\ref{SM2}) with $\alpha = 1$ gives that the LHS of (\ref{6.0v+})
is equal to
\begin{equation}\label{7.X}
\prod_{l=1}^k (w_l z_l)^N \sum_{\mu, \kappa \subseteq (k)^N} (-1)^{|\mu| + | \kappa|} s_{\mu'}(1/\mathbf z) s_{\kappa'}(1/\mathbf w)
\langle s_\kappa(Z) s_\mu(\bar{Z}) \rangle_{Z \in {\rm GinUE}_{N}},
 \end{equation}
 where the notation $1/\mathbf z$ (similarly $1/\mathbf w$) is used to denote the vector with entries the reciprocal of the entries of $\mathbf z$.
 The average in this expression is a particular case of (\ref{16.9.2X}), reducing the task to computing 
 $\langle s_\kappa(ZZ^\dagger)  \rangle_{Z \in {\rm GinUE}_{N}}$. This can be achieved by first changing variables
 $W = ZZ^\dagger$ (see \cite[Prop.~3.2.7]{Fo10}), then changing variables to the eigenvalues and eigenvectors of $\mathbf W$ (see
 \cite[Prop.~3.2.2]{Fo10}), and finally making use of the special case $\alpha = 1$, $a=0$ of the known Jack polynomial
 average (see e.g.~\cite[Eq.~(12.152)]{Fo10})
 \begin{equation}\label{7.X1}
 \langle P_\kappa^{(\alpha)}(\mathbf x) \rangle_{{\rm ME}_{2/\alpha,N}[\lambda^a e^{-\lambda}]} =
P_\kappa^{(\alpha)}((1)^N) [a+1+ (N-1)/\alpha]_\kappa^{(\alpha)}.
 \end{equation}
 Consequently the LHS of (\ref{6.0v+}) is reduced to the form
 \begin{equation}\label{7.X2} 
 \prod_{l=1}^k (w_l z_l)^N   \sum_{ \kappa \subseteq (k)^N}  s_{\kappa'}(1/\mathbf z) s_{\kappa'}(1/\mathbf w) [N]_\kappa^{(1)} =
  \prod_{l=1}^k (w_l z_l)^N   \sum_{ \kappa \subseteq (N)^k}  s_{\kappa}(1/\mathbf z) s_{\kappa}(1/\mathbf w) [N]_{\kappa'}^{(1)}. 
 \end{equation}

On the other hand, the RHS of (\ref{7.X2}) can be rewritten by substituting for $ s_{\kappa}(1/\mathbf z) s_{\kappa}(1/\mathbf w) $ 
according to (\ref{16.jlX}) with $\alpha = 1$, $N \mapsto k$, $A=D_1^{-1}$, $B = D_2^{-1}$, simplified by making
further use of (\ref{7.X1}), now with $\alpha = 1$, $N=k$ and $a=0$. This shows that (\ref{7.X2}) is equal to
 \begin{equation}\label{7.X2b} 
  \prod_{l=1}^k (w_l z_l)^N   \sum_{ \kappa \subseteq (N)^k}  [N]_{\kappa'}^{(1)} {s_\kappa((1)^k) \over  [k]_{\kappa}^{(1)} }
  \langle s_\kappa(D_1^{-1} X^\dagger D_2^{-1} X) \rangle_{X \in {\rm GinUE}_k}.
   \end{equation}
  But from  \cite[Eqns.~(12.104), (12.105), (13.1) and the fact that $h_\kappa|_{\alpha =1} =   h_\kappa'|_{\alpha =1}$]{Fo10} we have
   \begin{equation}\label{7.X2c} 
  [N]_{\kappa'}^{(1)} {s_\kappa((1)^k) \over  [k]_{\kappa}^{(1)} } = s_{\kappa'}((1)^N).
  \end{equation}
  With this noted, the sum over $\kappa    \subseteq (N)^k $ can be performed using the   dual Cauchy identity (\ref{SM2}) with $\alpha = 1$
  to give in place of (\ref{7.X2b}) the expression
     \begin{equation}\label{7.X2d}
    \prod_{l=1}^k (w_l z_l)^N    \Big \langle \prod_{l=1}^k (1 + \gamma_l)^N \Big \rangle,
   \end{equation}   
 where $\{\gamma_l\}$ are the eigenvalues   of $Z_1^{-1} X^\dagger Z_2^{-1} X$. Use of the block determinant identity
   \begin{equation}\label{7.X5} 
\det   \begin{bmatrix} A & B \\ C & D \end{bmatrix}   = \det AD \det (\mathbb I_k - A^{-1} B D^{-1} C),
   \end{equation}
   identifies (\ref{7.X2d}) with the RHS  of (\ref{6.0v+}).
   \end{proof}

%
%

 \begin{remark} $ {}$ \\
 1.~The duality (\ref{6.0v+}) can be generalised so that $Z$ and $\bar{Z}$ are replaced by $\Omega Z$ and $\Sigma \bar{Z}$ on the LHS
 \cite[first statement of Th.~4.1]{Se23}. Such generalised dualities have found application to the study of the spectral density in
 the Ginibre ensemble subject to an additive or multiplicative low rank perturbation \cite{LZ22}.\\
 2.~An alternative to the use of Schur polynomials in the derivation of the just mentioned generalisation of (\ref{6.0v+}) is to make use of
a certain two-sided matrix diffusion operator identity $\Delta_X Q(X,Y) \propto  
\Delta_Y Q(X,Y)$  for $Q(X,Y)$ a particular kernel function relating to the averages and $\Delta_Z$ a (generalised) Laplacian
\cite{Gr16}, \cite{LZ24}. \\
3.~For $H_1,H_2$ independent GUE matrices, the elliptic Ginibre ensemble is specified by matrices of the form
$\sqrt{1 + \tau} H_1 + \sqrt{1 - \tau} i H_2$, $|\tau| \le  1$; see e.g.~\cite[\S 2.3]{BF24}. This ensemble thus interpolates between the
GUE ($\tau = 1$) and the GUE ($\tau = 0$). Grassmann integration methods have been used to deduce the duality identity for products of
characteristic polynomials \cite{AV03}, \cite[Eq.~(5.3)]{Se23}
\begin{multline}
\bigg \langle \prod_{j=1}^k \det ( J - z_j \mathbb I_N) \det (J^\dagger - w_j \mathbb I_N) \bigg \rangle_{J \in {\rm GinUE}_N} \\
\propto  \bigg \langle \prod_{j=1}^k \det  \begin{bmatrix} - i D_1 + \sqrt{2\tau} A & C \\
C^\dagger & -i D_2 + \sqrt{2 \tau} B \end{bmatrix}^N\bigg \rangle_{\substack{A,B \in {\rm GUE}_k  \\ C \in {\rm GinUE}_k}},
 \end{multline} 
 where $D_1 = {\rm diag} (z_1,\dots,z_k)$ and $D_2 = {\rm diag} (w_1,\dots,w_k)$. In the GinUE case $\tau = 0$, this agrees with (\ref{6.0v+}).
 Also, taking the limit $w_1,\dots, w_k \to \infty$ gives agreement with the GUE result (\ref{2.X1}).
 \end{remark}
 
 Inspection of the proof of Proposition \ref{P5.2} shows that a key role is played by the eigenvalue distribution of $Z Z^\dagger$ belonging to
 a classical Hermitian ensemble (Laguerre ensemble), which moreover permit the structured formula (\ref{7.X1}) for the Jack polynomial
 average. In addition, to identify the evaluation of the LHS of the duality identity in the form of (\ref{7.X2}) with the RHS, it is crucial that the Jack
 polynomial average contains as a factor $P_\kappa^{(1)}((1)^N) [N]_\kappa^{(1)}$. 
 
 The first of these requirements is met by the ensemble of $N \times N$ non-Hermitian matrices
  formed by deleting $K$ rows and
 $K$ columns from an $(N+K) \times (N+K)$ unitary matrix chosen with Haar measure. We will denote this ensemble by ${\rm TrUE}_{(N,K)}$.
  One has that the distribution of squared singular values are the case
 $a_1 = 0$ and $a_2 = K - N$ of the Jacobi ensemble ME${}_{2,N}[x^{a_1} (1 - x)^{a_2}]$ \cite{FS03}. This requirement is met
 too by the so-called complex spherical ensemble of non-Hermitian matrices $G_1^{-1} G_2$, where $G_1, G_2$ are independent
 GinUE matrices, as it is by the generalisation of the complex spherical ensemble defined by the construction $(G^\dagger G)^{-1/2} X$,
 where $G$ is an $(N+K) \times N$  complex standard Gaussian matrix, and $X$ is a GinUE matrix (the spherical ensemble
 corresponds to the case $K=0$). Denoting this latter ensemble
 by SrUE${}_{(N,K)}$, one has that
  the squared singular values are the case $b_1=0$, $b_2=K$,  of the so-called Jacobi prime
 ensemble \cite[\S 2.4]{FI18} (a particular linear fractional transformation of the Jacobi ensemble)  ME${}_{2,N}[x^{b_1} (1 + x)^{-b_2-2N}]$,
 supported on $x > 0$ \cite{Kr09}, \cite[Exercises 3.6 q.3]{Fo10}. 
 
 It is also the case that the Jacobi and Jacobi prime ensembles admit structured formulas for
 averages of Jack polynomials. Thus we have \cite{Ka97g}, \cite{Ka93}
 (see \cite[Eq.~(12.143)]{Fo10})
\begin{equation}\label{MKK}
\langle P_\kappa^{(\alpha)}(\mathbf x) \rangle_{{\rm ME}_{2/\alpha,N}[\lambda^{a_1}(1 - \lambda)^{a_2}]}
= P_\kappa^{(\alpha)}((1)^N)
\frac{[a_1+1+(N-1)/ \alpha ]_\kappa^{(\alpha)}}{[a_1 + a_2  + 2 +2(N-1)/\alpha ]_\kappa^{(\alpha)}},
 \end{equation}
and \cite{Wa05},  \cite{FS09}
\begin{equation}\label{W}
\langle  P_\kappa^{(\alpha)}(\mathbf x)  \rangle_{{\rm ME}_{2/\alpha,N}[\lambda^{b_1}(1+\lambda)^{-b_1 - b_2 - 2 - 2(N-1)/\alpha} ]}
  = P_\kappa^{(\alpha)}((-1)^N) 
{[b_1 + 1 + (N-1)/ \alpha]_\kappa^{(\alpha)} \over [-b_2]_\kappa^{(\alpha)} }.
\end{equation}

It fact seeking a duality identity analogous to (\ref{6.0v+}) for  ${\rm TrUE}_{(N,K)}$ gives rise to an average involving
 SrUE${}_{(N',K')}$ for certain $N',K'$ \cite{SSD23}. Similarly for the case of starting with the average of a product of
 characteristic polynomials over SrUE${}_{(N,K)}$.
 
\begin{prop}\label{P5.2a}
We have
\begin{equation}\label{6.0v+B}
\Big  \langle \prod_{l=1}^k \det(z_l \mathbb I_N - Z)   \det({w}_l \mathbb I_N - \bar{Z}) \Big \rangle_{Z \in {\rm TrUE}_{(N,K)}} =
\bigg \langle \det \begin{bmatrix} D_1 & -Y \\ Y^\dagger & D_2 \end{bmatrix}^N \bigg \rangle_{Y \in {\rm SrUE}_{(k,N+K)}},
 \end{equation}
 where $D_1, D_2$ are diagonal matrices with diagonal entries $\mathbf z$ and $\mathbf w$ respectively. Also
 \begin{equation}\label{6.0vA}
\Big  \langle \prod_{l=1}^k \det(z_l \mathbb I_N - Z)   \det({w}_l \mathbb I_N - \bar{Z}) \Big \rangle_{Z \in {\rm SrUE}_{(N,K)}} =
\bigg \langle \det \begin{bmatrix} D_1 &  - Y \\ Y^\dagger & D_2 \end{bmatrix}^N \bigg \rangle_{Y \in {\rm TrUE}_{(k,K-k)}},
 \end{equation}
where it is required $K \ge k$.

\end{prop} 

\begin{proof}
Following the working which lead to (\ref{7.X2}), with use of (\ref{MKK}) replacing use of (\ref{7.X1}), shows that the LHS is
equal to 
  \begin{equation}\label{7.X4+} 
 \prod_{l=1}^k (w_l z_l)^N   \sum_{ \mu \subseteq (N)^k}   s_{\mu}(1/\mathbf z) s_{\mu}(1/\mathbf w) { [N]_{\mu'}^{(1)} \over  [N+K]_{\mu'}^{(1)}  }.
 \end{equation}  
 In this we replace the sum over $\mu \subseteq (N)^k$ by the sum over $\mu \subseteq (k)^N$, which is valid provided we replace $\mu$ by
 $\mu'$ in the summand. We then return to the identity  (\ref{7.X2}), now with $N \mapsto k$
  to substitute for $s_{\mu'}(1/\mathbf z) s_{\mu'}(1/\mathbf w)$ analogous to the strategy of the proof of Proposition \ref{P5.2} 
  (this requires too further use of (\ref{MKK}), now with
  $\alpha = 1$,  $b_2 = K + N$). Further use of the dual Cauchy identity (\ref{SM2}) with $\alpha = 1$,
  and use too of  (\ref{7.X5})
 identifies the resulting expression with (\ref{7.X4+}).
  This line of working provides too the derivation of (\ref{6.0vA}).
\end{proof}

Consider the special case of (\ref{6.0v+}) $D_1 = z \mathbb I_k, D_2 = \bar{z}  \mathbb I_k$, with (\ref{7.X5}) used to rewrite the determinant
on the RHS, and then with the change of variables $W=Y Y^\dagger$. The duality now takes the form \cite[Eq.~(5.6) with $\Sigma = \mathbb I_N$]{FR09}
  \begin{equation}\label{7.Y}
  \langle | \det (z \mathbb I_N - Z) |^{2k} \rangle_{Z \in {\rm GinUE}_N} = |z|^{2k N} \Big \langle \det \Big ( \mathbb I_k + {1 \over |z|^2} W \Big )^N \Big \rangle_{W \in {\rm ME}_{2,k}[e^{-x}]},
  \end{equation}  
  thus relating a non-Hermitian average to one over an Hermitian ensemble. With $k$ chosen to be proportional to $N$, this identity has recently been used to
  study the large $N$ form of the Coulomb gas partition function interpretation of the LHS \cite{BSY24}; it was used earlier in \cite{DS22} to determine the
  asymptotics with $k$ fixed, with the result
   \begin{equation}\label{7.Y1} 
   \langle | \det (z \mathbb I_N - N^{-1/2} Z) |^{2k} \rangle_{Z \in {\rm GinUE}_N}   = N^{k^2/2} e^{k N (|z|^2 - 1)}
   {(2 \pi )^{k/2} \over G(1+k)} \Big ( 1 + {\rm o}(1) \Big ), \quad |z| < 1,
   \end{equation}   
   where $G(z)$ denotes the Barnes G-function.
  From \cite{WW19} this is known to be valid for continuous $k > -1$, while \cite{DS22} also gives an
   extension of (\ref{7.Y1}) to the case that $|1-z|\sqrt{N}$ is of order unity. 
   
    Identities analogous to (\ref{7.Y}) follow from (\ref{6.0v+B}) and (\ref{6.0vA}). In the case of ${\rm TrUE}_{N,K}$, with
    $N/(N+K) =: \mu$, an application is the asymptotic formula analogous to (\ref{7.Y1}) \cite{SS24}, \cite{Se23}
   \begin{multline}\label{7.Y1a} 
   \langle | \det (z \mathbb I_N -  Z) |^{2k} \rangle_{Z \in {\rm TrUE}_{N,K}}   \\= N^{k^2/2}\mu^{Nk} 
   \Big ( {1 - \mu \over 1 - | z|^2} \Big )^{N k (1 - 1/\mu)} 
   \Big ( {\sqrt{1 - \mu} \over 1 - | z|^2} \Big )^{k^2}  
   {(2 \pi )^{k/2} \over G(1+k)} \Big ( 1 + {\rm o}(1) \Big ), \quad |z| < |\mu|^{1/2},
   \end{multline}    
    (note that the constant factor is the same as in (\ref{7.Y1})). Potential theoretic aspects of the Coulomb gas interpretation
    of the LHS of (\ref{7.Y1a}) in the case that $k$ is proportional to $N$ is the subject of the recent work \cite{B25+}.
   
  The duality formula \cite[Eq.~(5.6)]{FR09} for $\Sigma = A^\dagger A \ne \mathbb I_N$ referred to in the above paragraph in relation to
  (\ref{7.Y}) is distinct from (\ref{6.0v+}), and moreover permits analogues that are distinct from those of Proposition \ref{P5.2a}.
  There are two stages to such identities. The first is to express the averages in terms of Jack polynomial based hypergeometric
  functions, and the second is the alternative form of the latter as random matrix averages over ensembles
  of size $r$.
  
  \begin{prop}\label{P5.2b}
We have
\begin{align}\label{5.33}
& \langle | \det ( x \mathbb I_N - A X) |^{2r} \rangle_{X \in {\rm GinUE}_N} = |x|^{2r N}
{\vphantom{F}}_2^{\mathstrut} F_0^{(1)}(-r,-r;|x|^{-2} A A^\dagger) \nonumber \\
& \langle | \det ( x \mathbb I_N - A X) |^{2r} \rangle_{X \in {\rm TrUE}_{N,K}} = |x|^{2r N}
{\vphantom{F}}_2^{\mathstrut} F_1^{(1)}(-r,-r;N+K ;|x|^{-2} A A^\dagger)  \nonumber \\
& \langle | \det ( x \mathbb I_N - A X) |^{2r} \rangle_{X \in {\rm SrUE}_{N,K}} = |x|^{2r N}
{\vphantom{F}}_2^{\mathstrut} F_1^{(1)}(-r,-r;- K ;- |x|^{-2} A A^\dagger), \quad (K > r).
\end{align}
Consequently
\begin{align}\label{5.34}
& \langle | \det ( x \mathbb I_N - A  X) |^{2r} \rangle_{X \in {\rm GinUE}_N} = 
\Big \langle \prod_{l=1}^r \det  ( |x|^2 \mathbb I_N + t_l A A^\dagger   ) \Big \rangle_{\mathbf t \in {\rm ME}_{2,r}[e^{-\lambda}]}\nonumber \\
& \langle | \det ( x \mathbb I_N - A X) |^{2r} \rangle_{X \in {\rm TrUE}_{N,K}} \propto \det (|x|^2 \mathbb I_N - A A^\dagger)^r
\Big \langle \prod_{l=1}^r \det  (t_l  \mathbb I_N - \Lambda    )    \Big \rangle_{\mathbf t \in {\rm ME}_{2,r}[\lambda^K \mathbbm 1_{0<\lambda<1}]}
 \nonumber \\
& \langle | \det ( x \mathbb I_N - A X) |^{2r} \rangle_{X \in {\rm SrUE}_{N,K}} =
\Big \langle \prod_{l=1}^r \det  (  |x|^2 \mathbb I_N + t_l A A^\dagger  )    \Big \rangle_{\mathbf t \in {\rm ME}_{2,r}[(1 - \lambda)^{K-2r} \mathbbm 1_{0<\lambda<1}]},
\end{align}
where $\Lambda = - AA^\dagger(|x|^2 - A A^\dagger)^{-1}$, and in the final equation it is required that $K \ge 2r$.

\end{prop}
  
  \begin{proof}
  The derivation of (\ref{5.33}) is similar in all three cases. For definiteness we will consider the first only. Since the parameter
  $x$ can be inserted by scaling $\Sigma$, it suffices to consider the case $x=1$. With this done, we write
  $| \det ( x \mathbb I_N - A X) |^{2r}  = \det ( \mathbb I_N - A X)^r  \det ( \mathbb I_N - A^\dagger X^\dagger )^r$,
  and then expand each of these powers of characteristic polynomials according to 
      \begin{equation}\label{7.X6}  
    \det (\mathbb I_N - Q)^{r} =   \sum_{ \kappa \subseteq (N)^{r}} { [-r]_\kappa^{(1)} \over | \kappa|!} C_\kappa^{(1)}(Q),
  \end{equation}  
  (as implied by (\ref{2.4E}) with $\alpha = 1$)
 with parameters
  appropriately identified. The computation of
  the average is therefore reduced to computing $\langle C_\kappa^{(1)}(A X) C_\kappa^{(1)}(A^\dagger X^\dagger) \rangle_{X \in {\rm GinUE}_N}$,
  which we do using (\ref{16.9.2X}) to further reduce the problem to computing the average $\langle C_\kappa^{(1)}(XX^\dagger)  \rangle_{X \in {\rm GinUE}_N}$.
 This latter average has appeared in the proof of Proposition \ref{P5.2}, with its method of computation using (\ref{7.X1}) detailed therein. Use of
 (\ref{15.ckp}), 
 \begin{equation}\label{7.X7}  
  {[k]_\kappa^{(1)} \over C_\kappa^{(1)} ((1)^k)} = {1 \over |\kappa|!}(h_\kappa'|_{\alpha = 1})^2
    \end{equation}  
    (which is equivalent to (\ref{7.X2c}))
  and (\ref{3.40}) then allows the resulting expression for the LHS of the first average in
 (\ref{5.33}) can be identified with the Jack polynomial based hypergeometric function on the RHS.
 
 The task of going from (\ref{5.33}) to (\ref{5.34}) is to identify the Jack polynomial based hypergeometric functions with averages over
 matrix ensembles of size $r$. In the case of the first appearing ${\vphantom{F}}_2^{\mathstrut} F_1^{(1)}$ this is done by  applying the transformation
 (\ref{2.4H}) and then making use of the theory in the first paragraph of the proof of Proposition \ref{P2.6} with $N \mapsto r$, $p \mapsto N$
 in (\ref{2.4Ea}). For the second appearing ${\vphantom{F}}_2^{\mathstrut} F_1^{(1)}$, we consider its series form (\ref{3.40}) with the sum over $\kappa$
 replaced by a sum over $\kappa'$. The identity 
  \begin{equation}\label{7.X3} 
  [u]_\kappa^{(\alpha)} = (-\alpha)^{-|\kappa|} [- \alpha u ]_{\kappa'}^{(1/\alpha)},
 \end{equation}
 can then be used to rewrite the generalised Pochhammer symbols back in terms of $\kappa$.
 Also, we have that \cite[displayed equation above Eq.~(5.4) and Eq.~(3.18)]{FR09}
 \begin{equation}\label{7.Za} 
 h_{\kappa'}' = {[r \alpha]_\kappa^{(1/\alpha)} \over P_\kappa^{(1/\alpha)}((1)^r)}
 \end{equation}
 (note that this is independent of $r$ for the number of parts of $\kappa'$ less that or equal to $r$).
 Hence
 \begin{equation}\label{7.Zb}  
  {\vphantom{F}}_2^{\mathstrut} F_1^{(\alpha)}(-r,-b;-K;-A A^\dagger) = \sum_{\kappa' \subseteq (N)^r}  {[ \alpha b]_\kappa^{(1/\alpha)}  \over
[ \alpha K ]_\kappa^{(1/\alpha)}   } P_\kappa^{(1/\alpha)}((1)^r)   P_{\kappa'}^{(\alpha)}(A A^\dagger).
 \end{equation}
Except for the factor of $(-1)^{|\kappa|}$ we see from (\ref{MKK}) that the prefactors of $P_\kappa^{\alpha)}(A A^\dagger)$ in this expression can be replaced by the
average over ME${}_{2/\alpha,N}[\lambda^{a_1}(1 - \lambda)^{a_2}]$ with $\alpha \mapsto 1/\alpha$, $N = r$, $a_1 =  - 1 +(b-r+1)\alpha$ and
$a_2 = -1 +(K-b-r+1)\alpha$. After using this as a substitution, the sum over $\kappa'$ can be carried out according to the dual
Cauchy identity (\ref{SM2}) to give the ensemble average form
 \begin{equation}\label{7.Zc}  
 {\vphantom{F}}_2^{\mathstrut} F_1^{(\alpha)}(-r,-b;-K;A A^\dagger) =  \Big \langle \prod_{l=1}^r \det (\mathbb I_N - x_l A A^\dagger) \Big \rangle_{\mathbf x \in {\rm ME}_{2 \alpha,r}[\lambda^{a_1}(1 - \lambda)^{a_2}]}
 \bigg |_{\substack{a_1 =  - 1 +(b - r+1)\alpha \\ a_2 =  -1 + (K - b - r+1)\alpha}},
 \end{equation}
 valid provided $a_1,a_2 > -1$.
 Application of this gives the final duality in (\ref{5.34}).
 
  In the case of  the evaluation in terms of ${\vphantom{F}}_2^{\mathstrut} F_0^{(1)}$, 
  which is the first identity in (\ref{5.33}), one has available the identity \cite{BO05}
  \begin{equation}\label{7.Z} 
 {\vphantom{F}}_2^{\mathstrut} F_0^{(\alpha)}(-a,1+b-a+(n-1)/\alpha;x_1^{-1},\dots,x_n^{-1}) \propto (x_1 \cdots x_n)^{-a}
 {\vphantom{F}}_1^{\mathstrut} F_1^{(\alpha)}(-a;-b;-x_1,\dots,-x_n),
  \end{equation}
  valid for $a \in \mathbb Z^+$ and $b>a$. Scaling $x_j \mapsto \beta x_j/2 a_2$ and then taking the limit $a_2 \to \infty$ we have
  that ${\vphantom{F}}_1^{\mathstrut} F_1^{(\alpha)}$ with the first parameter a negative integer can
  be written as the average
  \begin{equation}\label{4.41} 
   {\vphantom{F}}_1^{\mathstrut} F_1^{(\alpha)}(-a;(a_1+n)/\alpha;x_1,\dots,x_n) \propto  \Big \langle \prod_{j=1}^n \prod_{l=1}^a (x_j - y_l) \Big \rangle_{\mathbf y \in {\rm ME}_{2 \alpha,a}[x^{a_1} e^{-\beta x /2}]}.
   \end{equation}
  Use of this with $\alpha = 1$, $a_1=0$, $a=n=r$ gives the first average in (\ref{5.34}). 
  \end{proof}
  
  \begin{remark} ${}$ \\
 1.~In the case $K=0$, by its definition TrUE${}_{N,K}$ becomes equal to $U(N)$ with Haar measure. For this latter ensemble, a duality identity of Fyodorov and
  Khoruzhenko \cite{FK07} states
   \begin{equation}\label{7.Zd}  
 \langle | \det ( x \mathbb I_N - A U) |^{2r} \rangle_{U \in U(N)}    \propto \Big \langle \prod_{l=1}^r \det ( |x|^2 \mathbb I_N + t_l A A^\dagger) \Big \rangle_{\mathbf t \in {\rm ME}_{2,r}[(1+ \lambda)^{-N-2r}]}.
  \end{equation}
  After a change of variables $t_l/(1+t_l) \mapsto t_l$ on the RHS, this can be checked to agree with the case $K=0$ of the second duality relation in (\ref{5.34}). 
  When $A$ is a multiplicative rank 1 perturbation of the identity, $A = {\rm diag} \, (a,1,\dots,1)$ with $|a|<1$, the eigenvalues of $A U$ relate to the point process for the
  zeros of the random power series $1/\mu - \sum_{j=1}^N c_j z^j$ where each $c_j$ is an independent standard complex Gaussian \cite{FI19} (this requires setting
  $a = 1/(\mu \sqrt{N})$ and taking $N \to \infty$). This point process is further studied in the recent work \cite{FKP24}.
  \\
  2.~One can check that the third identity in (\ref{5.34}) with $A = \mathbb I_N$ is consistent with the case $z_1 = \cdots = z_k =z$, $w_1 = \cdots = w_k =\bar{z}$ of
  (\ref{6.0vA}). Also, as with the first two identities in (\ref{5.34}), setting $A = \mathbb I_N$ allows for a Coulomb gas interpretation. Aspects of this in the case
  $r$ proportional to $N$ have been the subject of the recent work \cite{BFL25}.
  
  \end{remark}

\subsubsection{Real and quaternion entry matrices}
Under this heading attention will first be focussed on powers of characteristic polynomial averages
extending Proposition \ref{P5.2b}. As some preliminaries, we recall the notation GinOE${}_N$ (GinSE${}_N$) for 
$N \times N$ standard real (quaternion) Gaussian random matrices. We denote by
${\rm TrOE}_{N,K}$ (${\rm TrSE}_{N,K}$) the 
ensemble of $N \times N$ non-Hermitian matrices
  formed by deleting $K$ rows and
 $K$ columns from an $(N+K) \times (N+K)$ unitary matrix with real (quaternion) elements 
 chosen with Haar measure. Also, we use the notation ${\rm SrOE}_{N,K}$ (${\rm SrSE}_{N,K}$)
 for the ensemble of matrices of the form $(G^\dagger G)^{-1/2} X$, where $G$ is  
an $(N + K) \times N$ real (quaternion) standard Gaussian matrix, and $X$ is a GinOE (GinSE) matrix.
A feature of the eigenvalues in both the real and quaternion cases is that they appear in complex
conjugate pairs, and typically in the real case some eigenvalues will also be real \cite{Ed97}.
Further, in the quaternion case, when we define the characteristic polynomial $ \det ( x \mathbb I_N - B)$
 of a matrix $B$,
the viewpoint is that entries are quaternions in both $\mathbb I_N$ and
$B$. Thus as a complex matrix the size of $ x \mathbb I_N - B$ is $2N \times 2N$, and the
characteristic polynomial is of degree $2N$.

\begin{prop} (\cite{FR09} in relation to the first two equations and \cite{SSD23} for the next two.)
We have
\begin{align}\label{5.33a}
& \langle  \det ( x \mathbb I_N - A X) ^{r} \rangle_{X \in {\rm GinOE}_N} =  x^{r N}
{\vphantom{F}}_2^{\mathstrut} F_0^{(2)}(-r/2, (-r+1)/2 ;2 |x|^{-2} A A^\dagger) \nonumber \\
& \langle  \det ( x \mathbb I_N - A X) ^{r} \rangle_{X \in {\rm GinSE}_N} = x^{2r N}
{\vphantom{F}}_2^{\mathstrut} F_0^{(1/2)}(-r,-r-1;|x|^{-2} A A^\dagger/2) \nonumber \\
& \langle  \det ( x \mathbb I_N - A X)^{r} \rangle_{X \in {\rm TrOE}_{N,K}} = x^{r N}
{\vphantom{F}}_2^{\mathstrut} F_1^{(2)}(-r/2,(-r+1)/2;(N+K)/2 ;|x|^{-2} A A^\dagger)  \nonumber \\
& \langle  \det ( x \mathbb I_N - A X)^{r} \rangle_{X \in {\rm TrSE}_{N,K}} = x^{2r N}
{\vphantom{F}}_2^{\mathstrut} F_1^{(1/2)}(-r,-r-1;2(N+K) ;|x|^{-2} A A^\dagger)  \nonumber \\
& \langle  \det ( x \mathbb I_N - A X)^{r} \rangle_{X \in {\rm SrOE}_{N,K}} = x^{r N}
{\vphantom{F}}_2^{\mathstrut} F_1^{(2)}(-r/2,(-r+1)/2;- (K-1)/2 ;- |x|^{-2} A A^\dagger) \nonumber  \\
& \langle  \det ( x \mathbb I_N - A X) ^{r} \rangle_{X \in {\rm SrSE}_{N,K}} = x^{2r N}
{\vphantom{F}}_2^{\mathstrut} F_1^{(1/2)}(-r,-r-1;- 2K-1 ;- |x|^{-2} A A^\dagger),
\end{align}
where in the last two identities we require $K-1 > r$ and $2K +1> r$ respectively.
\end{prop}

\begin{proof} The derivations begin with the use of (\ref{7.X6}), 
$(N,r) \mapsto (2N,r)$ in the quaternion case. Also, we replace $C_\kappa^{(1)}(\mathbf x)$ by $|\kappa|! 
s_\kappa(\mathbf x)/ h_\kappa'|_{\alpha = 1}$ as is consistent with (\ref{15.ckp}). The next step is to make use
of (\ref{16.9.1X}) or (\ref{16.9.3X}) as appropriate, together with the key identities \cite{BF03}, \cite{FR09}
\begin{align}\label{pwX}
& {1 \over |\kappa|! 2^{|\kappa|}} h_{2 \kappa}' |_{\alpha = 1} = {2^{|\kappa|} [N/2]_\kappa^{(2)} \over C_\kappa^{(2)}((1)^N)}, \qquad
 {1 \over |\kappa|! 2^{|\kappa|}} h_{\kappa^2}' |_{\alpha = 1} = { [2 N]_\kappa^{(1/2)} \over 2^{|\kappa|} C_\kappa^{(1/2)}((1)^{N})}, \nonumber \\
& [u]_{2 \kappa}^{(1)} = 2^{2 |\kappa|} [u/2]_\kappa^{(2)}  [(u+1)/2]_\kappa^{(2)}, \qquad
 [u]_{\kappa^2}^{(1)} =  [u]_\kappa^{(1/2)}  [u-1]_\kappa^{(1/2)},
 \end{align} 
 inter-relating Jack polynomial quantities for $\alpha = 1$ with $\alpha = 2, 1/2$.
This reduces the task to computing $\langle P_\kappa^{(2/\beta)}(X X^\dagger) \rangle$ for the ensemble of interest,
which in turn, after a change of variables, reduces the problem to computing $\langle P_\kappa^{(2/\beta)}(\mathbf x) \rangle$,
where the average now is over the corresponding squared singular values. For GinOE${}_N$ (GinSE${}_N$) the squared
singular values belong to ME${}_{\beta,N}[\lambda^{\beta/2 - 1} e^{-\beta \lambda/2}]$ with $\beta = 1$ ($\beta = 4$)
\cite[case $n-m$ of (3.16)]{Fo10};
for TrOE${}_N$ (TrSE${}_N$) the squared
singular values belong to ME${}_{\beta,N}[\lambda^{\beta/2 - 1} (1 - \lambda)^{(\beta/2)(K-N+1)-1}]$ with $\beta = 1$ ($\beta = 4$)
\cite[Eq. (3.113) with $N \mapsto N + K$, $n_1, n_2 \mapsto N$]{Fo10};
SrOE${}_N$ (SrSE${}_N$) the squared
singular values belong to ME${}_{\beta,N}[\lambda^{\beta/2 - 1} (1 + \lambda)^{-(\beta/2)(K+2N)]}$ with $\beta = 1$ ($\beta = 4$)
\cite[Eq. (3.81)]{Fo10}. Consequently the sought averages  of $ P_\kappa^{(2/\beta)}(\mathbf x) $ over the squared singular values
can be computed by making use of  (\ref{7.X1}) or (\ref{W}) as appropriate. Assembling the results, the series form
of various Jack polynomial based hypergeometric functions with $\alpha = 2$ (real case) and $\alpha = 1/2$ (quaternion case) 
are recognised.
\end{proof}   

We can now proceed to express the averages given in terms of Jack polynomial based hypergeometric functions in the above proposition,
in terms of random matrix averages of sizes relating to the power according to the formulas identified in (\ref{7.Zc}), (\ref{7.Z}) and
(\ref{4.41}).

\begin{prop}
(\cite{FR09} in relation to the Ginibre ensembles and \cite{SSD23} in relation to the truncated ensembles.)
Let $\Lambda$ be as in Proposition \ref{P5.2b}. We have
  \begin{align}\label{7.Zi} 
 \langle \det (\mathbb I_N - A X)^r \rangle_{X \in {\rm GinOE}_N}  & =  \Big  \langle \prod_{l=1}^{\lfloor r/2 \rfloor} \det ( \mathbb I_N + t_l A A^\dagger) 
  \Big \rangle_{\mathbf t \in {\rm ME}_{4,\lfloor r/2 \rfloor} [\lambda^\nu e^{-\lambda}]} \nonumber \\
  \langle \det (\mathbb I_N - A X)^r \rangle_{X \in {\rm TrOE}_{N,K}}  &  \propto
  \det ( \mathbb I_N - A A^\dagger)^{\lfloor r/2 \rfloor} 
\Big \langle \prod_{l=1}^{\lfloor r/2 \rfloor} \det  (t_l  \mathbb I_N - \Lambda    )    \Big \rangle_{\mathbf t \in {\rm ME}_{4,\lfloor r/2 \rfloor}[\lambda^K (1 - \lambda)^\nu]} \nonumber  \\
 \langle \det (\mathbb I_N - A X)^r \rangle_{X \in {\rm SrOE}_{N,K}}  &  =
 \Big \langle \prod_{l=1}^{\lfloor r/2 \rfloor} \det  (\mathbb I_N +  t_l   A A^\dagger    )    \Big \rangle_{\mathbf t \in {\rm ME}_{4,\lfloor r/2 \rfloor}[\lambda^\nu (1 - \lambda)^{K-2r+1)}]}, 
  \end{align}
  where $\nu = 0$ ($r$ even) and $\nu = 2$ ($r$ odd). Also
  \begin{align}\label{7.Zj} 
 \langle \det (\mathbb I_N - A X)^r \rangle_{X \in {\rm GinSE}_N}  & =  \Big  \langle \prod_{l=1}^{r} \det ( \mathbb I_N + t_l A A^\dagger) 
  \Big \rangle_{\mathbf t \in {\rm ME}_{1,r} [ e^{-\lambda}]} \nonumber \\
  \langle \det (\mathbb I_N - A X)^r \rangle_{X \in {\rm TrSE}_{N,K}}  &  \propto
  \det ( \mathbb I_N - A A^\dagger)^{r} 
\Big \langle \prod_{l=1}^{r} \det  (t_l  \mathbb I_N - \Lambda    )    \Big \rangle_{\mathbf t \in {\rm ME}_{1,r}[\lambda^K  \mathbbm 1_{0<\lambda<1}]} \nonumber  \\
 \langle \det (\mathbb I_N - A X)^r \rangle_{X \in {\rm SrSE}_{N,K}}  &  =
 \Big \langle \prod_{l=1}^{r} \det  (\mathbb I_N +  t_l   A A^\dagger    )    \Big \rangle_{\mathbf t \in {\rm ME}_{1,r}[(1 - \lambda)^{K-r-1/2}]}.
  \end{align} 

\end{prop}

\begin{remark}
We note that the first  identity in (\ref{7.Zi}) with $r=2$ is consistent with (\ref{6.0a}).
\end{remark}

We next consider analogues of the dualities for products of non-Hermitian random matrix ensembles
with complex entries (\ref{6.0v+}), (\ref{6.0v+B}) and (\ref{6.0vA}), in the case of real entries or quaternion
entries. In the case of GinUE in (\ref{6.0v+}) we saw that the RHS also considered of an average over a
particular GinUE. In the case of GinOE (GinSE) it turns out that the duality gives rise to an average over 
complex anti-symmetric (complex symmetric) matrices \cite{TZ14}, \cite{NK02}, \cite{Se23}. Let the set of these latter
matrices be denoted $\mathcal A_N(\mathbb C)$ and $\mathcal S_N(\mathbb C)$ respectively.
Fundamental in their derivation
are integration formulas over zonal polynomials analogous to (\ref{16.9.1X}) and (\ref{16.9.3X}) \cite{FS09} (see too
\cite[pg.~44]{Se23} for a clear statement).

\begin{prop}
Let $A$ be an $N \times N$ matrix.
For  $X \in \mathcal S_N(\mathbb C)$ chosen from a measure unchanged by $X \mapsto U X U^T$, $U \in U(N)$
we have
\begin{equation}\label{5.46}
\Big \langle C_\kappa^{(2)}(A X A^T X^\dagger) \Big \rangle_X = {s_{2\kappa}(A) \over s_{2\kappa}((1)^N)} \langle C_\kappa^{(2)}(X X^\dagger) \rangle_{X}.
\end{equation}
Also, for  $X \in \mathcal A_N(\mathbb C)$ chosen from a measure unchanged by $X \mapsto U X U^T$, $U \in U(N)$
we have
\begin{equation}\label{5.46a}
\Big \langle C_\kappa^{(1/2)}(A X A^T X^\dagger) \Big \rangle_X = {s_{\kappa^2}(A) \over s_{\kappa^2}((1)^N)} \langle C_\kappa^{(1/2)}(X X^\dagger) \rangle_{X}
\end{equation}
(note that here the matrices in $C_\kappa^{(1/2)}(\cdot)$ for $N$ even are doubly degenerate and so are only to be included once, and similarly for $N$ odd
apart from a zero eigenvalue which also is not included).
\end{prop}

\begin{cor}\label{C5.2}
Denote by G$\mathcal A_N$ (G$\mathcal S_N$)
the set of $N \times N$ anti-symmetric complex (symmetric complex) matrices chosen with measure proportional to $e^{-{\rm Tr} X X^\dagger/2}$
($e^{-{\rm Tr} X X^\dagger}$). We have
\begin{equation}\label{5.47}
\Big \langle \prod_{j=1}^{2k} \det (G - z_j \mathbb I_N) \Big \rangle_{G \in {\rm GinOE}_N} =
\Big \langle \det \begin{bmatrix} X & Z \\ -Z & X^\dagger  \end{bmatrix}^{N/2} \Big \rangle_{X \in {\rm G} \mathcal A_{2k}} 
\end{equation}
and
\begin{equation}\label{5.48}
\Big \langle \prod_{j=1}^{2k} \det (G - z_j \mathbb I_{N}) \Big \rangle_{G \in {\rm GinSE}_N} =
\Big \langle \det \begin{bmatrix} X & Z \\ -Z & X^\dagger  \end{bmatrix}^{N} \Big \rangle_{X \in {\rm G} \mathcal S_{2k}}. 
\end{equation}
\end{cor}

\begin{proof}
Following \cite{Se23}, we adopt a strategy analogous to that used in the proof of Proposition \ref{P5.2}, giving the details
only of (\ref{5.47}). 
Consider the LHS. Making use of the dual Cauchy identity (\ref{SM2}) with $\alpha = 1$, then (\ref{16.9.1X})
with $A=Z^{-1}$ shows that it can be written
\begin{multline}\label{5.47a}
\prod_{l=1}^{2k} z_l^N \sum_{\kappa \subseteq (2k)^N}s_{\kappa'}(1/\mathbf z)  \delta_{\kappa = 2 \mu} { \langle C_\mu^{(2)}(G G^T) \rangle_{G \in {\rm GinOE}_N}  
\over  C_\mu^{(2)}((1)^N)} \\
=
\prod_{l=1}^{2k} z_l^N \sum_{\kappa \subseteq (2k)^N} s_{\kappa'}(1/\mathbf z)   \delta_{\kappa = 2 \mu}  2^{|\mu|} [N/2]_\mu^{(2)} 
= \prod_{l=1}^{2k} z_l^N \sum_{\mu \subseteq (N)^k} s_{\mu^2}(1/\mathbf z) 2^{|\mu|} [N/2]_{\mu'}^{(2)}, 
\end{multline}
where the first equality follows from the fact that the eigenvalues of $G G^T$ belong to the matrix ensemble
ME${}_{1,N}[\lambda^{-1/2} e^{- \lambda/2}]$ 
(recall the text below (\ref{pwX})) and use of (\ref{7.X1}) with $\alpha = 2$, $a=-1/2$, while the second equality results upon replace $\mu$ by $\mu'$ in the
summand.

The next step is to  make use of (\ref{5.46a}) with $N=k$, for the purpose of substituting for $s_{\mu^2}(1/\mathbf z)$. In this
$\langle C_\mu^{(1/2)}(X X^\dagger) \rangle_{X \in {\rm G}\mathcal A_{2k}} = 
C_\mu^{(1/2)}((1)^k)[2k-1]_\mu^{(1/2)}$ as follows from (\ref{7.X1}) and
 the fact that for $X \in {\rm G}\mathcal A_{2k}(\mathbb C)$, the eigenvalues of $X X^\dagger$
are doubly degenerate, with independent eigenvalues distributed according to
ME${}_{4,k}[e^{-\lambda}]$ \cite[Table 3.1]{Fo10}. Noting the formulas
$$
{s_{\mu^2}((1)^k) \over C_\mu^{(1/2)}((1)^k)[2k-1]_\mu^{(1/2)} }= {1 \over | \mu |!}, \qquad P_{\mu'}^{(2)}((1)^N) =
{ [N/2]_{\mu'}^{(2)}  \over h_\mu' |_{\alpha = 1/2}},
$$
(see  \cite[Eq.~(2.62)]{Se23} for the first, and \cite[Eq.~(13.8) together with $h_{\kappa'} = \alpha^{|\kappa|}
h_\kappa' |_{\alpha \mapsto 1/\alpha}$ as noted three lines below]{Fo10} for the second), and making use too of (\ref{15.ckp})
then shows that 
(\ref{5.47a}) is equal to
\begin{equation}\label{5.47b}
 \prod_{l=1}^{2k} z_l^N \sum_{\mu \subseteq (N)^k}  P_{\mu'}^{(2)}((1)^N)  \langle P_\mu^{(1/2)}(Z^{-1} X Z^{-1} X^\dagger) \rangle_{X \in {\rm G}\mathcal A_{2k}}. 
\end{equation}
The sum herein is carried out using the dual Cauchy identity (\ref{SM2}) with $\alpha = 1/2$, to obtain an expression analogous to (\ref{7.X2d}). To obtain the RHS of
(\ref{5.47}) from this requires the use of (\ref{7.X5}), and the fact that the eigenvalues of the matrix in the latter are doubly degenerate. 
\end{proof}

\begin{remark} ${}$ \\
1.~Let $A,B$ be $N \times N$ complex symmetric (anti-symmetric) matrices. Let $X \in \mathcal S_N(\mathbb C)$ ($\mathcal A_N(\mathbb C)$) be chosen from
a measure unchanged by $X \mapsto U X U^T$ for $U \in U(N)$. Then in addition to (\ref{5.46}) and (\ref{5.46a}) one has  \cite{FS09}
\begin{align}
\Big \langle C_\kappa^{(2)}(A X)  C_\mu^{(2)}(B X^\dagger )  \Big \rangle_{X \in \mathcal S_N(\mathbb C)}
& = \delta_{\kappa,\mu} {C_{\kappa}^{(2)}(AB) \over s_{2\kappa}((1)^N)} \langle C_\kappa^{(2)}(X X^\dagger) \rangle_{X\in \mathcal S_N(\mathbb C)} \nonumber \\
\Big \langle C_\kappa^{(1/2)}(A X)  C_\mu^{(1/2)}(B X^\dagger )  \Big \rangle_{X \in \mathcal A_N(\mathbb C)}
& =\delta_{\kappa,\mu}{C_{\kappa}^{(1/2)}(AB) \over s_{\kappa^2}((1)^N)} \langle C_\kappa^{(1/2)}(X X^\dagger) \rangle_{X\in \mathcal A_N(\mathbb C)}. 
\end{align}
\\
2.~A duality which reverses the role of ${\rm G} \mathcal S_{2k}$ and ${\rm GinSE}_N$ in (\ref{5.48}), and its companion reversing the roles
of the ensembles on either sides of (\ref{5.47}), but with the Gaussian complex anti-symmetric matrices replaced by Gaussian complex self dual matrices
(the latter are obtained from even size anti-symmetric matrices by premultiplying by $\mathbb I_N \otimes \begin{bmatrix} 0 & 1 \\  -1 & 0 \end{bmatrix}$),
has attracted recent interest in the works \cite{AAKP24}, \cite{KKR24}, \cite{Fo24+}.
\end{remark}

Documented in the literature is the analogue of Corollary \ref{C5.2} for the non-Hermitian ensembles TrOE${}_{N,K}$ and TrSE${}_{N,K}$. For this one requires
the spherical anti-symmetric (symmetric) ensemble S$\mathcal A_{N,K}$ (S$\mathcal S_{N,K}$), defined to have joint element PDF proportional to
\begin{equation}\label{SSE}
\det (\mathbb I_N + X X^\dagger)^{-N-K},
\end{equation}
where $X$ is anti-symmetric (symmetric). 

\begin{prop}\label{P5.8} (\cite{SS24}, \cite{Se23}) We have
\begin{equation}\label{5.48a}
\bigg \langle \prod_{j=1}^{2k} \det (G - z_j \mathbb I_N) \bigg \rangle_{G \in {\rm TrOE}_{N,K}} =
\bigg \langle \det \begin{bmatrix} X & Z \\ -Z & X^\dagger  \end{bmatrix}^{N/2} \bigg \rangle_{X \in {\rm S} \mathcal A_{2k,(N+K)/2-1}} 
\end{equation}
and
\begin{equation}\label{5.48b}
\bigg \langle \prod_{j=1}^{2k} \det (G - z_j \mathbb I_N) \bigg \rangle_{G \in {\rm TrSE}_{N,K}} =
\bigg \langle \det \begin{bmatrix} X & Z \\ -Z & X^\dagger  \end{bmatrix}^{N} \bigg \rangle_{X \in {\rm S} \mathcal S_{2k,N+K+1}}. 
\end{equation}
\end{prop}

\begin{proof}
The general strategy of the proof of Corollary \ref{C5.2} suffices, with changes of detail as implied by the use of
(\ref{W}) rather than (\ref{7.X1}).
\end{proof}

\begin{remark}
By taking the limit $z_{2k} \to \infty$  in the identities of Corollary \ref{C5.2} and Proposition \ref{P5.8}, they remain formally
unchanged except for the replacement of $2k$ by $2k-1$  (odd number of products). Intermediate working, such as the in-text
identity five lines below (\ref{5.47a}) require modification if one was to not distinguish the parity from the outset.
\end{remark}

\section{Dualities for moments}\label{S7}
Consider the class of Hermitian $\beta$ ensembles ME${}_{\beta,N}[e^{-N \beta V(x)/2}]$, where $V(x)$ is such that in the large
$N$ limit the support of the eigenvalue density is a compact, single interval. Classical examples of this setting are $V(x)$ equal to
\begin{equation}\label{9.1}
x^2, \qquad (x - \gamma \log x) \mathbbm 1_{x>0}, \qquad( \gamma_1 \log x + \gamma_2 \log (1 - x))  \mathbbm 1_{0<x<1},
\end{equation}
as discussed in the context of global scaling in \S 4.2.1 and \S 4.3.1.

With $m_\kappa(\mathbf x)$ denoting the monomial symmetric polynomial in the variables $\mathbf x$ indexed
by the partition $\kappa$, the mixed moments are specified by $m(\kappa;N,\beta) := \langle m_\kappa(\mathbf x) \rangle$. For example,
with $\kappa = (k,(0)^{N-1})$ one has that this average corresponds to the power sum of the eigenvalues
$\langle \sum_{l=1}^N x_l^k \rangle$, which in turn is equal to the $k$-th moment of the spectral density. When $\kappa$
consists of two or more nonzero parts, one defines from the mixed moments 
$\mu(\kappa;N,\beta)$ the corresponding mixed cumulants, defined
in the usual way according to the logarithm of the exponential generating function of the former. For the cumulants specified
by partitions of no more than $n$ parts, the Laurent series generating function is specified by
\begin{equation}\label{9.2}
{W}_n(y_1,\dots,y_n;N,\beta) := {1 \over y_1 \cdots y_n} \sum_{p_1,\dots,p_n = 0}^\infty {\mu((p_1,\dots,p_n);N,\beta) \over
y_1^{p_1} \cdots y_n^{p_n} };
\end{equation}
here after suitable ordering $(p_1,\dots,p_n)$ corresponds to the partition labelling the cumulant. Introducing the averages
\begin{equation}\label{9.3}
\Big \langle \prod_{l=1}^q ( A_l - \langle A_l \rangle) \Big \rangle, \quad A_i = \sum_{j=1}^N {1 \over y_i - x_j},
\end{equation}
the generating function (\ref{9.2}) can be written as a multinomial in such terms for $q=2,\dots,n$ which is homogeneous of degree 
$n$ in the total number of factors of the form $(A_l - \langle A_l \rangle )$. In fact for $n=2$ and 3, (\ref{9.2}) is precisely
equal to (\ref{9.3}) with $q=2$ and 3. We will write $\langle  \prod_{l=1}^q  A_l \rangle_{\rm c}$ to denote the linear combination of products
of averages relating to the cumulants (the symbol ``c'' denotes connected, or fully truncated; here we omit the $ \langle A_l \rangle$ 
in $(A -  \langle A_l \rangle)$ which
can be absorbed into the linear combination).

In addition to the quantity $W_n$ (\ref{9.2}), introduce the connected quantity
\begin{equation}\label{9.4}
P_{n+1}(y_1,\dots,y_n;y;\beta,N) = \bigg \langle  \prod_{l=1}^n  A_l  \sum_{l=1}^N {V'(y) - V'(y_l) \over y - y_l} \bigg \rangle_{\rm c}.
\end{equation}
Introduce too the notation $I = \{y_1,\dots,y_{n-1}\}$. Together with the $W_n$ (\ref{9.2}) these quantities satisfy the
coupled hierarchy of equations, referred to as loop equations,
\begin{multline}\label{9.4a}
(\beta/2) \sum_{J \subseteq I} W_{|J|+1}(y,J) W_{n-|J|}(y,I \backslash J) + (\beta/2) W_{n+1}(y,y,I) + (\beta/2 - 1)
{\partial W_n(y,I)  \over \partial y}\\
= {\beta N \over 2} \Big ( V'(y) W_n(y,I) - P_n(I;y) \Big ) -
\sum_{y_i \in I} {\partial \over \partial y_i} \Big (
{W_{n-1}(y,I\backslash \{y_i \}) - W_{n-1}(I) \over y - y_i} \Big );
\end{multline}
see \cite{BEMN10}, \cite{BMS11}, \cite[Eq.~(2.5)]{WF14}. For ease of presentation, the dependence of the $W_n$ and
$P_n$ on $N,\beta$ has been suppressed. The remaining step is to introduce the scaled quantity
\begin{equation}\label{9.5}
\widetilde{W}_n(y_1,\dots,y_n) := {1 \over N} (N\beta/2)^{n-1} {W}_n(y_1,\dots,y_n),
\end{equation}
which for large $N$ is expected to be independent of $N$ at leading order (for justifications, see \cite{BG11}), and furthermore independent of $\beta$. One similarly
introduces a scaling of $ P_n$ with an analogous large $N$ property. This gives the rewrite of (\ref{9.4a})
\begin{multline}\label{9.4a+}
(\beta N/2) \sum_{J \subseteq I} \widetilde{W}_{|J|+1}(y,J) \widetilde{W}_{n-|J|}(y,I \backslash J) + {1 \over N} \widetilde{W}_{n+1}(y,y,I) + (\beta/2 - 1)
{\partial \widetilde{W}_n(y,I)  \over \partial y}\\
= {\beta N \over 2} \Big ( V'(y) \widetilde{W}_n(y,I) - \widetilde{P}_n(I;y) \Big ) -
{\beta N \over 2} \sum_{y_i \in I} {\partial \over \partial y_i} \Big (
{\widetilde{W}_{n-1}(y,I\backslash \{y_i \}) -  \widetilde{W}_{n-1}(I) \over y - y_i} \Big ).
\end{multline}
The following duality relation is now evident.

\begin{prop}\label{P6.1} (\cite{EKR15}, \cite{Fo22})
The quantity $\widetilde{W}_n(y_1,\dots,y_n) = \widetilde{W}_n(y_1,\dots,y_n;N,\beta)$ is invariant under the mappings $\beta \mapsto 4/\beta$,
$N \mapsto -\beta N/2$.
\end{prop}

\begin{proof}
Upon the stated mappings, the loop equations (\ref{9.4a}) remain formally unchanged.
\end{proof}

\begin{remark} ${}$ \\
1.~One recalls that the meaning of the mapping in $N$ is with respect to a $1/N$ expansion. It follows from Proposition \ref{P6.1} that such
an expansion for $\beta = 2$ involves only powers of $1/N^2$. \\
2.~A viewpoint of the duality linking $\beta = 1$ and $\beta = 4$ as first identified in \cite{MW03} within the theory of Lie groups and symmetric
spaces is given in \cite{MV11}; see also \cite{Be19}.
\end{remark}

\begin{cor}\label{C6.1}
The scaled moments (independent of $N$ and $\beta$ to leading order in $N$)  $\widetilde{m}_k = \widetilde{m}_k(\beta,N)$, and smoothed density
$\widetilde{\rho}_{N,(1)}(x;\beta)$ satisfy the duality relations
\begin{equation}\label{9.9}
 \widetilde{m}_k(\beta,N) =  \widetilde{m}_k(4/\beta,-\beta N/2), \qquad  \widetilde{\rho}_{N,(1)}(x;\beta) = \widetilde{\rho}_{-\beta N/2,(1)}(x;4/\beta).
\end{equation} 
\end{cor}

The duality for the scaled moments is easy to illustrate. Thus, with $\tau := \beta/2$,
 for the Gaussian $\beta$ ensemble one has \cite{DE05}, \cite{MMPS12}, \cite{WF14}
\begin{align}\label{6.9}
 \widetilde{m}_2^{\rm G}(\beta,N) & = 1 + N^{-1} (-1 + \tau^{-1}) \nonumber \\
 \widetilde{m}_4^{\rm G}(\beta,N) & = 2 + 5 N^{-1} (-1 + \tau^{-1}) + N^{-2} (3 - 5 \tau^{-1} + 3 \tau^{-2}) \nonumber \\
  \widetilde{m}_6^{\rm G}(\beta,N) & = 5 + 22 N^{-1} (-1 + \tau^{-1}) + N^{-2} (32 - 54 \tau^{-1} + 32 \tau^{-2})   \nonumber  \\ & \quad +
N^{-3} (-15 + 32 \tau^{-1} - 32 \tau^{-2} + 15 \tau^{-3}),
\end{align}
and for the Laguerre $\beta$ ensemble \cite{MRW15}, \cite[Prop.~3.11]{FRW17},
\begin{align}
& \widetilde{m}_2^{\rm L}(\beta,N)  = (1+\gamma) + N^{-1} (-1 + \tau^{-1}) \nonumber \\
& \widetilde{m}_4^{\rm L}(\beta,N)  = (2 + 3 \gamma + \gamma^2) +  N^{-1}(4+3\gamma) (-1 + \tau^{-1}) + 2 N^{-2} (1 -  2 \tau^{-1} +  \tau^{-2} )  \nonumber \\
 &  \widetilde{m}_6^{\rm L}(\beta,N)  = (5 + 10 \gamma + 6 \gamma^2 + \gamma^3) + N^{-1}(16 + 21 \gamma + 6 \gamma^2)(-1 + \tau^{-1})   \nonumber \\
& \quad+ N^{-2} ((17 + 11 \gamma) \tau^{-2} - (33  + 21 \gamma )\tau^{-1} + (17  + 11 \gamma)) + N^{-3}(-1+\tau^{-1}) (-6 +11 \tau^{-1} - 6 \tau^{-2}).
\end{align}
An observation from \cite{WF14} is that the polynomials in $\tau^{-1}$ in the $1/N$ expansion (\ref{6.9}) have all their zeros on the unit
circle and also exhibit an interlacing property.
Explicit functional forms are also available for low order moments in the Jacobi $\beta$ ensemble \cite{MRW15}, \cite{FRW17}, which
exhibit a rational function, rather than a polynomial, structure.
A consequence of the density duality in (\ref{9.9}) is that the 5th order linear differential equations which can be
derived in the cases of the GOE and GSE are related to each other by the implied mapping \cite{WF14}. In \cite{RF21}
the 7th order linear differential equation satisfied by the soft edge scaled density for the Gaussian $\beta$ ensemble
with $\beta = 6$ is transformed to 7th order linear differential equation for the soft edge scaled density of the
Gaussian $\beta$ ensemble
with $\beta = 2/3$.

It is known how to map from the Jacobi $\beta$ ensemble to the (generalised) circular Jacobi $\beta$ ensemble
\cite[Eq.~(1.8)]{FS23}. Consequently, there is an analogue of Proposition \ref{P6.1} for the  
 latter ensemble; see \cite[Remark 4.3]{Fo10}. In relation to a consequence,
denote by $\rho_{(2),\infty}((x,0);\beta)$ the bulk scaled (unit density) two-point correlation for the circular $\beta$ ensemble. Define the
structure function as the Fourier transform
\begin{equation}\label{9.10}
S(k;\beta) := \int_{-\infty}^\infty \Big ( \rho_{(2),\infty}^T((x,0);\beta) + \delta(x) \Big ) e^{ikx} \, dx,
\end{equation}
where $ \rho_{(2),\infty}^T((x,0);\beta) :=  \rho_{(2),\infty}((x,0);\beta) - 1$. It is known that  $f(k;\beta) := \pi \beta S(k;\beta)/|k|$
extends to an analytic function in the range $|k| < \min(2\pi, \pi \beta)$. Most significantly from the viewpoint of the present theme is that
 $f$ exhibits the particular duality \cite{FJM00}
\begin{equation}\label{9.11}
f(k;\beta) = f \Big ( - {2k \over \beta};{4 \over \beta} \Big ),
\end{equation}
which as discussed in \cite{FS23} is in keeping with the more general (but finite $N$) duality exhibited by  the circular Jacobi $\beta$ ensemble.
We remark too that as is consistent with (\ref{9.11}), 
 the small $k$ expansion of (\ref{9.10}) has for the coefficient of $k^p$ polynomials in $u=2/\beta$ of
degree $p$, which are palindromic or anti-palindromic \cite{FJM00, Fo21}; see also \cite{FS23} for observations relating to the zeros.

An important detail with regards to the results of Proposition \ref{P6.1} and Corollary \ref{C6.1} is the factor $N \beta$ in the weight function
$e^{- N \beta V(x)/2}$. If this factor is removed, and one considers instead the matrix ensemble ${\rm ME}_{\beta,N}[e^{-\tilde{V}(x)}]$ for
some $\tilde{V}$ independent of $N, \beta$, a distinct statistical state can be obtained by taking the scaled high
temperature limit $\beta^{-1} \to \infty$ with $N\beta/2 = \alpha$ fixed. As noted in \cite{Fo22}, in the classical cases of the
Gaussian, Laguerre and Jacobi weights, a duality relation for the low temperature, $\beta \to 0$ with $N$ fixed, limit applies.

Consider for definiteness the Gaussian case with $\tilde{V}(x) = x^2/2$ \cite{ABG12}. Let the moments of the resulting normalised density be denoted
$\{ m_{2k}^{{\rm G}^*}(\alpha) \}_{k=1,2,\dots}$. With $\{ \tilde{m}_{2k}^{\rm G}(\beta,N) \}$ as in Corollary \ref{C6.1}, from the definitions
$m_{2k}^{{\rm G}^*}(\alpha) = \lim_{N \to \infty, N \beta/2 = \alpha} (N/\beta/2)^k  \tilde{m}_{2k}^{\rm G}(\beta,N)$ and we read off from
(\ref{6.9}) that
\begin{align}\label{6.9a}
 \widetilde{m}_2^{\rm G^*}(\alpha) & = \alpha + 1 \nonumber \\
 \widetilde{m}_4^{\rm G^*}(\alpha) & = 2 \alpha^2 + 5 \alpha + 3 \nonumber \\
  \widetilde{m}_6^{\rm G^*}(\alpha) & = 5 \alpha^3 + 22 \alpha^2 + 32 \alpha + 15.
\end{align}
Insert now a factor $\beta$ in front of $\tilde{V}$ and consider ME${}_{\beta,N}[e^{-\beta \tilde{V}(x)}]$.
Denoting the corresponding moments by $\{ m_{2k}^{{\rm G}}(\beta,N) \}$, we have $m_{2k}^{{\rm G}}(\beta,N)  = N^{k+1} 
  \tilde{m}_{2k}^{\rm G}(\beta,N)$. We see from (\ref{6.9}) that these moments have a well defined expansion in $1/\beta$
 with leading terms
 \begin{align}\label{6.9b}
\lim_{\beta \to \infty} {m}_2^{\rm G}(\beta,N) & = N^2 - N \nonumber \\
\lim_{\beta \to \infty} {m}_4^{\rm G}(\beta,N) & = 2 N^3 - 5 N^2 + 3N \nonumber \\
\lim_{\beta \to \infty} {m}_6^{\rm G}(\beta,N) & = 5 N^4 - 22 N^3 + 32 N^2 - 15 N,
\end{align}
thus providing evidence for a  duality formula linking  low temperature moments
to the scaled high temperature
moments 
\begin{equation}\label{6.9c}
\lim_{\beta \to \infty} m_{2k}^{{\rm G}}(\beta,N) = (-1)^k N m_{2k}^{{\rm G}^*}(\alpha) \Big |_{\alpha = - N}.
\end{equation}
This has been established in  \cite{Fo22}
using a loop equation analysis.

  \subsection*{Acknowledgements}
  Feedback on an earlier draft from Sungsoo Byun and Bojian Shen is
  appreciated.
	This work  was supported
	by the Australian Research Council 
	 Discovery Project grant DP250102552.

\small
\providecommand{\bysame}{\leavevmode\hbox to3em{\hrulefill}\thinspace}
\providecommand{\MR}{\relax\ifhmode\unskip\space\fi MR }
\providecommand{\MRhref}[2]{%
  \href{http://www.ams.org/mathscinet-getitem?mr=#1}{#2}
}
\providecommand{\href}[2]{#2}

  \end{document}